\documentclass[12pt]{article}
\usepackage[utf8]{inputenc}
\usepackage[a4paper, total={6.5in, 9.5in}]{geometry}
\usepackage{enumitem}
\usepackage[reqno]{amsmath}
\counterwithin{equation}{section}
\usepackage{exscale} 
\usepackage[colorlinks=true, pdfstartview=FitV, linkcolor=blue, citecolor=red, urlcolor=black,pagebackref=false]{hyperref}
\usepackage[
    type={CC},
    modifier={by},
    version={4.0},
]{doclicense}
\usepackage{fontenc}
\usepackage{makeidx}
\usepackage{amsfonts}
\usepackage{amstext}
\usepackage{amsrefs}
\usepackage{amsthm}
\usepackage{amssymb}
\usepackage{bbm}
\usepackage[utf8]{inputenc}
\usepackage[dvipsnames]{xcolor}
\usepackage{tikz} 
\usetikzlibrary{decorations.pathreplacing,calligraphy,fadings,patterns,arrows.meta,bending }
\usepackage{appendix}
\usepackage{booktabs}
\usepackage{float}
\usepackage{esint}
\usepackage{authblk}
\usepackage{fancyhdr}
\usepackage{mathrsfs}
\usepackage{array}
\usepackage{multirow}
\usepackage{pstricks-add}
\usepackage{float}
\usepackage{colortbl}

\newcommand{\R}{\mathbb{R}}
\newcommand{\N}{\mathbb{N}}

\newcommand{\jnabla}{{\langle\nabla\rangle}}
\newcommand{\dXY}{{\mathrm{dist}(X,Y)}}

\theoremstyle{plain}
\newtheorem{theo}{Theorem}[section]
\newtheorem{prop}[theo]{Proposition}
\newtheorem{lemma}[theo]{Lemma}
\newtheorem{cor}[theo]{Corollary}

\newtheorem{rem}[theo]{Remark}

\newtheorem{definition}[theo]{Definition}

\definecolor{mylightgreen}{rgb}{0.816, 0.91, 0.789}
\newcommand{\gcell}{\cellcolor{mylightgreen}}
\definecolor{mylightred}{rgb}{0.949, 0.781, 0.648}
\newcommand{\rcell}{\cellcolor{mylightred}}
\definecolor{mylightblue}{rgb}{0.602, 0.75, 0.918}
\newcommand{\bcell}{\cellcolor{mylightblue}}

\newcommand{\meas}{\mathcal{M}}
\newcommand{\diff}{\mathrm{d}}
\newcommand{\setW}{{\mathcal{W}_{d,s}}}
\newcommand{\Tepsilon}{{G_\varepsilon}}
\newcommand{\Tzero}{{G_0}}
\let\Im\relax
\DeclareMathOperator{\Im}{Im}

\allowdisplaybreaks

\title{Propagation Estimates for the Boson Star Equation}
\date{}
\author{S\'ebastien Breteaux, J\'er\'emy Faupin, Viviana Grasselli}
\affil{Université de Lorraine, CNRS, IECL, F-57000 Metz, France}

\begin{document}
\maketitle
	\begin{abstract}
		We consider the boson star equation with a general two-body interaction potential $w$ and initial data $\psi_0$ in a Sobolev space. Under general assumptions on $w$, namely that $w$ decomposes as a sum of a finite, signed measure and an essentially bounded function, we prove that the (local in time) solution cannot propagate faster than the speed of light, up to a sharp exponentially small remainder term. If $w$ is short-range and $\psi_0$ is regular and small enough, we prove in addition  asymptotic phase-space propagation estimates and minimal velocity estimates for the (global in time) solution, depending on the momentum of the scattering state associated to $\psi_0$.
	\end{abstract}

\tableofcontents

	\section{Introduction}

In this paper we study the pseudo-relativistic (or semi-relativistic) Hartree equation
\begin{eqnarray}
	\label{eq:hartree_semi_intro}
	\begin{cases}
i\partial_t \psi = (\langle\nabla\rangle + w*|\psi|^2)\psi, \quad  x \in \R^d, \, \quad t\ge0 , \\
\psi|_{t=0} = \psi_0,
	\end{cases}
\end{eqnarray}
in any dimension $d\ge3$. Here and throughout the paper, we use the notation $\langle\nabla\rangle$ for the Fourier multiplier
\begin{equation*}
\langle\nabla\rangle:=\sqrt{1-\Delta}=\mathcal{F}^{-1}\sqrt{1+|\xi|^2}\mathcal{F},
\end{equation*}
where $\mathcal{F}$ is the Fourier transform normalized such that $\mathcal{F}$ is unitary on $L^2$.

Eq. \eqref{eq:hartree_semi_intro} describes the effective dynamics,  in the mean-field limit, of an $N$-body quantum system of pseudo-relativistic bosons of mass $m=1$, with two-body gravitational interaction given by the convolution potential $w$; see \cite{ElgartSchlein}. It is thus used as a model for a pseudo-relativistic boson star. Here we work in units such that Planck's constant divided by $2\pi$ and the velocity of light are equal to $1$. We also assume that the mass of the bosons is $1$ for simplicity, but we could consider the kinetic energy $\sqrt{m^2-\Delta}$ (or $\sqrt{m^2-\Delta}-m$), for any $m>0$, without changing our results.

A physically relevant choice of the interaction potential $w$ is the attractive Newtonian potential, $w(x)=-\kappa|x|^{-1}$ for some $\kappa>0$, but we will consider in this paper general potentials in suitable spaces, imposing different conditions depending on the results. The assumption $w\ge0$ corresponds to repulsive, or defocusing interactions, while $w\le0$ corresponds to attractive, or focusing ones.

We use the shorthand $\psi_t$ for a solution $(t,x)\mapsto\psi(t,x)$ to \eqref{eq:hartree_semi_intro}. The boson star equation \eqref{eq:hartree_semi_intro} exhibits (formally at this stage) three important conserved quantities: the mass,  
\begin{equation}\label{eq:def_mass_intro}
    M(\psi_t):=\int_{\mathbb{R}^d}|\psi_t|^2;
\end{equation}
the energy, defined by
\begin{equation}\label{eq:def_energy_intro}
    E(\psi_t):=\frac12\int_{\mathbb{R}^d}\big|\langle\nabla\rangle^{\frac12}\psi_t\big|^2+\frac14\int_{\mathbb{R}^d}(w*|\psi_t|^2)|\psi_t|^2
\end{equation}
and the momentum,
\begin{equation}\label{eq:def_momentum_intro}
    P(\psi_t):=-\frac{i}{2}\int_{\mathbb{R}^d}\bar\psi_t\nabla\psi_t.
\end{equation}
We consider in this paper a translation invariant boson star equation, but several of our results might be extended to a boson star placed in an external potential $V(x)$, provided that suitable regularity and decay assumptions on $V$ are imposed. In this case, of course, the momentum is not a conserved quantity anymore.

For the attractive Newtonian potential $w(x)=-\kappa|x|^{-1}$ with $\kappa>0$, the convergence, in the mean-field limit, of the ground state energy of the pseudo-relativistic $N$-body Schrödinger Hamiltonian to the ground state energy of the static, pseudo-relativistic Hartree equation corresponding to \eqref{eq:hartree_semi_intro}, was established by Lieb and Yau in their seminal paper \cite{LiebYau1}. The fermionic case is also considered in \cite{LiebYau1}, in relation with the Chandrasekhar theory of stellar collapse. As mentioned before, the convergence in the mean-field limit of the time-dependent pseudo-relativistic Schrödinger equation to the pseudo-relativistic Hartree equation \eqref{eq:hartree_semi_intro}  was proven in \cite{ElgartSchlein}. See also \cites{Knowles-Pizzo,Lee} for explicit rates of convergence in the mean-field limit.

In \cite{lenzmann_semirel}, local and global well-posedness are studied in dimension $d=3$ for \eqref{eq:hartree_semi_intro} with a general external potential $V$ and a Yukawa-type interaction potential, $w(x)=\kappa|x|^{-1}e^{-\mu|x|}$ with $\kappa\in\mathbb{R}$, $\mu\ge0$, and for initial data $\psi_0$ in the Sobolev space $H^s$, with $s\ge\frac12$ (for attractive interaction potentials, a smallness condition must be imposed, either on $\psi_0$ or on $w$). The local well-posedness is extended in \cite{HerrLenzmann} to low-regularity initial states, namely $\psi_0\in H^s$ with $s>\frac14$. Here it should be noticed that the energy space corresponds to the Sobolev regularity $H^{1/2}$. In \cites{ChoOzawa1,ChoOzawa2,ChoOzawaSazaki}, global existence -- and small initial data scattering -- are proven for (sums of) short-range interaction potentials of the form $w_i(x)=\kappa_i|x|^{-\alpha_i}$, for some suitable values of $\alpha_i$; see also \cite{HerrTesfahun} for related results in the case of Yukawa-type potentials $w(x)=\kappa|x|^{-1}e^{-\mu|x|}$ with $\mu>0$.   We will compare some of these results to ours  below, after the statement of our main contributions. For long-range potentials $w(x)=\kappa|x|^{-1}$, ``modified scattering states'' must be introduced; see \cite{Pusateri} for global existence and small initial data scattering in this case.

In relation with gravitational collapse, finite time blow-up for \eqref{eq:hartree_semi_intro} is proven in \cite{frohlich_lenzm07} for attractive Newtonian potentials $w(x)=-\kappa|x|^{-1}$ with $\kappa>0$ and the mass $M(\psi_0)=\|\psi_0\|_{L^2}^2$ of the initial state larger than some critical value. The blow-up phenomenon is analysed in the mean-field limit in \cite{MichelangeliSchlein}. For the existence of solitary waves and stability results around them, under suitable assumptions, we refer to \cites{FrohlichJonssonLenz1, FrohlichJonssonLenz2}. We will not study finite time blow-up nor the existence of solitary waves here.

Our main concern is the speed of propagation of boson stars whose dynamics is given by \eqref{eq:hartree_semi_intro}. We generally consider a large class of interaction potentials,
\begin{equation}\label{eq:w-general}
    w\in\meas +L^\infty,
\end{equation}
where $\meas $ stands for the Banach space of finite, signed Radon measures on $\mathbb{R}^d$. In particular, we can take $w$ as a Dirac delta measure and hence consider a pure-power cubic non-linearity. Depending on the results, we will restrict the class of admissible potentials.

We will begin with establishing local and global existence of solutions to \eqref{eq:hartree_semi_intro}. Our proof follows the usual strategy of applying a fixed point argument to solve Duhamel's equation associated to \eqref{eq:hartree_semi_intro} in a suitable function space. For the global existence, we will distinguish two distinct regimes: the long-range regime, $w\in L^{d/2+1,\infty}+L^\infty$, and the short-range regime, $w\in\meas +L^q$, $1\le q<\frac{2d}{3}$, with possibly a smallness condition involving the initial state $\psi_0$. Compared to the extensively studied non-linear Schrödinger or Hartree equations (see e.g.  \cites{cazenaveWeissler_H1,cazenaveWeissler_Hs,Tao,GinibreVelo} and references therein), a difficulty, as in previous works, comes from the loss of derivatives in the dispersive estimates associated to the semi-relativistic kinetic energy, which, for the $L^1\to L^\infty$ estimate, may be written as
\begin{equation}\label{eq:dispersive_est_intro}
    \big\|e^{-it\langle\nabla\rangle}f\big\|_{L^\infty}\lesssim |t|^{-\frac{d}{2}}\|\langle\nabla\rangle^{\frac{d}{2}+1} f\|_{L^1}, 
\end{equation}
see Appendix \ref{app: pointw dec}.

Given a solution $\psi_t$ to \eqref{eq:hartree_semi_intro}, we will aim at estimating the speed of propagation of the boson star, in the sense of proving time-decay estimates for the probability that the velocity~$\frac{x}{t}$ of the state $\psi_t$ belongs to a certain domain, possibly depending on the position and momentum of the initial state $\psi_0$. We will establish both a maximal and a minimal velocity estimate, expressed under different forms. Our maximal velocity bound is of the form
\begin{equation}\label{eq:max-vel-intro}
    \|\mathbf{1}_Y\psi_t\|_{L^2}\le e^{t-\mathrm{dist}(X,Y)}\|\mathbf{1}_X\psi_0\|_{L^2},
\end{equation}
for any convex subsets $X,Y\subset\mathbb{R}^d$, where $\mathrm{dist}(X,Y)$ stands for the distance from $X$ to~$Y$. It holds for any local in time solution $\psi_t$ to \eqref{eq:hartree_semi_intro}, for any interaction potential $w\in\meas +L^\infty$, provided that the initial state $\psi_0$ is regular enough, depending on $w$. For times $t\ll \mathrm{dist}(X,Y)$, \eqref{eq:max-vel-intro} shows that the probability for the star to travel from $X$ to $Y$ is exponentially small, hence justifying that the maximal velocity of propagation for \eqref{eq:hartree_semi_intro} is equal to the velocity of light ($1$ in our units). A feature of \eqref{eq:max-vel-intro} is that it gives a sharp exponentially small error  (in a sense that will be made precise below). Our proof is adapted from \cite{SigalWu}. The maximal velocity estimate \eqref{eq:max-vel-intro} can be extended to general disjoint subsets $X,Y$, dropping the convexity assumption, at the cost of losing optimality. See below for more details. 

Our (asymptotic) minimal velocity estimate takes the form
 \begin{equation}\label{eq:min-vel-intro}
\Big\| \mathbf{1}_{[0,\alpha)}\Big(\frac{x^2}{t^2}\Big) \psi_t\Big\|_{L^2} \to 0, \qquad t\to\infty,
\end{equation}
for any initial state $\psi_0$ associated to a scattering state with an ``asymptotic instantaneous velocity'' larger than $\alpha$. The precise definition of the instantaneous velocity operator will be given below. In order to construct initial states with localized asymptotic instantaneous velocity, we will need first to establish ``phase-space'' propagation estimates and to prove the existence, and right-invertibility, of wave operators on a suitable set of small and regular initial data. For this we will have to restrict the admissible class of potentials to short-range interaction potentials $w\in \meas +L^q$ with $q<\frac{d}{2}$.

Propagation estimates, including maximal velocity, minimal velocity and phase-space propagation estimates, played a crucial role in the eighties and nineties in the scattering theory of $N$-body quantum systems, especially in the proof of asymptotic completeness of the wave operators, see, among others, \cites{Enss83,SigalSoffer88,SigalSoffer90,Gerard92,Derezinski93,Skibsted91,Graf90,HunzikerSigalSoffer99}. Propagation estimates were later extended to the framework of non-relativistic QED in e.g. \cites{DerezinskiGerard99,FrohlichGriesemerSchlein02,BonyFaupinSigal12,FaupinSigal14}. In the recent years, proving bounds on the maximal speed of propagation for quantum information in various physical contexts has been the subject of many works. We refer to \cite{ArbunichPusateriSigalSoffer21} and \cite{SigalWu} for the development of general methods for proving maximal velocity estimates for quantum systems, the former by the means of adiabatic space-time localization observables (ASTLO), while the latter is based on analyticity properties leading to exponential bounds.
In relation with Lieb-Robinson bounds, maximal velocity of quantum transport for Bose-Hubbard type Hamiltonians has been established in \cites{FaupinLemmSigal22_LR, FaupinLemmSigal22_maxv,LemmRubilianiZhang23,FaupinLemmSigalZhang25, LemmRubilianiZhang25}. For related works in the case of open quantum systems described by a Lindblad master equation, we refer to \cites{BreteauxFaupinLemmSigal22,BreteauxFaupinLemmSigalYangZhang24,SigalWu25}. Finally a maximal velocity estimate for the (non-relativistic) Hartree equation, using the ASTLO method, has been derived in \cite{ArbunichFaupinPusateriSigal23}.

Our paper is organized as follows. In the next section, we describe our main results in precise terms. Section \ref{sec:prelim} contains preliminary technical estimates that are subsequently used in the proof of our main results. In Section \ref{sec:local}, we prove local existence for \eqref{eq:hartree_semi_intro}, for the general class of potentials $w$ satisfying \eqref{eq:w-general}.
Sections \ref{sec:global-long} and \ref{sec:global-short} are devoted to the proof of global existence, in the long-range and short-range regimes, respectively. In Section \ref{sec:max-vel}, we prove the maximal velocity estimate \eqref{eq:max-vel-intro} and in Section \ref{sec:min-vel} we study the scattering theory for short-range potentials and establish the minimal velocity bound \eqref{eq:min-vel-intro}. Appendix \ref{app: commute} recalls estimates on commutators between weights and fractional derivatives in $L^p$ spaces. In Appendix \ref{app: pointw dec} we derive various time-decay estimates for the linear flow $e^{-it\langle\nabla\rangle}$ that are important ingredients in the proofs of our main results.

\section{Main results}

In this section we state and comment our main results on the existence and properties of solutions to \eqref{eq:hartree_semi_intro}.   In what follows $L^p$ and $L^{p,\infty}$ stand for the usual Lebesgue and weak Lebesgue spaces over $\mathbb{R}^d$, respectively; we also denote by $\meas $ the space of signed, finite Radon measures on $\mathbb{R}^d$, equipped with the total variation norm $\|\mu\|_{\meas }=|\mu|(\mathbb{R}^d)$.
We introduce the following class of interaction potentials
\begin{equation}\label{eq:def_Wds}
\setW=\begin{cases}
    L^{\frac{d}{2s},\infty}+L^\infty &\text{ if } s<\frac{d}{2},\\
    \bigcup_{q>1}(L^{q,\infty}+L^\infty) &\text{ if } s=\frac{d}{2},\\
    \meas+L^\infty & \text{ if } s> \frac{d}{2},    
\end{cases}
\end{equation}
where, in the sequel, $s$ will correspond to the regularity of the initial data $\psi_0$ in \eqref{eq:hartree_semi_intro}.

\begin{rem}$ $
\begin{enumerate}
    \item As mentioned in the introduction, for $s>\frac{d}2$, $w$ can be chosen as a sum of a Dirac delta measure and a function in any $L^p$ space, $1\le p\le\infty$, leading in \eqref{eq:hartree_semi_intro} to a potential which is a sum of a cubic non-linearity and a convolution non-linearity by a function in a broad class.
	\item The class $L^{\frac{d}{2s}, \infty} + L^\infty$ includes the interactions given by (sums of) convolution potentials of the form $w(x) = \frac{1}{|x|^\alpha}$ for $\alpha \in [0, 2s]$. Indeed, $\frac{1}{|x|^{2s}} \in L^{\frac{d}{2s}, \infty}(\R^d)$, while for $\alpha \in [0,2s)$ we have 
	\begin{equation*}
		\frac{1}{|x|^\alpha} = \frac{\mathbbm{1}_{|x|\leq 1}}{|x|^\alpha}  +  \frac{\mathbbm{1}_{|x|> 1}}{|x|^\alpha} \leq \frac{1}{|x|^{2s}} + \frac{\mathbbm{1}_{|x|> 1}}{|x|^\alpha}  \in L^{\frac{d}{2s}, \infty}(\R^d) + L^\infty(\R^d). 
	\end{equation*}
	\item For $s=\frac{d}{2}$, as usual, we cannot reach the endpoint case $L^{1,\infty}$ because of the absence of embedding of $H^s$ into $L^\infty$. Note the equality $\bigcup_{q>1}(L^{q,\infty}+L^\infty)=\bigcup_{q>1}(L^{q}+L^\infty)$. Here we do not try to refine further the functional space considered in this critical case.
\end{enumerate}
\end{rem}

To simplify the exposition, we only consider non-negative times to state (and prove) our results, even if they clearly also hold for negative times. Similarly, all our results hold without change if one replaces $\langle\nabla\rangle$ by $-\langle\nabla\rangle$ in \eqref{eq:hartree_semi_intro}.

\subsection{Local existence}

We begin with the local existence of solutions to \eqref{eq:hartree_semi_intro} for initial data in $H^s$, $s\ge0$. 

\begin{theo}[Local existence I]\label{th:local-long-range-intro}
	 Let $s\ge 0$, $\psi_0$ in $H^{s}$ and $w$ in $\setW$. Then there exists $T_{\max}$ in $(0,\infty]$ such that  \eqref{eq:hartree_semi_intro} admits a unique solution 
	\begin{equation*}
		\psi \in \mathcal{C}^0([0, T_{\max}), H^{s}) \cap \mathcal{C}^1 ([0, T_{\max}), H^{s-1}),
	\end{equation*}
	where either $T_{\max}=\infty$ or $\lim_{t\to T_{\max}} \|\psi_t\|_{H^s} =\infty$. Moreover, for any $T$ in $[0, T_{\max})$, the map
	\begin{equation*}
	    H^s\ni\psi_0\mapsto\psi|_{[0,T]}\in\mathcal{C}^0([0,T],H^s)\cap \mathcal{C}^1 ([0, T], H^{s-1})
	\end{equation*}
	is continuous.
\end{theo}

The proof of Theorem \ref{th:local-long-range-intro} relies on a standard fixed point argument in $H^s$ to solve Duhamel's equation associated to \eqref{eq:hartree_semi_intro}, using multilinear estimates proven in Section \ref{sec:prelim}.
We also remark that the mass $M(\psi_t)=\|\psi_t\|_{L^2}^2$ is a conserved quantity for  \eqref{eq:hartree_semi_intro} (see Lemma \ref{lemma: en conservation} for a precise statement), therefore the previous theorem with $s=0$ implies global existence for $w\in L^\infty$ and $\psi_0\in L^2$. 

For initial states with limited regularity, the class of admissible potentials $w$ for local existence can be extended    thanks to the Strichartz  estimates given in Appendix \ref{app: pointw dec}. In the next statement $L^{b,2}$ stands for the usual Lorentz space, see Subsection \ref{subsec:Lorentz-ineq} for the definition.

	\begin{theo}[Local existence II]\label{th: local_with_Strichartz-intro}
	    Let $0\le s<(d+1)/(2d-2)$,  $w$ in $L^{(d+1)/(4s),\infty}+L^\infty$ and~$\psi_{0}$ in~$H^{s}$. There exists~$T_{\max}$ in~$(0,\infty]$ such that \eqref{eq:hartree_semi_intro} admits a unique solution 
\begin{equation}\label{eq:space-Strichartz-intro}
\psi\in\mathcal{C}^{0}([0,T_{\max}),H^{s})\cap \mathcal{C}^{1}([0,T_{\max}),H^{s-1})\cap L^a_{\mathrm{loc}}([0,T_{\max}),L^{b,2}),
\end{equation}
where $\frac1a=\frac{(d-1)s}{d+1}$ and $\frac1b=\frac12-\frac{2s}{d+1}$.

If $d\ge4$, $s\ge (d+1)/(2d-2)$, $w\in  L^{(d-1)/2,\infty}+L^\infty$ and~$\psi_{0}\in H^{s}$, then \eqref{eq:hartree_semi_intro} admits a unique solution satisfying \eqref{eq:space-Strichartz-intro}, with $\frac1a=\frac12$ and $\frac1b=\frac12-\frac1{d-1}$.

If $d=3$, $s\ge1$, $w\in L^{q,\infty}+L^\infty$ with $q>1$ and $\psi_0\in H^s$, then \eqref{eq:hartree_semi_intro} admits a unique solution satisfying \eqref{eq:space-Strichartz-intro} with $\frac1a=\frac1{2q}$ and $\frac1b=\frac12-\frac1{2q}$.
	\end{theo}
Similarly as in Theorem \ref{th:local-long-range-intro}, both a blow-up alternative and a continuity property of the solution with respect to the initial data hold under the conditions of Theorem \ref{th: local_with_Strichartz-intro}. See Theorem \ref{th: local_with_Strichartz} below for a more precise and more general statement involving any admissible pair $(a,b)$ for the Strichartz estimates; as observed in \cite{GinibreVelo_Timedecay}, a useful property of \eqref{eq:hartree_semi_intro} is that Strichartz estimates hold for both wave and Schrödinger admissible pairs. In Appendix \ref{app: pointw dec}, we justify that Strichartz estimates hold in Lorentz spaces, in order to be able to consider $w$ in the weak space $L^{(d-1)/2,\infty}$ for $d\ge4$.

As shown in \cite{HerrLenzmann}, the regularity condition on the initial states can be improved, at least for the Coulomb potential. Indeed, local existence for initial data $\psi_0\in H^s$ for $s>\frac14$ in dimension $3$ is proven in \cite{HerrLenzmann}, in the case where $w(x)=-|x|^{-1}$ is the attractive Coulomb potential (which belongs to $L^{3,\infty}$), while Theorem \ref{th: local_with_Strichartz-intro} requires $s\ge\frac13$. In this paper we do not try to optimize the regularity condition of initial states to ensure local existence for a given potential.

\subsection{Maximal velocity estimates}

Our first main result provides a maximal velocity estimate for the boson star between two convex subsets of $\mathbb{R}^d$. It holds for any local in time solutions to \eqref{eq:hartree_semi_intro}.

Recall that $\mathrm{dist}(X,Y)$ stands for the distance between two subsets of $\mathbb{R}^d$.

\begin{theo}[Sharp maximal velocity estimate for convex subsets]\label{th:max-vel-intro}
    Let $X,Y\subset\mathbb{R}^d$ be two convex subsets, $s\geq 0$, $w$ in $\setW$ and $\psi_0$ in $H^{s}$ such that~$\mathbf{1}_X\psi_0=\psi_0$.
    
    If $\psi$ is the local solution to \eqref{eq:hartree_semi_intro} on $[0,T_{\mathrm{max}})$ given by Theorem~\ref{th:local-long-range-intro} or Theorem~\ref{th: local_with_Strichartz-intro}, then
    \begin{equation}\label{eq:sharp_MVE}
       \forall t\in [0,T_{\mathrm{max}})\,,\qquad \|\mathbf{1}_Y\psi_t\|_{L^2}\le e^{t-\dXY}\|\psi_0\|_{L^2}.
    \end{equation}
\end{theo}

\begin{rem}$ $
\begin{enumerate}
\item    The maximal velocity bound \eqref{eq:sharp_MVE} also holds if one replaces the conditions on $w$ and $\psi_0$ ensuring the local existence of solutions to \eqref{eq:hartree_semi_intro} given by Theorem \ref{th: local_with_Strichartz-intro}, by the slightly more general conditions ensuring local existence given by Theorem \ref{th: local_with_Strichartz}.
\item
We recall that the speed of light is equal to $1$ in our units. For the relativistic dispersion relation $\sqrt{-c^2\Delta+m^2c^4}-mc^2$, a simple scaling argument gives, instead of \eqref{eq:sharp_MVE}, the maximal velocity bound
    \begin{equation*}
        \|\mathbf{1}_Y\psi_t\|_{L^2}\le e^{mc^2(ct-\dXY)}\|\psi_0\|_{L^2}.
    \end{equation*}
\item
The exponentially small error term in the maximal velocity estimate \eqref{eq:sharp_MVE} is ``sharp'' in the sense that:
\begin{enumerate}
    \item If $C<1$, there exist convex subsets $X$ and $Y$ such that the estimate 
\begin{equation*}
        \forall t\in [0,T_{\mathrm{max}})\,,\quad \|\mathbf{1}_Y \psi_t\|_{L^2} \le C e^{t-\dXY}\|\psi_0\|_{L^2},
\end{equation*}
does \emph{not} hold. This is obvious since, if $X=Y$ then $\|\mathbf{1}_Y \psi_t \|_{L^2}= \|\psi_0\|_{L^2}$ at $t=0$ for any $\psi_0$ as in the statement of Theorem \ref{th:max-vel-intro}.
\item If $c<1$, there exist convex subsets $X$, $Y$ and $w$, $\psi_0$ as in the statement of Theorem \ref{th:max-vel-intro} such that the estimate 
\begin{equation*}
        \forall t\in [0,T_{\mathrm{max}})\,,\quad \|\mathbf{1}_Y \psi_t\|_{L^2} \le e^{ct-\dXY} \|\psi_0\|_{L^2},
\end{equation*}
does \emph{not} hold. This statement is proven in the case of the free evolution, $w=0$, in our companion paper \cite{BFG2a}*{Corollary B.2}.
\end{enumerate} 
\end{enumerate}

\end{rem}

Our proof of Theorem \ref{th:max-vel-intro} is based on an elegant and powerful argument recently introduced in \cite{SigalWu} for linear Schrödinger-type equations of the form $i\partial_t\psi=\omega(-i\nabla)\psi_t+V(x)\psi_t$, with $\omega(\xi)$ admitting an analytic continuation to a bounded region of $\mathbb{C}^d$. Substantial modifications are however necessary to accommodate the non-linear dynamics considered here, and to reach the sharp bound stated in \eqref{eq:sharp_MVE} (note that the bound obtained in \cite{SigalWu}, for more general dispersion relations and general open sets $X,Y\in\R^d$, is of the form $\|\mathbf{1}_Y\psi_t\|_{L^2}\le C_{\mu,c}e^{\mu(ct-\dXY)}\|\psi_0\|_{L^2}$ for some $C_{\mu,c}>0$ and any $\mu<1$ and $c>1$, without explicit control on $C_{\mu,c}$).

The idea of the proof of Theorem \ref{th:max-vel-intro} is to estimate, for a suitably chosen function $\ell$:
\begin{align*}
\|\mathbf{1}_Y \psi_t\|_{L^2}
    \le \underset{\le\, \exp\big(-\frac{\dXY}{2}\big)}{\underbrace{\|\mathbf{1}_Ye^{\ell(x)}\|_{\mathcal{B}(L^2)}}} \quad
    \underset{\le\,  \exp(t)\exp\big(-\frac{\dXY}{2}\big)\|\psi_0\|_{L^2}}{\underbrace{\|e^{-\ell(x)}\psi_t\|_{L^2}}}
   \,,
\end{align*}
The function $\ell$ is constructed thanks to the convexity of $X$ and $Y$, through a separation argument; see \cite{FaupinLemmSigalZhang25} for a similar construction and see Section \ref{sec:max-vel} for the precise expression of $\ell$ we use here. The main technical issue to justify this estimate lies in the proof of the fact that $e^{-\ell(x)}\psi_t$ is well-defined in $L^2$ and can be estimated by Gronwall's Lemma. Instead of using an analyticity argument as in \cite{SigalWu}, we introduce a bounded approximation $\ell_\varepsilon$ and establish estimates that are uniform in $\varepsilon$. A careful analysis allows us to obtain the sharp bound stated in \eqref{eq:sharp_MVE}. Some technical results entering the proof of Theorem \ref{th:max-vel-intro} are deferred to our companion paper \cite{BFG2a} where we prove a maximal velocity estimate for the non-autonomous pseudo-relativistic Schrödinger equation.

Using results from \cite{BFG2a} together with a covering argument in the spirit of that used in \cite{SigalWu}, we can extend the maximal velocity bound to non-convex subsets $X$, $Y$, up to a polynomial growth in $\mathrm{dist}(X,Y)$:

\begin{prop}[Maximal velocity estimate for general subsets]\label{th:max-vel-general-intro1}
There exists~$C_d>0$ such that, if $X,Y\subseteq\mathbb{R}^d$ are Borel subsets, $s\ge 1$, $w$ is in $\setW$ and~$\psi_0$ is in $H^{s}$ with~$\mathbf{1}_X\psi_0=\psi_0$, then the local solution $\psi$ to \eqref{eq:hartree_semi_intro} on $[0,T_{\mathrm{max}})$ given by Theorem~\ref{th:local-long-range-intro}  satisfies
    \begin{equation}\label{eq:general_MVE}
       \forall t\in[0,T_{\max})\,,\qquad \|\mathbf{1}_Y\psi_t\|_{L^2}\le C_d \, e^{t-\dXY} \, \langle\dXY\rangle^d \, \|\psi_0\|_{L^2}.
    \end{equation}
\end{prop}

The restriction to more regular initial states in Proposition \ref{th:max-vel-general-intro1}, compared to Theorem \ref{th:max-vel-intro}, comes from the fact that, in order to use the covering argument, we need to write, similarly as in \cite{ArbunichFaupinPusateriSigal23} (for the non-relativistic Hartree equation), the solution to  \eqref{eq:hartree_semi_intro} as $\psi_t=U_t\psi_0$, with $U_t=U_{t,0}$ the propagator generated by the time-dependent (and ``initial state-dependent'') Hamiltonian $H_t=\langle\nabla\rangle+w*|\psi_t|^2$. For the same reason we only consider the local solution to \eqref{eq:hartree_semi_intro} given by Theorem \ref{th:local-long-range-intro}, not the extension to a larger class of potentials provided by Theorem \ref{th: local_with_Strichartz-intro}. Under these conditions we can apply the abstract results derived in 
\cite{BFG2a} for time-dependent pseudo-relativistic Hamiltonians of the form $H_t=\langle\nabla\rangle+V_t$.

\subsection{Global existence}

As already mentioned before, for potentials $w\in L^\infty$ and initial states $\psi_0$ in $L^2$, Theorem \ref{th:local-long-range-intro} together with the conservation of mass (see Lemma \ref{lemma: en conservation} below) imply the existence of a unique global solution to \eqref{eq:hartree_semi_intro} associated to $\psi_0$. For more general potentials (and more regular initial data), the next  theorems provide the global existence of solutions to \eqref{eq:hartree_semi_intro} in two distinct regimes, namely assuming that the convolution potential $w$ has either a ``long-range'' or a ``short-range'' behavior. 

For long-range potentials, similarly as in \cite{lenzmann_semirel} (where the Yukawa-type interaction potential $w(x)=\kappa|x|^{-1}e^{-\mu|x|}$ is considered in dimension $d=3$, with $\kappa\in\mathbb{R}$, $\mu\ge0$), we can use the conservation of the energy (and of the mass) to obtain the global existence. To do that, we need that the energy defined in \eqref{eq:def_energy_intro} is well-defined and real. Hence we need to assume that the initial data $\psi_0$ belong to $H^s$ with $s\ge\frac12$ and that $w$ is even.

\begin{theo}[Global existence for long-range interaction potentials I]\label{th:global-long-range-intro}
	Let $s\ge\frac12$. There exists a universal constant $C_0>0$ such that, for all $w$ even of the form $w = w_d + w_\infty \in L^{d, \infty} + L^\infty$ and $\psi_0 \in H^{s}$ verifying
	\begin{equation}
	\label{eq: glob ex w_1 small-intro}
	    \| (w_d)_-\|_{L^{d, \infty}}\|\psi_0\|^2_{L^{2}}<C_0 ,
	\end{equation}
	Eq. \eqref{eq:hartree_semi_intro} admits a unique solution 
	\begin{equation*}
		\psi \in \mathcal{C}^0([0,\infty), H^{s}) \cap \mathcal{C}^1 ([0,\infty), H^{s-1}).
	\end{equation*}
\end{theo}

\begin{rem}$ $
\begin{enumerate}
\item The universal constant $C_0$ appearing in the statement of the previous theorem can be chosen as $C_0=2C_S^{-1}$ where $C_{S}$ is the optimal constant in the Sobolev embedding $H^{\frac12} \hookrightarrow L^{\frac{2d}{d-1},2}$. Here $L^{\frac{2d}{d-1},2}$ stands for the usual Lorentz space (see Subsection \ref{subsec:Lorentz-ineq} for the definition).
\item The same continuity property with respect to the initial data as that stated in Theorem \ref{th:local-long-range-intro} holds.
\item The smallness condition \eqref{eq: glob ex w_1 small-intro} cannot be avoided since, as mentioned in the introduction, it is proven in \cite{frohlich_lenzm07} that finite time blow-up holds , in dimension $d=3$ and for the Newtonian potential $w(x)=-\kappa|x|^{-1}$ (which belongs to $L^{3,\infty}$) with $\kappa>0$ if the mass $M(\psi_0)=\|\psi_0\|_{L^2}^2$ of the initial state is larger than some critical value.
\end{enumerate}
\end{rem}

Using the Strichartz estimates given in Appendix \ref{app: pointw dec} and taking initial states with a small enough $H^{1/2}$-norm, the class of potentials $L^{d,\infty}+L^\infty$ considered in the previous theorem can be extended to $L^{(d+1)/2,\infty}+L^\infty$. More precisely, we have the following result.

\begin{theo}[Global existence for long-range interaction potentials II]\label{th:global-long-range-intro-small-energy-opt}
	Let $s=\frac12$. There exists a universal constant $C_0>0$ such that, for all $w$ even of the form $w = w_{d/2} + w_\infty \in L^{(d+1)/2,\infty} + L^\infty\subset L^{d/2}+L^\infty$ and $\psi_0 \in H^{\frac12}$ verifying
	\begin{equation*}
	    \| (w_{d/2})_-\|_{L^{\frac{d}2}}(E(\psi_0)+\|\psi_0\|^2_{H^{\frac12}}+\|(w_\infty)_-\|_{L^\infty}\|\psi_0\|^4_{L^{2}})<C_0 ,
	\end{equation*}
	Eq. \eqref{eq:hartree_semi_intro} admits a unique solution 
	\begin{equation*}
		\psi \in \mathcal{C}^0([0,\infty), H^{\frac12}) \cap \mathcal{C}^1 ([0,\infty), H^{-\frac12})\cap L^a_{\mathrm{loc}}([0,\infty),L^{b,2}),
	\end{equation*}
	where $\frac1a=\frac1b=\frac{d-1}{2d+2}$.
\end{theo}

A more general version of Theorem \ref{th:global-long-range-intro-small-energy-opt} will be given in Theorem \ref{th:global-long-range-intro-small-energy} below, involving any admissible pair $(a,b)$ for the Strichartz estimates from Appendix \ref{app: pointw dec}. A related result is proven in \cite{ChoOzawa2} for sums of potentials of the form $w_i(x)=\kappa_i|x|^{-\alpha_i}$ with $0<\alpha_i<2d/(d+1)$ (which corresponds to potentials in $L^{q_1}+L^{q_2}$ with $(d+1)/2< q_1,q_2<\infty$).

For short-range interaction potentials, the conservation of energy is not sufficient anymore to obtain the global existence of solutions to \eqref{eq:hartree_semi_intro}, but one can rely instead on dispersive estimates satisfied by the free half-Klein-Gordon equation $i\partial_t\psi_t=\langle\nabla\rangle\psi_t$, see \eqref{eq:dispersive_est_intro} for the $L^1\to L^\infty$ estimate. For $s>\frac{d}2$, $r\ge0$ and $1\le p\le\infty$, we introduce the subspace~$S^{s,r,p}$ of~$L^{\infty}([0,\infty),H^{s})$ containing the functions~$\varphi$ such that
\begin{equation}\label{eq:def-Ssrp}
\|\varphi\|_{S^{s,r,p}}:=\sup_{t\ge0}\|\varphi_{t}\|_{H^{s}}+\sup_{t\ge0}\langle t\rangle^{\frac{d}{2r}}\|\varphi_{t}\|_{L^{\infty}\cap L^{p}}<\infty\,.
\end{equation}
We also use the notation $H^{s,p}$ for the Bessel-Sobolev space with norm $\|f\|_{H^{s,p}} = \|\langle \nabla \rangle^s f\|_{L^p}$ (see Subsection \ref{subsec:weighted-Bessel} for the precise definition).

	\begin{theo}[Global existence for short-range interaction potentials]
		\label{th:global-short-range-intro}
	Let $s\ge\frac{d}{2} +1$, $1\le r <d$ and $1\le q<\frac{2d}{3}$ be such that $\frac{1}{d}+\frac{1}{2r}<\frac{1}{q}$. Define $p$ by the relation $1=\frac{1}{r}+\frac{2}{p}$.

There exists $\varepsilon_0>0$ such that the following holds: for all $w\in\meas+L^{q}$ and~$\psi_{0}\in H^{s}\cap H^{s,p'}$ satisfying
\begin{equation}\label{eq:cond_w_psi0 intro}
\|w\|_{\meas+L^{q}}\,\|\psi_{0}\|_{H^{s}\cap H^{s,p'}}^{2}\le\varepsilon_0,
\end{equation}
Eq.~\eqref{eq:hartree_semi_intro} admits a unique solution $\psi$ in $S^{s,r,p}$.
 
 This solution $\psi$ belongs to $\mathcal{C}^0([0,\infty),H^{s})\cap \mathcal{C}^1([0,\infty),H^{s-1})$ and satisfies, for all $t\ge0$,
\begin{align}
&\|\psi_t\|_{L^\infty\cap L^{p}} \lesssim \langle t\rangle^{-\frac{d}{2r}}\|\psi_0\|_{H^{s}\cap H^{s,p'}},\label{eq:estim1_psit}\\
&\big\|\psi_t-\psi^{(0)}_t\big\|_{L^\infty\cap L^{p}} \lesssim \varepsilon_0\langle t\rangle^{-\frac{d}{2r}}\|\psi_0\|_{H^{s}\cap H^{s,p'}},\\
&\|\psi_t\|_{H^s} \lesssim \|\psi_0\|_{H^{s}\cap H^{s,p'}}, \label{eq:estim_psit_energy}\\
&\big\|\psi_t-\psi^{(0)}_t\big\|_{H^s} \lesssim \varepsilon_0 \|\psi_0\|_{H^{s}\cap H^{s,p'}},\label{eq:estim4_psit}
\end{align}
where $\psi^{(0)}_t:=e^{-it\langle\nabla\rangle}\psi_0$. In particular, the map $\psi_0\mapsto \psi\in S^{s,r,p}$, defined on the subset of $\psi_0$'s in $H^s\cap H^{s,p'}$ satisfying \eqref{eq:cond_w_psi0 intro},
is continuous.
	\end{theo}

\begin{rem}
	Theorem~\ref{th:global-short-range-intro} gives global existence for all potentials $w \in L^{q,\infty}$ with $1< q<\frac{2d}{3}$, by applying the theorem with $q+\varepsilon<\frac{2d}{3}$ and observing that $L^{q,\infty}\subset \meas + L^{q+\varepsilon} $. In Figure \ref{fig:well-posedness-short} we represent the class of admissible potentials for Theorem~\ref{th:global-short-range-intro}. 
\end{rem}

As mentioned in the introduction, global existence -- and small initial data scattering -- are proven in \cite{ChoOzawa2} for (sums of) short-range interaction potentials of the form $w_i(x)=\kappa_i|x|^{-\alpha_i}$ with $2<\alpha_i<d$ (which corresponds to potentials in $L^{q_1}+L^{q_2}$ with $1< q_1,q_2<d/2$). The proof in \cite{ChoOzawa2} relies in particular on the use of Strichartz estimates with Schrödinger admissible pairs. See also \cite{NakamuraTsutaya} for a single potential $w$ satisfying $|w(x)|\lesssim|x|^{-\alpha}$ with $\frac32<\alpha<d$ and \cite{Yang} for small data scattering for low regularity initial states and a smooth potential $w$.

\subsection{Scattering states and scattering operators}

In the case of short-range interaction potentials, one can show that the global solution given by Theorem \ref{th:global-short-range-intro} scatters to a free solution in $H^s$, in the sense of the following theorem. For any $\delta>0$, we denote by  $\mathcal{B}_{\mathcal{E}}(\delta)$ the closed ball of radius $\delta$ in a normed vector space~$\mathcal{E}$.

	\begin{theo}[Scattering for short-range interaction potentials]
		\label{th:scattering-short-range-intro}
Under the conditions of Theorem \ref{th:global-short-range-intro}, the global solution $\psi$ to \eqref{eq:hartree_semi_intro} scatters to a free solution:
			\begin{equation}\label{eq:decay-scat0-intro}
			\| \psi_t - e^{-it\langle\nabla\rangle} \psi_+\|_{H^s} \to 0, \qquad t\to\infty,
		\end{equation}
		where the scattering state $\psi_+\in H^s$ is defined by
		\begin{equation}\label{eq:def_psi+}
			\psi_+ := \psi_0 - i \int_0^\infty e^{i\tau\langle\nabla\rangle} (w * |\psi_\tau|^2)\psi_\tau\, \mathrm{d}\tau.
		\end{equation}
Moreover there exists $\delta_w>0$ such that the ``inverse'' wave operator 
		\begin{align*}
    W_+:\mathcal{B}_{H^s\cap H^{s,p'}}(\delta_w)&\to H^s, \\ 
    \psi_0&\mapsto  W_+\psi_0:= \psi_0 - i \int_0^\infty e^{i\tau\langle\nabla\rangle} (w * |\psi_\tau|^2)\psi_\tau\, \mathrm{d}\tau,
\end{align*}
is continuous. 
	\end{theo}
	
\begin{rem}\label{rk:decay-scat0-intro-rate}
    Our proof gives an explicit rate of decay in \eqref{eq:decay-scat0-intro}, namely we have the bound
    \begin{equation}\label{eq:decay-scat0-intro-rate}
			\| \psi_t - e^{-it\langle\nabla\rangle} \psi_+\|_{H^s} \lesssim \langle t\rangle^{1-\frac dr \min\{1, \frac{r}{q}\}}\|\psi_0\|_{H^s\cap H^{s,p'}},
		\end{equation}
	uniformly in $t\ge0$.
\end{rem}

As in previous works \cites{ChoOzawa1,ChoOzawa2,ChoOzawaSazaki}, the proof of Theorem \ref{th:scattering-short-range-intro} is a rather straightforward application of the regularity and decay properties of the solution to \eqref{eq:hartree_semi_intro}, see \eqref{eq:estim1_psit}--\eqref{eq:estim4_psit} in our context.

The wave operator $\Omega_+$, mapping any scattering state $\psi_+$ to an initial state $\psi_0$ can be defined similarly as $W_+$, constructing a solution to the Cauchy problem at $\infty$ for the pseudo-relativistic Hartree equation,
\begin{align}
	\label{eq: hartree semir infinity}
	\begin{cases}
				i\partial_t \psi_t = (\langle\nabla\rangle + w*|\psi_t|^2)\psi_t, \\
    \lim_{t\to\infty}\|\psi_t-e^{-it\langle\nabla\rangle}\psi_+\|_{L^2}=0. 
	\end{cases}
\end{align}
See Theorems \ref{th:global-short-range-infinity} and \ref{cor:scat2} below for precise statements. Formally, $\Omega_+$ is the inverse of $W_+$ by definition; however, since $W_+$ maps $H^s\cap H^{s,p'}$ to $H^s$ and $\Omega_+$ is defined on (a small ball in) $H^s\cap H^{s,p'}$, the composition $\Omega_+W_+$ is ill-defined in general. To overcome this difficulty, for suitable values of $\gamma$ and $s$, we restrict $W_+$ to the weighted Sobolev space $H^s_\gamma$ with norm
\begin{equation*}
    \|\varphi\|_{H^s_\gamma}:=\|\langle x\rangle^\gamma\langle\nabla\rangle^s\varphi\|_{L^2}.
\end{equation*}
We will choose parameters ensuring that $H^s_\gamma\hookrightarrow H^s\cap H^{s,p'}$, and hence that $W_+$ is well-defined on a ball in $H^s_\gamma$ with sufficiently small radius. Moreover, restricting the class of admissible potentials to $w\in\mathcal{M}+L^q$ with $1\le q<\frac{d}{2}$, we will show that $W_+$ maps this ball in (a small ball in) $H^s\cap H^{s,p'}$, and likewise for $\Omega_+$. We then have the following result on the invertibility of the wave operator.

\begin{theo}[Right invertibility of the wave operators for short-range interaction potentials]\label{thm:invertibility-intro}
    Let $s\ge\frac d2+1$, $\max\{1,\frac{d}{4}\}\le r<\frac{d}{2}$, $1\le q\le r$ and $\frac{d}{2r}<\gamma<\min\{2,\frac dr-1\}$. Let $w \in \meas + L^q$.  There exists $\delta_w>0$ such that, for all $\varphi\in\mathcal{B}_{H^s_\gamma}(\delta_w)$,
	\begin{align}\label{eq:inverse-intro}
		\Omega_+W_+\varphi=W_+\Omega_+\varphi=\varphi.
	\end{align}
\end{theo}

\begin{rem}
    The operators $\Omega_+$ and $W_+$ may not map $\mathcal{B}_{H^s_\gamma}(\delta_w)$ into itself. We only have that $\Omega_+$ and $W_+$ map $\mathcal{B}_{H^s_\gamma}(\delta_w)$ into $\mathcal{B}_{H^s_\gamma}(\delta'_w)$ for some $\delta'_w>\delta_w$ (one can take $\delta'_w=\delta_w+C\varepsilon_0$ for some positive constant $C$ and some small enough $\varepsilon_0>0$, see Theorem \ref{thm:invertibility} below for a precise statement). Eq.~\eqref{eq:inverse-intro} therefore shows that $W_+:\mathrm{Ran}\,\Omega_+\to\mathcal{B}_{H^s_\gamma}(\delta_w)$ is a right inverse of $\Omega_+:\mathcal{B}_{H^s_\gamma}(\delta_w)\to\mathrm{Ran}\,\Omega_+$ and vice versa.  
\end{rem}

The proof of Theorem \ref{thm:invertibility-intro} requires to control, in a rather precise way, the asymptotic behavior in weighted $L^p$ spaces of global solutions to  \eqref{eq:hartree_semi_intro} (constructed in Theorem \ref{th:scattering-short-range-intro}). For this reason, the class of allowed potentials in Theorem \ref{thm:invertibility-intro} is more restrictive than that considered in Theorem \ref{th:scattering-short-range-intro}. Our proof will also provide estimates on the operators $W_+-\mathrm{Id}$ and $\Omega_+-\mathrm{Id}$, see Theorem \ref{thm:invertibility} below for a precise statement.  

Theorem \ref{thm:invertibility-intro} will be crucial in order to construct a suitable set of initial states leading to an asymptotic minimal velocity estimate, see Theorem \ref{th:min-vel-intro} below.

\subsection{Asymptotic propagation and minimal velocity estimates}

Our last concern is the asymptotic behavior of the speed of propagation of solutions to \eqref{eq:hartree_semi_intro}. We introduce the ``instantaneous'' velocity operator 
\begin{equation*}
\Theta:=[x,\langle\nabla\rangle]=-i\nabla\langle\nabla\rangle^{-1}.
\end{equation*}
The next theorem shows that, along the evolution associated to \eqref{eq:hartree_semi_intro}, the instantaneous velocity and the average velocity $\frac{x}{t}$ converge to each other.

\begin{theo}[Asymptotic phase-space propagation estimate for short-range interaction potentials]\label{th:propag-est-intro}
Let $f,g\in \mathcal{C}_0^\infty(\mathbb{R})$ be such that $\mathrm{supp}(g)\cap\mathrm{supp}(f)=\emptyset$. Under the conditions of Theorem \ref{th:global-short-range-intro}, and assuming in addition that $\|\langle x\rangle\psi_0\|_{L^2}<\infty$, the global solution $\psi$ to \eqref{eq:hartree_semi_intro} given by Theorem \ref{th:global-short-range-intro} satisfies
\begin{equation}\label{eq:propag1-intro}
\Big\|g\Big(\frac{x^2}{t^2}\Big)f(\Theta^2) \psi_t\Big\|_{L^2} \to0, \qquad t\to\infty.
\end{equation}
\end{theo}

\begin{rem}\label{rem:asympt_max-vel}$ $
\begin{enumerate}
\item Since $0\le \Theta^2\le1$, 
taking $f\in \mathcal{C}_0^\infty(\mathbb{R})$ such that $f=1$ on $[0,1]$, we have $f(\Theta^2)=\mathrm{Id}$. Hence the previous theorem implies that, for all $\varepsilon>0$,
\begin{equation*}
\Big\|\mathbf{1}_{[1+\varepsilon,\infty)}\Big(\frac{x^2}{t^2}\Big) \psi_t\Big\|_{L^2} \to0, \qquad t\to\infty.
\end{equation*}
We thus recover a maximal velocity estimate, in the sense that the boson star cannot propagate faster than the speed of light (equal to $1$ in our unit), asymptotically as $t\to\infty$.
\item Our proof gives an explicit rate of decay in \eqref{eq:propag1-intro}, more precisely we will show that
\begin{equation*}
\Big\|g\Big(\frac{x^2}{t^2}\Big)f(\Theta^2) \psi_t\Big\|_{L^2} \lesssim \big(\langle t \rangle^{-1} + \langle  t\rangle^{1-\frac{d}{r}\min\{1, \frac{r}{q}\}}\big)\|\langle x\rangle\psi_0\|_{L^2},
\end{equation*}
uniformly in $t\ge0$.
\end{enumerate}
\end{rem}

Combining Theorems \ref{th:scattering-short-range-intro} and \ref{th:propag-est-intro}, one can deduce a minimal velocity estimate, in the sense that if the initial state $\psi_0$ is associated to a scattering state $\psi_+$ with an instantaneous velocity localized in $[\alpha,1]$ with $0<\alpha<1$, then 
\begin{equation*}
\Big\|\mathbf{1}_{[0,\alpha)}\Big(\frac{x^2}{t^2}\Big) \psi_t\Big\|_{L^2} \to 0.
\end{equation*}
It is however unclear in general how to identify initial states $\psi_0$'s associated to scattering states $\psi_+$ with an instantaneous velocity localized in $[\alpha,1]$. To overcome this difficulty we restrict the class of admissible potentials and use Theorem \ref{thm:invertibility-intro}. This leads to the following result.

\begin{theo}[Asymptotic minimal velocity estimate for short-range interaction potentials]\label{th:min-vel-intro}
    Let $s\ge\frac  d2+1$, $\max\{1,\frac{d}{4}\}\le r<\frac{d}{2}$, $1\le q\le r$ and $\frac{d}{2r}<\gamma<\min\{2, \frac dr-1\}$. Let $w\in \mathcal{M}+L^{q}$ and $0<\alpha<1$. There exists $\delta_w$ such that, for all initial states
    \begin{equation*}
    \psi_0=\Omega_+\psi_+ \quad\text{ with }\quad \psi_+\in \mathcal{B}_{H^s_\gamma}(\delta_w) \quad\text{ and }\quad \psi_+=\mathbf{1}_{[\alpha,1]}(\Theta^2)\psi_+,
    \end{equation*}
    the global solution $\psi$ to \eqref{eq:hartree_semi_intro} given by Theorem \ref{th:global-short-range-intro} satisfies
    \begin{equation*}
\Big\| \mathbf{1}_{[0,\alpha)}\Big(\frac{x^2}{t^2}\Big) \psi_t\Big\|_{L^2} \to 0, \qquad t\to\infty.
\end{equation*}
\end{theo}

\begin{rem}
    It is not difficult to construct states $\psi_+\in \mathcal{B}_{H^s_\gamma}(\delta_w)$ satisfying in addition $\psi_+=\mathbf{1}_{[\alpha,1]}(\Theta^2)\psi_+$. Indeed, given a smooth function $f\in\mathcal{C}_0^\infty(\mathbb{R})$ such that $\mathrm{supp}(f)\subset(\alpha,1)$, it suffices to choose $\tilde\psi_+$ such that $f(\Theta^2)\tilde\psi_+=\tilde\psi_+$ and $\|\tilde\psi_+\|_{H^s_\gamma}$ is small enough. Since one can verify that $\|f(\Theta^2)\psi_+\|_{H^s_\gamma}\le C_f\|\psi_+\|_{H^s_\gamma}$ for some $C_f>0$, the state given by $\psi_+=f(\Theta^2)\tilde\psi^+$ satisfies the conditions of the theorem for $\|\tilde\psi_+\|_{H^s_\gamma}$ small enough. 
\end{rem}

It should be noted that the maximal velocity estimate mentioned in Remark \ref{rem:asympt_max-vel} holds for all initial states, without requiring a construction as in Theorem  \ref{th:min-vel-intro}. This is due to the fact that the instantaneous velocity $\Theta$ is bounded. For the non-relativistic Hartree equation considered in \cite{ArbunichFaupinPusateriSigal23}, proving a maximal velocity estimate does require a suitable construction of initial states; instead of using ``asymptotic energy cutoffs'' as in \cite{ArbunichFaupinPusateriSigal23}, it would be interesting to follow the approach developed in the present paper based on  invertibility properties of the wave operator.

\subsection{Summary of results}
For the reader's convenience, we summarize the conditions required to obtain our results in two tables and three figures. Table~\ref{tab: exist} collects the well-posedness results. The relationships between the regularity of the initial data and the class of admissible potentials in these results is illustrated in Figure~\ref{fig:well-posedness} for a generic spatial dimension $d\ge4$ while Figure~\ref{fig:well-posedness-3d} highlights features specific to dimension $3$, including some endpoint cases. Figure~\ref{fig:well-posedness-short} depicts the relation between the parameters $p$ and $q$ appearing in Theorem~\ref{th:global-short-range-intro}, which concerns the global existence in the case of short-range potentials; this existence theory is the one employed in our analysis of scattering and minimal velocity estimates. Finally, Table~\ref{tab:speed} summarizes our results on the maximal and minimal velocity estimates.
\begin{table}[H]
    \centering
\begin{tabular}{c|rl|c|c|c|c}
\multirow{2}{*}{$s$ in} & \multirow{2}{*}{$w$ in} &  & \multicolumn{2}{c|}{Example of $w$} & Well- & \multirow{2}{*}{Theorem}\tabularnewline
\cline{4-5} \cline{5-5} 
 &  &  & $|x|^{-\alpha}$ , $\alpha$ in & $\;\delta_{0}\;$ & posedness & \tabularnewline
\hline 
\rcell$\big[0,\frac{d}{2}\big)$  & \rcell$L^{\frac{d}{2s},\infty}+L^{\infty}\phantom{,}$  & \rcell & $[0,2s]$\rcell & \rcell & \multirow{5}{*}{Local} & \rcell\ref{th:local-long-range-intro}\tabularnewline
\rcell$\big\{\frac{d}{2}\big\}$  & \rcell$L^{q,\infty}+L^{\infty},$  & \negthickspace{}\negthickspace{}$q>1$\rcell & $[0,d)$\rcell & \rcell &  & \rcell\ref{th:local-long-range-intro}\tabularnewline
\rcell$\big(\frac{d}{2},\infty\big)$ & \rcell$\meas+L^{\infty}\phantom{,}$  & \rcell & $[0,d)$\rcell & $\checkmark$ \rcell &  & \rcell\ref{th:local-long-range-intro}\tabularnewline
\bcell$\big[0,\frac{d+1}{2d-2}\big)$ & \bcell$L^{\frac{d+1}{4s},\infty}+L^{\infty}\phantom{,}$  & \bcell & $\big[0,\frac{4ds}{d+1}\big]$\bcell & \bcell &  & \bcell\ref{th: local_with_Strichartz-intro}\tabularnewline
\bcell$\big[\frac{d+1}{2d-2},\infty\big)$ & \bcell$L^{\frac{d-1}{2},\infty}+L^{\infty}\phantom{,}$  & $^{\dagger}$\bcell & $\big[0,\frac{2d}{d-1}\big)$\bcell & \bcell &  & \bcell\ref{th: local_with_Strichartz-intro} \tabularnewline
\cline{6-6} 
\rcell$\{0\}$  & \rcell$L^{\infty}\phantom{,}$  & \rcell & \rcell$\{0\}$ & \rcell & \multirow{4}{*}{Global} & \rcell\ref{th:local-long-range-intro}\tabularnewline
\rcell$\big[\frac{1}{2},\infty\big)$ & \rcell$L^{d,\infty}+L^{\infty}\phantom{,}$  & \rcell & \rcell$[0,1]$ & \rcell &  & \rcell\ref{th:global-long-range-intro}\tabularnewline
\bcell$\big\{\frac{1}{2}\big\}$  & \bcell$L^{\frac{d+1}{2},\infty}+L^{\infty}\phantom{,}$  & \bcell & \bcell$\big[0,\frac{2d}{d+1}\big]$ & \bcell &  & \bcell\ref{th:global-long-range-intro-small-energy-opt}\tabularnewline
\gcell$\big[\frac{d}{2}+1,\infty\big)$ & \gcell$\meas+L^{q},$  & \negthickspace{}\negthickspace{}$q<\frac{2d}{3}$\gcell & \gcell$\big(\frac{3}{2},d\big)$ & \gcell$\checkmark$  &  & \gcell\ref{th:global-short-range-intro}\tabularnewline
\end{tabular}\caption{Well-posedness results for \eqref{eq:hartree_semi_intro} with
$\psi_{0}\in H^{s}$. We underline that the result of Theorem \ref{th:global-short-range-intro}
requires the additional condition $\psi_{0}\in H^{s,p}$ for suitable
$p$. Moreover, while at fixed $s<\frac{d}{2}$ Theorem \ref{th: local_with_Strichartz-intro}
provides local existence for a larger class of potentials compared
to Theorem \ref{th:local-long-range-intro}, it guarantees uniqueness
of a solution only in a smaller functional space.\protect \\
$^{\dagger}$If $d=3$, then one can actually only take $w$
in $L^{q,\infty}+L^{\infty}$ for any $q>\frac{d-1}{2}$.}
	\label{tab: exist}
\end{table}
\begin{figure}[H]
    \begin{center}
         \includegraphics[angle=0,width=16cm,origin=c]{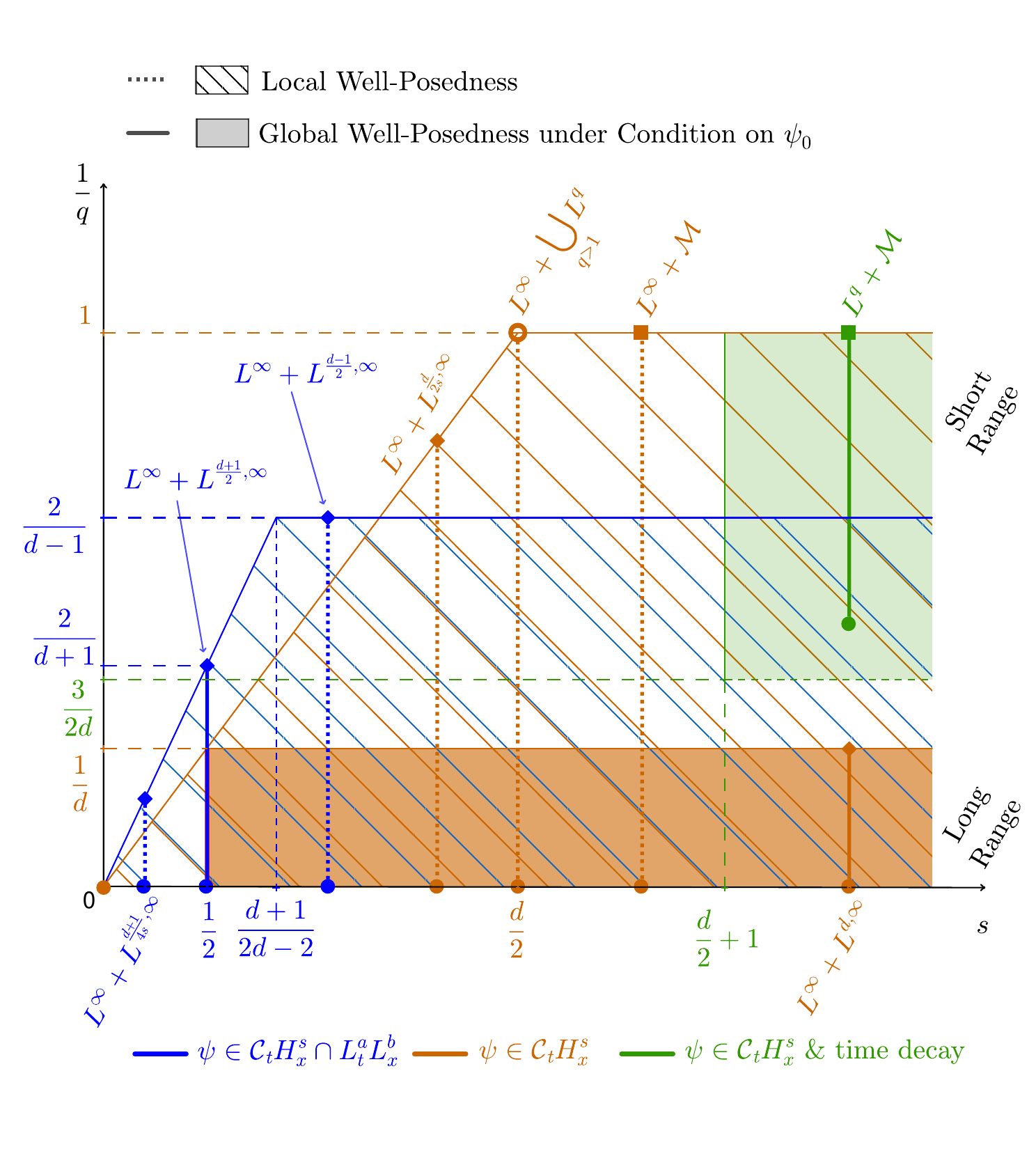}
    \caption{Dimension $d\geq 4$. Well-Posedness of \eqref{eq:hartree_semi_intro}: Admissible class of $w$ depending on the regularity $s$ of $\psi_0$.\\
    A thick vertical line with extremities at $1/q_1$ and $1/q_2$ represents the space $L^{q_1,\tilde{q}_1}+L^{q_2,\tilde{q}_2}$, with $\tilde{q}_j=q_j$ if the extremity is a disk (\mbox{\LARGE $\bullet$}), $\tilde{q}_j=\infty$ if the extremity is diamond (\mbox{\footnotesize \rotatebox[origin=c]{45}{$\blacksquare$}}). There are two special cases: the $L^{q,\tilde{q}}$ space is replaced by $\meas$ if the extremity is a square (\mbox{\footnotesize$\blacksquare$}), and  by $\bigcup_{q>q_j}L^q$ if it is a circle (\mbox{\Large $\boldsymbol{\circ}$}). 
    Solid lines correspond to global well-posedness while dotted lines correspond to local well-posedness.
    }
    \label{fig:well-posedness}
    \end{center}
\end{figure}		
\begin{figure}[H]\label{fig:}
    \begin{center}
         \includegraphics[angle=0,width=16cm,origin=c]{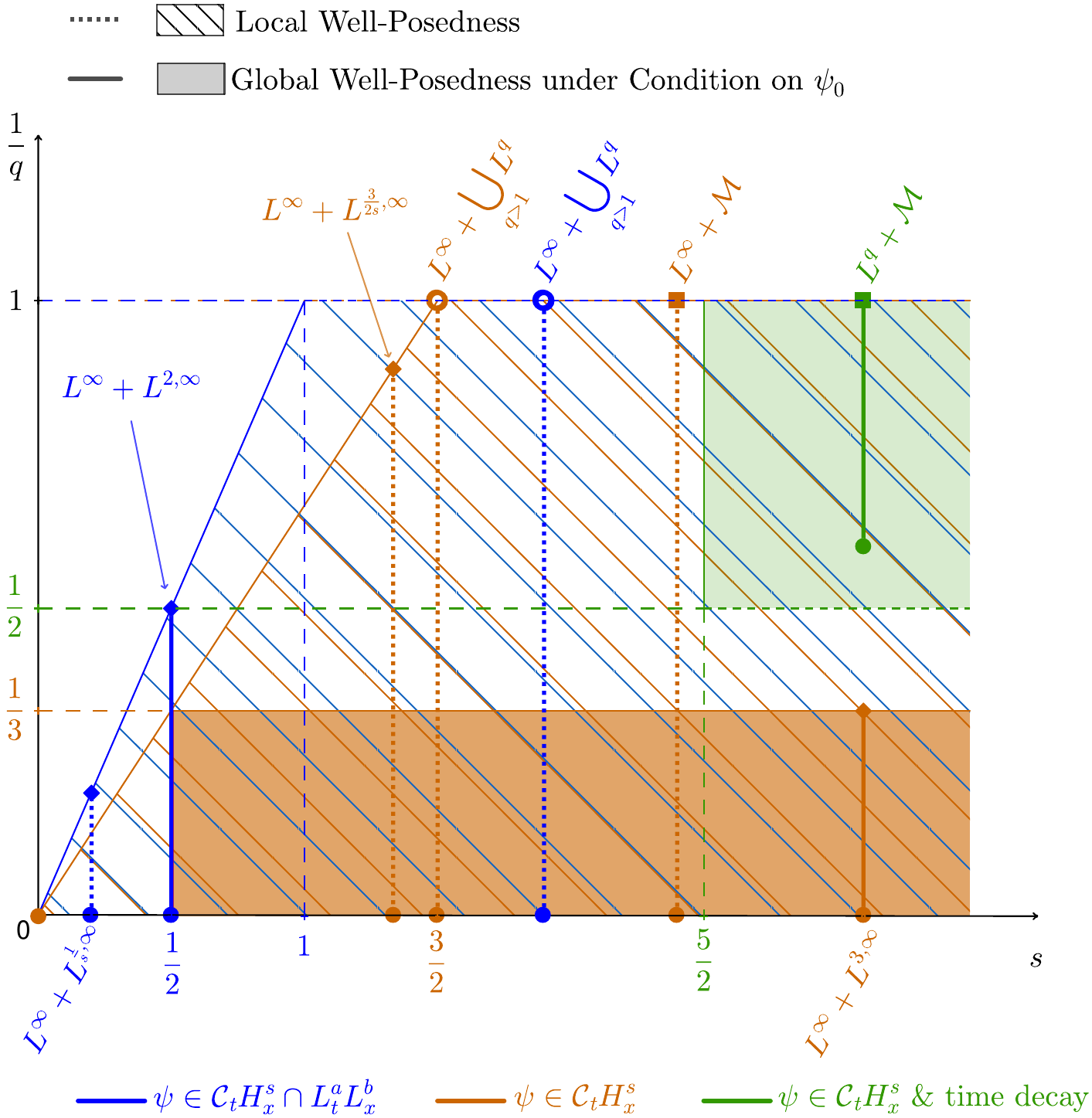}
    \caption{Dimension $d=3$. Well-Posedness of \eqref{eq:hartree_semi_intro}: Admissible class of $w$ depending on the regularity $s$ of $\psi_0$.\\
    A thick vertical line with extremities at $1/q_1$ and $1/q_2$ represents the space $L^{q_1,\tilde{q}_1}+L^{q_2,\tilde{q}_2}$, with $\tilde{q}_j=q_j$ if the extremity is a disk (\mbox{\LARGE$\bullet$}), $\tilde{q}_j=\infty$ if the extremity is diamond (\mbox{\footnotesize \rotatebox[origin=c]{45}{$\blacksquare$}}). There are two special cases: the $L^{q,\tilde{q}}$ space is replaced by $\meas$ if the extremity is a square (\mbox{\footnotesize$\blacksquare$}), and  by $\bigcup_{q>q_j}L^q$ if it is a circle (\mbox{\Large$\boldsymbol{\circ}$}). 
    Solid lines correspond to global well-posedness while dotted lines correspond to local well-posedness.
    }
    \label{fig:well-posedness-3d}
    \end{center}
\end{figure}				
\begin{figure}[H]
    \begin{center}
         \includegraphics[angle=0,width=12cm,origin=c]{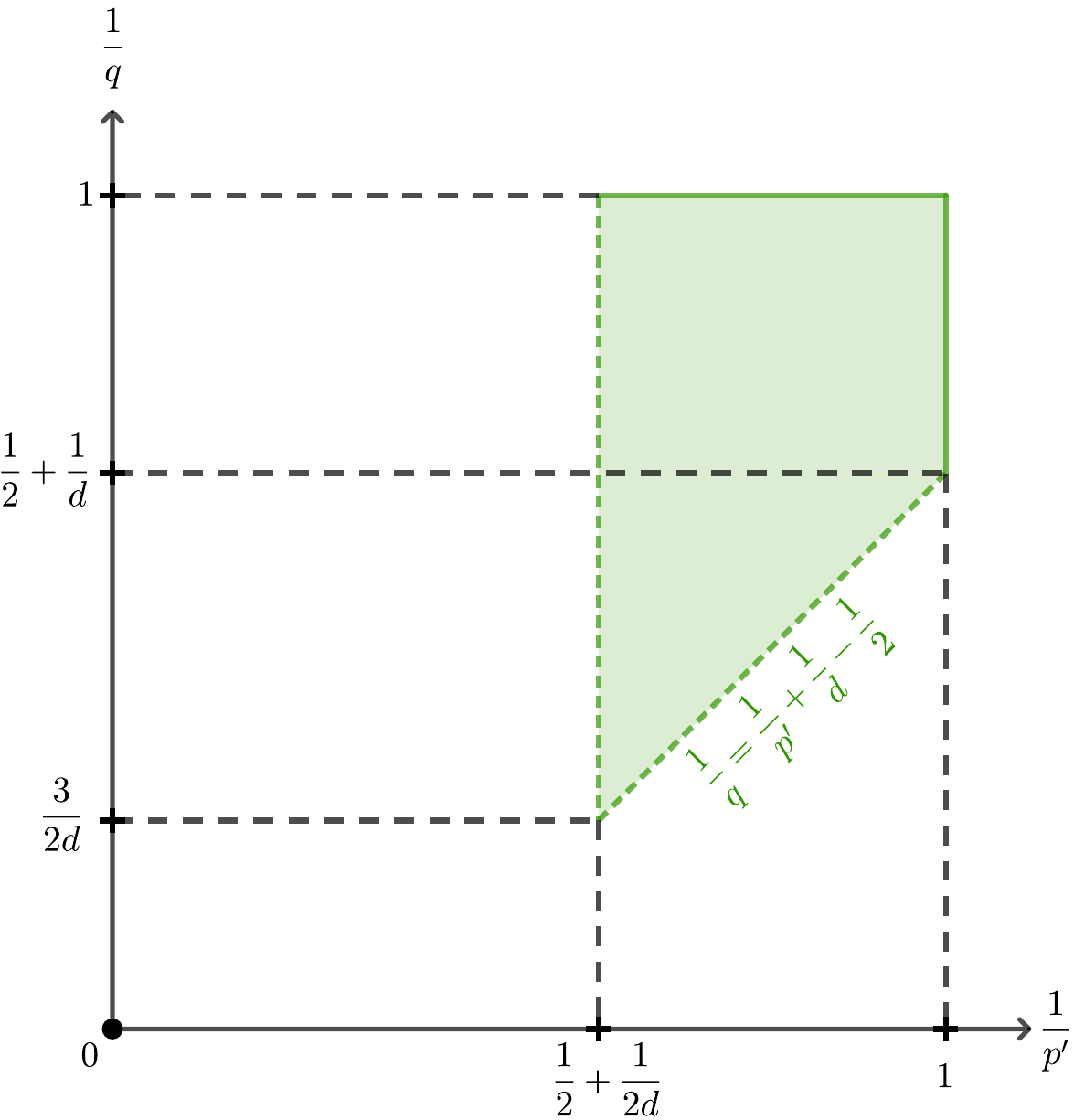}
    \caption{Global Well-Posedness in the Short range case: The green region represents the domain of admissible pairs $(\frac1{p'},\frac1{q})$ such that, if $\psi_0$ is in $H^s\cap H^{s,p'}$ and $w$ in $L^q+\meas$ and their norms are sufficiently small, then \eqref{eq:hartree_semi_intro} admits a global solution, provided by Theorem~\ref{th:global-short-range-intro}. }
    \label{fig:well-posedness-short}
    \end{center}
\end{figure}
\begin{table}[H]
    \centering
	\begin{tabular}{c|rl|c|c|c|c}
\multirow{2}{*}{$s$ in} & \multirow{2}{*}{$w$ in} &  & \multicolumn{2}{c|}{Example of $w$} & Velocity & \multirow{2}{*}{Theorem}\tabularnewline
\cline{4-5} \cline{5-5} 
 &  &  & $|x|^{-\alpha}$, $\alpha$ in & $\ \delta_{0}\ $ & estimate & \tabularnewline
\hline 
$\big[0,\frac{d}{2}\big)$ & $L^{\frac{d}{2s},\infty}+L^{\infty}\text{\ensuremath{\phantom{,}}}$ &  & $[0,2s]$ &  & \multirow{3}{*}{Maximal} & \ref{th:max-vel-intro}\tabularnewline
$\big\{\frac{d}{2}\big\}$ & $L^{q,\infty}+L^{\infty},$ & \negthickspace{}\negthickspace{}$q>1$ & $[0,d)$ &  &  & \ref{th:max-vel-intro}\tabularnewline
$\big(\frac{d}{2},\infty\big)$ & $\meas+L^{\infty}\text{\ensuremath{\phantom{,}}}$ &  & $[0,d)$ & $\checkmark$ &  & \ref{th:max-vel-intro}\tabularnewline
\cline{6-6}
\multirow{2}{*}{$\big[\frac{d}{2}+1,\infty\big)$} & \multirow{2}{*}{$\mathcal{M}+L^{q},$} & \multirow{2}{*}{\negthickspace{}\negthickspace{}$q<\frac{d}{2}$} & \multirow{2}{*}{$(2,d)$} & \multirow{2}{*}{$\checkmark$} & Maximal & \ref{th:max-vel-intro}\tabularnewline
 &  &  &  &  & \& Minimal & \& \ref{th:min-vel-intro}\tabularnewline
\end{tabular}
	\caption{Speed of propagation results for \eqref{eq:hartree_semi_intro} with $\psi_0 \in H^s$. We underline that the minimal velocity result holds under the additional condition that $\psi_0$ is in a weighted $H^s$ space.} 
    \label{tab:speed}
\end{table}
\section*{Acknowledgements}
This research was funded, in whole or in part, by l’Agence Nationale de la Recherche
(ANR), project ANR-22-CE92-0013. We are grateful to I.M.~Sigal for fruitful collaborations.

\section{Preliminaries}\label{sec:prelim}

 In this preliminary subsection, we first introduce functional spaces that will play an important role in our analysis, the Bessel-Sobolev spaces, and we recall a version of the fractional Leibniz rule in this context. Next we derive some multilinear estimates that will allow us to control the non-linearity appearing in the pseudo-relativistic Hartree equation \eqref{eq:hartree_semi_intro}.

\subsection{Weighted Bessel-Sobolev spaces and fractional Leibniz rule}\label{subsec:weighted-Bessel}

We denote by $\mathcal{S}$ the usual Schwartz space of rapidly decreasing functions from $\mathbb{R}^d$ to $\mathbb{C}$ and by $\mathcal{S}'$ the associated space of tempered distributions. For a tempered distribution $f$ in $\mathcal{S}'$ and $s$ in $\mathbb{R}$, we write $\langle \nabla \rangle^s f = \mathcal{F}^{-1}(\langle \xi \rangle^s \mathcal{F} f)$. Recall that $\mathcal{F}$ stands for the Fourier transform normalized such that $\mathcal{F}$ is unitary on $L^2$. For $p\in [1,\infty]$, the Bessel-Sobolev spaces  are defined by
\[
H^{s,p}:=\{f\in \mathcal{S}' \mid \langle \nabla \rangle^s f \in L^p\},
\]
and endowed with the norm $\|f\|_{H^{s,p}} = \|\langle \nabla \rangle^s f\|_{L^p}$.
If $p=2$, we use the shorhand $H^s=H^{s,2}$. When $1<p<\infty$ the Bessel-Sobolev spaces $H^{s,p}$ coincide with the fractional Sobolev spaces $W^{s,p}$ with equivalent norms, but it is not the case when $p=1$ or $p=\infty$, which we we will also need in the sequel. 

In the case of non integer exponents $s$, the Leibniz rule cannot be used, but a useful replacement is the fractional Leibniz rule, see e.g. \cite{grafakos_oh_2014}.

\begin{prop}
\label{prop:Leibniz-fractionnaire}
    Let $1\leq p_0<\infty$, $1<p_j,\tilde{p}_j\leq\infty$ satisfying $\frac{1}{p_0}=\frac{1}{p_j}+\frac{1}{\tilde{p}_j}$ and $j\in\{1,2\}$. If $s\ge0$,  there exists $C>0$ such that, for all $f,g\in\mathcal{S}$,
    \[
    \|fg\|_{H^{s,p_0}}  \leq C \big(  \|f\|_{H^{s,p_1}}\|g\|_{L^{\tilde{p}_1}}+\|f\|_{L^{p_2}}\|g\|_{H^{s,\tilde{p}_2}} \big)\,.
    \]
\end{prop}

\begin{rem}
    We will use Proposition~\ref{prop:Leibniz-fractionnaire} with functions $f$ in $H^{s,p_1}$ and $L^{p_2}$ and $g$ in $H^{s,\tilde{p}_2}$ and $L^{\tilde{p}_1}$,
    which is fine if these spaces are included in the closure of $\mathcal{S}$ for the corresponding norms.
    
    We will sometime have $p_2=\infty$ (respectively $\tilde{p}_1=\infty$). In this case, we will have to make sure that $f$ (respectively $g$) is in the closure of $\mathcal{S}$ with respect to the $L^\infty$ norm, which is the space $\mathcal{C}^0_{\infty}$ of continuous functions vanishing at infinity.
\end{rem}

We will also use the following weighted Bessel-Sobolev spaces, defined as
\begin{equation}\label{eq:def_Hsgamma0}
    H^{s,p}_\gamma := \{ f \in \mathcal{S}' \ |\ \langle x\rangle^\gamma \langle \nabla \rangle^sf \in L^p\}, \quad p \in [1, \infty], \quad s\geq 0, \quad \gamma\ge0,
\end{equation}
and endowed with the norm $\|f\|_{H^{s,p}_\gamma} = \|\langle x\rangle^\gamma \langle \nabla \rangle^sf\|_{L^p}$. We observe that, for all $\gamma\ge0$ and $s\ge0$,
\begin{equation*}
    H^s_\gamma\hookrightarrow H^s,
\end{equation*}
while if $\gamma>\frac{d}{2r}$ wit $1\le r\le\infty$, it follows from H\"older's inequality that 
\begin{equation*}
    H^s_\gamma\hookrightarrow H^{s,p'},
\end{equation*}
with $p'=\frac{2r}{r+1}$.
We also define the shorthands 
\begin{equation}\label{eq:def_Hsgamma}
    H^s_\gamma:= H^{s,2}_\gamma, \quad L^p_\gamma := H^{0,p}_\gamma.
\end{equation}

In Appendix \ref{app: commute} we recall the following property (we only state and prove it for $0\le\gamma\le2$ for simplicity and since we only need it for such $\gamma$ in our context). 
\begin{lemma}
\label{lemma: weight sob norm}
     Let $0\leq \gamma\leq 2$,  $s\geq 0$ and $1\leq p \leq \infty$. There exist $c, c'>0$ such that 
 \begin{equation*}
    c'\,  \| \langle \nabla \rangle^s \langle x \rangle^\gamma f \|_{L^p} \leq \|f\|_{  H^{s,p}_\gamma } \leq c\,  \| \langle \nabla \rangle^s \langle x \rangle^\gamma f \|_{L^p} ,
 \end{equation*}
 for all $f\in H^{s,p}_\gamma$.
\end{lemma}
Combined with Proposition \ref{prop:Leibniz-fractionnaire}, this lemma implies a weighted version of the fractional Leibniz rule. 
\begin{prop}
\label{prop:weighted leibniz}
    Let $0\leq \gamma\leq 2$, $1\leq p_0<\infty$, $1<p_j,\tilde{p}_j\leq\infty$ satisfying $\frac{1}{p_0}=\frac{1}{p_j}+\frac{1}{\tilde{p}_j}$ and $j\in\{1,2\}$. If $s\geq 0$, then there exists $C>0$ such that, for all $f,g\in\mathcal{S}$,
    \begin{equation*}
        \|fg\|_{H^{s,p_0}_\gamma} \leq C ( \|f \|_{H^{s,p_1}} \| g \|_{L_\gamma^{\tilde{p}_1}} + \|f\|_{L^{p_2}}\| g\|_{H_\gamma^{s,\tilde{p}_2}}) .
    \end{equation*}
\end{prop}

\subsection{Functional inequalities in Lorentz spaces }
\label{subsec:Lorentz-ineq}
For $p$ in $(1,\infty)$ and $q$ in $[1,\infty]$ the Lorentz space $L^{p,q}$ can be defined as the real interpolation space $[L^1,L^\infty]_{1-1/p,q}$. The spaces $L^{p,\infty}$ coincide with the weak-Lebesgue spaces, $L^{p,p}=L^p$ and if $1\leq q_1\leq q_2\leq\infty$ then $L^{p,q_1}$ is continuously embedded in $L^{p,q_2}$.

The Lorentz spaces have several interesting features, in particular they include the weak spaces $L^{p,\infty}$, which themselves include $1/|x|^{d/p}$,
Hölder and Young inequalities hold in Lorentz spaces, as well as a refined version of the usual Sobolev embeddings. More precisely we have the following three propositions, which are proven for instance in \cite{LemarieRieusset}*{Chapter 2}.

Similarly as before, for $f$ in $\mathcal{S}'$ (such that $\mathcal{F}f\in L^1_{\mathrm{loc}}$) and $s\ge0$, we write $|\nabla|^s f = \mathcal{F}^{-1}(|\xi|^s \mathcal{F} f)$. 

\begin{prop}[Sobolev inequality in Lorentz spaces]\label{pro:sobolev-lorentz}
Assume that $0<\frac{s}{d}<\frac{1}{p}<1$
and~$1\leq q\leq\infty$.
If $\frac{1}{p}-\frac{s}{d}=\frac{1}{p_0}$~,
then
\[
\|u\|_{L^{p_0,q}} \lesssim \| |\nabla|^s u \|_{L^{p,q}}\,.
\]
In particular
\[
\|u\|_{L^{p_0,p}} \lesssim \| u \|_{H^{s,p}}\,.
\]
\end{prop}
\begin{prop}[Hölder inequality in Lorentz spaces]\label{pro:holder-lorentz}
    If $1<p,p_1,p_2<\infty$, $1\leq q,q_1,q_2 \leq \infty$, $\frac{1}{p}+\frac{1}{p'}=1$, $\frac{1}{q}+\frac{1}{q'}=1$, $\frac{1}{p_1}+\frac{1}{p_2}=\frac{1}{p}$, and $\frac{1}{q_1}+\frac{1}{q_2}=\frac{1}{q}$, then
    \begin{align*}
        \|f_1 f_2\|_{L^{p,q}} & \lesssim \|f_1\|_{L^{p,q}} \|f_2\|_{L^{\infty}}\\
        \|f_1 f_2\|_{L^{1}} & \lesssim \|f_1\|_{L^{p,q}} \|f_2\|_{L^{p',q'}} \\
        \|f_1 f_2\|_{L^{p,q}} & \lesssim \|f_1\|_{L^{p_1,q_1}} \|f_2\|_{L^{p_2,q_2}}
    \end{align*}
    provided that the right-hand sides are finite.
\end{prop}
\begin{prop}[Young inequality in Lorentz spaces]\label{pro:young-lorentz}
    If $1<p,p_1,p_2<\infty$, $1\leq q,q_1,q_2 \leq \infty$, $\frac{1}{p}+\frac{1}{p'}=1$, $\frac{1}{q}+\frac{1}{q'}=1$, $\frac{1}{p_1}+\frac{1}{p_2}=\frac{1}{p}+1$, $\frac{1}{q_1}+\frac{1}{q_2}=\frac{1}{q}$, then
    \begin{align*}
    \|f_1* f_2\|_{L^{p,q}} &\lesssim \|f_1\|_{L^{p,q}} \|f_2\|_{L^{1}} \\
    \|f_1* f_2\|_{L^{\infty}} &\lesssim \|f_1\|_{L^{p,q}} \|f_2\|_{L^{p',q'}} \\
    \|f_1 *f_2\|_{L^{p,q}} &\lesssim \|f_1\|_{L^{p_1,q_1}} \|f_2\|_{L^{p_2,q_2}}
    \end{align*}
    provided that the right-hand sides are finite.
\end{prop}

In the next proposition, we remark that the Young inequalities in Lebesgue or Lorentz spaces involving the Lebesgue space $L^1$ can be generalized to the space $\meas$ of finite, signed Radon measures (recall that $\meas$ is equipped with the total variation norm $\|\mu\|_{\meas }=|\mu|(\mathbb{R}^d)$).
\begin{prop}\label{pro:young-lorentz-measure}
    If $1<p<\infty$, $1\leq q \leq \infty$, then
    \begin{align*}
    \|f_1* f_2\|_{L^{1}} 
    &\lesssim 
    \|f_1\|_{L^{1}} 
    \|f_2\|_{\meas}\,,\\
    \|f_1* f_2\|_{L^{\infty}} 
    &\lesssim 
    \|f_1\|_{L^{\infty}} 
    \|f_2\|_{\meas}\,,\\
    \|f_1* f_2\|_{L^{p,q}} 
    &\lesssim 
    \|f_1\|_{L^{p,q}} 
    \|f_2\|_{\meas}\,,
    \end{align*}
provided that the right-hand sides are finite.
\end{prop}

\begin{proof}
If $f_1\in L^{1}$ and $f_2 \in\meas$, then $(f_1*f_2)(x)$ is defined almost everywhere and defines an element of $L^{1}$. Indeed $\iint|f_1(x-y)|\diff x\,\mathrm{d}|f_2|(y)=\int\|f_1\|_{L^{1}}\,\diff |f_2|(y)=\|f_1\|_{L^{1}}\,|f_2|(\mathbb{R}^{d})<\infty$, hence the Fubini theorem applies and 
\[
\iint|f_1(x-y)|\diff |f_2|(y)\,\diff x<\infty.
\]
This implies that $y\mapsto|f_1(x-y)|$ is $|f_2|$-integrable for almost every $x$. We can thus define the function
\begin{equation}
(f_1*f_2)(x)=\int f_1(x-y)\,\diff f_2(y) \label{eq:def-conv-measure}
\end{equation}
almost everywhere. Using again the Fubini theorem, this function is in $L^1$. 

Actually~\eqref{eq:def-conv-measure} also makes sense when $f_1$ is in $L^\infty$, and the inequality $\|f_1*f_2\|_{L^\infty}\leq \|f_1\|_{L^\infty} \|f_2\|_{\meas}$ holds. Moreover if $f_1$ is in $L^\infty \cap L^1$ both definitions of $f_1*f_2$ coincide. Hence, for $f_2$ in $\meas$, the operator $\Phi_{f_2}$ defined on $L^1+L^\infty$ by $\Phi_{f_2}(f_1)=f_1*f_2$ is well defined, continuous from $L^1$ to $L^1$ and from $L^\infty$ to $L^\infty$ and thus, by real interpolation, it defines a continous operator from $L^{p,q}$ to $L^{p,q}$ for $1<p<\infty$ and $1\leq q\leq \infty$, with an operator norm smaller than $C_{p,q}\|f\|_{\meas}$. This implies the stated Young inequality.
\end{proof}

To conclude this subsection we recall a result about derivatives of convolution products:
\begin{prop}[Young inequality in Bessel-Sobolev spaces]
    If $0\leq s < \infty$, $1\leq p_j \leq \infty$ for $j$ in $\{0,1,2\}$, and $1+\frac1{p_0}=\frac1{p_1}+\frac1{p_2}$ then
    \begin{align}
    \|f_1* f_2\|_{H^{s,p_0}} 
    &\lesssim \|f_1\|_{\meas} \|f_2\|_{H^{s,p_0}} \label{eq:conv-M-Hsp}\\
    \|f_1* f_2\|_{H^{s,p_0}} 
    & \lesssim \|f_1\|_{L^{p_1}} \|f_2\|_{H^{s,p_2}}
    \end{align}
    provided that the right-hand sides are finite.
    
    If, moreover, $1< p_j < \infty$ for $j$ in $\{0,1,2\}$, then
    \begin{equation}
    \|f_1* f_2\|_{H^{s,p_0}} 
     \lesssim \|f_1\|_{L^{p_1,\infty}} \|f_2\|_{H^{s,p_2}}\,.
    \end{equation}
\end{prop}

\begin{proof}Using the definition of Bessel-Sobolev spaces, the property that the Fourier transform of a convolution product is a constant times the product of the Fourier transforms, we deduce that
   \begin{align*}
       \|f_1*f_2\|_{H^{s,p_0}}
       &= \|\mathcal{F}^{-1}\langle\xi\rangle^s\mathcal{F}(f_1*f_2)\|_{L^{p_0}}\\ \allowdisplaybreaks
       &=(2\pi)^{\frac{d}{2}} \|\mathcal{F}^{-1}(\mathcal{F}(f_1)\langle\xi\rangle^s\mathcal{F}(f_2))\|_{L^{p_0}} \\ \allowdisplaybreaks
       &= \|f_1*(\jnabla^sf_2)\|_{L^{p_0}}.
   \end{align*}
It just remains to apply the Young inequalities in Propositions~\ref{pro:young-lorentz} and \ref{pro:young-lorentz-measure} to obtain the results stated above.
\end{proof}

\subsection{Multilinear estimates}

To handle the non-linearity in the pseudo-relativistic Hartree equation \eqref{eq:hartree_semi_intro}, we will use several times the following estimates.

\begin{lemma}[Multilinear estimates I]\label{lem:bounds-forcing-short-range I}
Assume $s\geq0$, $p>2$ and $1< q\leq \infty$.
Let $w_{1}\in\meas$ and $w_{q}\in L^{q,m}$ with $m \in \{q, \infty\}$.
For $u_j\in H^s\cap L^\sigma$ 
(or $u_j\in H^s\cap L^{\sigma,\alpha}$ for the last inequality),
 $j\in  \mathbb{Z}_3 = \mathbb{Z}/(3\mathbb{Z})$, the following inequalities hold:
\begin{align}
\|(w_{1}*(u_{0}u_{1}))u_{2}\|_{H^{s}} &  \lesssim\|w_{1}\|_{\meas}\sum_{j\in\mathbb{Z}_{3}}\|u_{j}\|_{H^{s}}\|u_{j+1}\|_{L^{\sigma}}\|u_{j+2}\|_{L^{\sigma}} \, \textnormal{with} \, \sigma = \infty \,, \, s>\frac{d}{2}\,,  \label{eq:convolution-Hartree-Hs-M1}\\
 \allowdisplaybreaks
\|(w_{1}*(u_{0}u_{1}))u_{2}\|_{H^{s,p'}} & \lesssim\|w_{1}\|_{\meas}\sum_{j\in\mathbb{Z}_{3}}\|u_{j}\|_{H^{s}} \|u_{j+1}\|_{L^{\sigma}} \|u_{j+2}\|_{L^{\sigma}} \, \textnormal{with} \, \frac{1}{p'} = \frac{1}{2}+\frac{2}{\sigma}\,, \label{eq:convolution-Hartree-Hspprime-M1}\\
\|(w_{q}*(u_{0}u_{1}))u_{2}\|_{H^{s,p'}} & \lesssim\|w_{q}\|_{L^{q,\infty}}\sum_{j\in\mathbb{Z}_{3}}\|u_{j}\|_{H^{s}} \|u_{j+1}\|_{L^{\sigma}} \|u_{j+2}\|_{L^{\sigma}} \, \textnormal{with} \, \frac{1}{p'}+\frac{1}{2}=\frac{1}{q}+\frac{2}{\sigma}, \label{eq:convolution-Hartree-Hspprime-Lq}\\
\|(w_{q}*(u_{0}u_{1}))u_{2}\|_{H^{s}} &  \lesssim\|w_{q}\|_{L^{q,m}}\sum_{j\in\mathbb{Z}_{3}}\|u_{j}\|_{H^{s}} \|u_{j+1}\|_{L^{\sigma,\alpha}} \|u_{j+2}\|_{L^{\sigma,\alpha}}  \, \textnormal{with} \begin{cases}
    1=\frac{1}{q}+\frac{2}{\sigma}\\
    1=\frac{1}{m}+\frac{2}{\alpha} . 
\end{cases}\label{eq:convolution-Hartree-Hs-Lq}
\end{align}
\end{lemma}

\begin{proof}Let us first remark that if $u$ is in $H^s$ with $s>d/2$, then $u$ is in $\mathcal{C}^0_\infty$ which is the closure of $\mathcal{S}$ for the $L^\infty$ norm. This will be important for our applications of the fractional Leibniz rule with $L^\infty$. 
For $u_j$, $u_{j+1}$ in $H^s$, as $H^s$ is an algebra for $s>d/2$, $u_ju_{j+1}$ is also in $H^s$ and~\eqref{eq:conv-M-Hsp} shows that $w_1*(u_ju_{j+1})$ lies in $H^s$ which, again, is included in $\mathcal{C}^0_\infty$. This allows us to apply the fractional Leibniz rule with $L^\infty$ below.

The fractional Leibniz rule, along with Young and H\"older's inequalities, yield both (\ref{eq:convolution-Hartree-Hs-M1}) and~(\ref{eq:convolution-Hartree-Hspprime-M1}) by considering $p\geq 2$:

\begin{align*}
 \|(w_{1}*(u_{0}u_{1}))\,u_{2}\|_{H^{s,p'}} & \leq\|w_{1}*(u_{0}u_{1})\|_{H^{s,\frac{1}{\frac{1}{\sigma}+\frac{1}{2}}}}\|u_{2}\|_{L^{\sigma}}+\|w_{1}*(u_{0}u_{1})\|_{L^{\frac{\sigma}{2}}}\|u_{2}\|_{H^{s}}\\
 & \leq\|w_{1}\|_{\meas}\big(\|u_{0}u_{1}\|_{H^{s,\frac{1}{\frac{1}{\sigma}+\frac{1}{2}}}}\|u_{2}\|_{L^{\sigma}}+\|u_{0}u_{1}\|_{L^{\frac{\sigma}{2}}}\|u_{2}\|_{H^{s}}\big)\\
 & \leq\|w_{1}\|_{\meas}\sum_{j\in\mathbb{Z}_{3}}\|u_{j}\|_{H^{s}}\|u_{j+1}\|_{L^{\sigma}}\|u_{j+2}\|_{L^{\sigma}}\,.
\end{align*}

For~\eqref{eq:convolution-Hartree-Hspprime-Lq}, the same inequalities as above yield
\begin{align*}
  \|(w_{q}*(u_{0}u_{1}))\,u_{2}\|_{H^{s,p'}}
 & \leq\|w_{q}*(u_{0}u_{1})\|_{H^{s,\frac{1}{\frac{1}{\sigma}+\frac{1}{q}-\frac{1}{2}}}}\|u_{2}\|_{L^{\sigma}}+\|w_{q}*(u_{0}u_{1})\|_{L^{\frac{1}{\frac{2}{\sigma}+\frac{1}{q}-1}}}\|u_{2}\|_{H^{s}}\\
 & \leq\|w_{q}\|_{L^{q,\infty}}\big(\|u_{0}u_{1}\|_{H^{s,\frac{1}{\frac{1}{\sigma}+\frac{1}{2}}}}\|u_{2}\|_{L^{\sigma}}+\|u_{0}u_{1}\|_{L^{\frac{\sigma}{2}}}\|u_{2}\|_{H^{s}}\big)\\
 & \leq\|w_{q}\|_{L^{q,\infty}}\sum_{j\in\mathbb{Z}_{3}}\|u_{j}\|_{H^{s}}\|u_{j+1}\|_{L^{\sigma}}\|u_{j+2}\|_{L^{\sigma}}\,.
\end{align*}

For~\eqref{eq:convolution-Hartree-Hs-Lq}, the fractional Leibniz rule, along with Young and H\"older's inequalities in Lorentz spaces, yield
\begin{align*}
  \|(w_{q}*(u_{0}u_{1}))\,u_{2}\|_{H^{s}}
 & \leq\|w_{q}*(u_{0}u_{1})\|_{H^{s,\frac{1}{\frac{1}{\sigma}+\frac{1}{q}-\frac{1}{2}}}}\|u_{2}\|_{L^{\sigma}}+\|w_{q}*(u_{0}u_{1})\|_{L^{\infty}}\|u_{2}\|_{H^{s}}\\
 & \leq\|w_{q}\|_{L^{q,m}}\big(\|u_{0}u_{1}\|_{H^{s,\frac{1}{\frac{1}{\sigma}+\frac{1}{2}}}}\|u_{2}\|_{L^{\sigma}}+\|u_{0}u_{1}\|_{L^{\frac{\sigma}{2},\frac{\alpha}{2}}}\|u_{2}\|_{H^{s}}\big)\\
 & \leq\|w_{q}\|_{L^{q,m}}\sum_{j\in\mathbb{Z}_{3}}\|u_{j}\|_{H^{s}}\|u_{j+1}\|_{L^{\sigma,\alpha}}\|u_{j+2}\|_{L^{\sigma,\alpha}}\,.
\end{align*}
This concludes the proof.
\end{proof}

We will also need the following estimates in weighted Sobolev spaces.

\begin{lemma}[Weighted multilinear estimates I]
\label{lemma: weight multilin}
Assume $0\le \gamma \le 2$, $s\ge 0$ and $1 \leq q <\infty$. Let $w _1\in \meas $ and $w_q \in L^q$. For $ u   \in H^s_\gamma \cap L^\sigma_\gamma$ it holds
\begin{align}
    \|(w_1*|u|^2)u \|_{
L^2_\gamma} \lesssim & \|w_1\|_{\meas} \|u\|_{L^\infty} ^2 \|u \|_{L^2_\gamma}, \label{eq: weighted multilin M1 in L2}\\
\|(w_q*|u|^2)u \|_{L^2_\gamma} \lesssim & \|w_q\|_{L^q}
    \|u\|_{L^{2q'}}^2 \|u\|_{L^2_\gamma}, \label{eq: weighted multilin Lq in L2}
    \end{align}
as well as
\begin{align}
\|(w_1*|u|^2)u \|_{H^s_\gamma} \lesssim & \|w_1\|_{\meas} \|u\|_{L^\infty} (\|u\|_{L^\infty}  \|u \|_{H^s_\gamma}  + \|u\|_{H^s}\|u \|_{L^\infty_\gamma})\quad \textrm{if}\quad s>\frac{d}{2}, \label{eq: weighted multilin M in Hs}\\
      \|(w_q*|u|^2)u \|_{H^s_\gamma} \lesssim & \|w_q\|_{L^q}
    \|u\|_{L^{2q'}}(\|u\|_{L^{2q'}}\| u \|_{H^s_\gamma} + \|u \|_{H^s}\| u \|_{L^{2q'}_\gamma}). \label{eq: weighted multilin Lq in Hs}
          \end{align}

\end{lemma}

\begin{proof}
The first two inequalities are a direct consequence of H\"older and Young inequalities. 
Applying the weighted fractional Leibniz rule as well as Young inequality we have 
\begin{align*}
    \|(w_1*|u|^2)u \|_{H^{s}_\gamma} \lesssim & \|w_1 *|u|^2\|_{L^\infty} \| u \|_{H^{s}_\gamma} + \|w_1*|u|^2\|_{H^s} \| u \|_{L^\infty_\gamma}\\
\lesssim & \|w_1\|_{\meas} \|u\|_{L^\infty} (\|u\|_{L^\infty}  \|u \|_{H^s_\gamma}  + \|u\|_{H^s}\|u \|_{L_\gamma^\infty}),
\end{align*}
which proves the third inequality. 
  We prove similarly the last inequality. Let $\sigma\ge 2$ such that $1 = \frac{1}{q}+ \frac{2}{\sigma}$, we write
    \begin{align*}
        \|(w_q*|u|^2)u \|_{H_\gamma^{s}} \lesssim & \, \|w_q*|u|^2\|_{L^{\infty} }\| u \|_{H_\gamma^s}\\
        &+ \|w_q* (|u|^2)\|_{H^{s,\frac{1}{\frac{1}{\sigma} + \frac{1}{q} - \frac12}}} \|u \|_{L_\gamma^{\sigma}}\\
        \lesssim & \, \|w_q\|_{L^q} (\|u\|^2_{L^{\sigma}} \| u \|_{H^s_\gamma}+ \| |u|^2\|_{H^{s,\frac{1}{\frac{1}{\sigma} + \frac12}}} \| u \|_{L_\gamma^{\sigma}})\\
        \lesssim & \,\|w_q\|_{L^q} \|u\|_{L^{\sigma}} (\|u\|_{L^{\sigma}}\| u \|_{H_\gamma^s} + \|u \|_{H^s}\| u \|_{L_\gamma^{\sigma}}).
    \end{align*}
This concludes the proof.
\end{proof}

In order to study the global existence and asymptotic properties of solutions to \eqref{eq:hartree_semi_intro} in the case where $w$ is short-range (see Sections \ref{sec:global-short} and \ref{sec:min-vel}), we will need slightly refined versions of the previous estimates that we state in the following two lemmas.

\begin{lemma}[Multilinear estimates II]\label{lem:bounds-forcing-short-range}
Assume  $s\geq0$, $p\geq2$, $q,r\geq1$ such that $q\leq2r$, $\frac{2}{p}+\frac{1}{r}=1$, $q\neq\infty$
and $w_{1}\in\meas$, $w_{q}\in L^{q}$. For $u_j\in H^s\cap L^\infty$, $j\in  \mathbb{Z}_3 = \mathbb{Z}/(3\mathbb{Z})$, the inequalities
\begin{align}
\|(w_{q}*(u_{0}u_{1}))u_{2}\|_{H^{s}}  \lesssim\|w_{q}\|_{L^{q}}\sum_{j\in\mathbb{Z}_{3}} &\|u_{j}\|_{H^{s}}\|u_{j+1}\|_{H^s}^{1-\theta(\frac{r}{q})}\|u_{j+2}\|_{H^s}^{1-\theta(\frac{r}{q})} \nonumber\\
&\|u_{j+1}\|_{L^{p}\cap L^{\infty}}^{\theta(\frac{r}{q})}\|u_{j+2}\|_{L^{p}\cap L^{\infty}}^{\theta(\frac{r}{q})} , \label{eq:convolution-Hartree-Hs-Lq-theta}\\
\|(w_{1}*(u_{0}u_{1}))u_{2}\|_{H^{s,p'}}  \lesssim\|w_{1}\|_{\meas}\sum_{j\in\mathbb{Z}_{3}}&\|u_{j}\|_{H^{s}}\|u_{j+1}\|_{H^{s}}^{1-\theta(r-\frac{1}{2})} \|u_{j+2}\|_{H^{s}}^{1-\theta(r-\frac{1}{2})}\nonumber\\
&\|u_{j+1}\|_{L^{p}\cap L^{\infty}}^{\theta(r-\frac{1}{2})}\|u_{j+2}\|_{L^{p}\cap L^{\infty}}^{\theta(r-\frac{1}{2})}, \label{eq:convolution-Hartree-Hspprime-M1-theta}\\
\|(w_{q}*(u_{0}u_{1}))u_{2}\|_{H^{s,p'}}  \lesssim\|w_{q}\|_{L^{q,\infty}}\sum_{j\in\mathbb{Z}_{3}}&\|u_{j}\|_{H^{s}}\|u_{j+1}\|_{H^{s}}^{1-\theta(\frac{r}{q}-\frac{1}{2})} \|u_{j+2}\|_{H^{s}}^{1-\theta(\frac{r}{q}-\frac{1}{2})}\nonumber\\
&\|u_{j+1}\|_{L^{p}\cap L^{\infty}}^{\theta(\frac{r}{q}-\frac{1}{2})}\|u_{j+2}\|_{L^{p}\cap L^{\infty}}^{\theta(\frac{r}{q}-\frac{1}{2})}\label{eq:convolution-Hartree-Hspprime-Lq-theta}
\end{align}
hold, with $\theta(x)=\min\{x,1\}$.
\end{lemma}

\begin{proof}
Eq.~(\ref{eq:convolution-Hartree-Hs-Lq-theta}) follows from \eqref{eq:convolution-Hartree-Hs-Lq} with $m=q$. This implies $\alpha = \sigma$, 
where $\sigma$ is defined by $1-\frac{1}{q}=\frac{2}{\sigma}$. Hence, 
\begin{itemize}
\item if $r\leq q$, then $2\leq\sigma\leq p$ and we set
$
\frac{1}{\sigma}=\frac{\theta}{p}+\frac{1-\theta}{2}\Leftrightarrow\theta=\frac{r}{q}\,,
$
so that $\|u\|_{L^{\sigma}}\leq\|u\|_{L^{p}}^{\theta}\|u\|_{L^{2}}^{1-\theta}$,
\item else $r>q$, and then $p<\sigma$ and $\|u\|_{L^{\sigma}}\leq\|u\|_{L^{p}\cap L^{\infty}}$.
\end{itemize}

Eq.~(\ref{eq:convolution-Hartree-Hspprime-M1-theta}), follows from~\eqref{eq:convolution-Hartree-Hspprime-M1} where $\sigma$ is defined through $1+\frac{1}{p'}=1+\frac{1}{2}+\frac{2}{\sigma}$ observing that
\begin{itemize}
\item If $r\leq \frac32$ then $2\leq\sigma\leq p$ and we set 
$\theta=r-\frac12$,
which gives $\|u\|_{L^{\sigma}}\leq\|u\|_{L^{p}}^{\theta}\|u\|_{L^{2}}^{1-\theta}$,
\item else $3/2<r$ and then $p\leq\sigma$ and $\|u\|_{L^{\sigma}}\leq\|u\|_{L^{p}\cap L^{\infty}}$\,.
\end{itemize}

Finally, \eqref{eq:convolution-Hartree-Hspprime-Lq-theta}, follows from \eqref{eq:convolution-Hartree-Hspprime-Lq}, with
$1+\frac{1}{p'}-\frac{1}{2}-\frac{1}{q}=\frac{2}{\sigma}$, as
\begin{itemize}
\item if $r\leq\frac{3}{2}q$, then $2\leq\sigma\leq p$ and we set
$\theta=\frac{r}{q}-\frac{1}{2}$, so that $\|u\|_{L^{\sigma}}\leq\|u\|_{L^{p}}^{\theta}\|u\|_{L^{2}}^{1-\theta}$,
\item else $\frac{3}{2}q<r$, and then $p\leq\sigma$ and $\|u\|_{L^{\sigma}}\leq\|u\|_{L^{p}\cap L^{\infty}}$.\qedhere
\end{itemize}
\end{proof}

\begin{lemma}[Weighted multilinear estimates II]
\label{lemma: weight multilin 2}
Assume $s\ge 0$, $0 \leq \gamma \leq 2$, $p\geq2$, $q,r \ge 1$, such that $q\leq 2r$, $\frac{2}{p}+\frac{1}{r}=1$
and $w_{1}\in\meas$, $w_{q}\in L^{q}$. For $u\in H^s_\gamma$, it holds
\begin{align}
\|(w*|u|^2)u \|_{L^2_\gamma} \lesssim & \|w\|_{\meas + L^q} \|u\|_{L^2_\gamma} \|u\|_{L^p\cap L^\infty}^{2\theta (\frac{r}{q})} (\|u\|_{L^2}^{2-2\theta (\frac{r}{q})} + \|u\|_{L^\infty}^{2-2\theta (\frac{r}{q})}), \label{eq:forcing-short-range-L1+Lq-weight-L2}
\end{align}
where $w = w_1 + w_q$ and 
\begin{align}
  \|  (w_q*|u|^2)u \|_{H^s_\gamma}\lesssim & \|w_q\|_{L^q} \|u\|_{H^s}^{2-2\theta (\frac{r}{q})} \|u\|_{L^p\cap L^\infty}^{\theta (\frac{r}{q})}\nonumber\\
  &( \|u\|_{L^p\cap L^\infty}^{\theta (\frac{r}{q})} \|u\|_{H^s_\gamma} + \|u\|_{H^s}^{\theta (\frac{r}{q})} \|u\|_{L^{2q'}_\gamma}), \label{eq:forcing-short-range-Leb-Hs-to-weight-Hs}
\end{align}
with $\theta(x)=\min\{x,1\}$. 
\end{lemma}

\begin{proof}
  To prove the inequalities we reason analogously to the proof of Lemma \ref{lem:bounds-forcing-short-range}, by using the bound $\|u\|_{L^\sigma}\leq \|u\|_{L^2}^{1-\theta (\frac{r}{q})} \|u\|_{L^p\cap L^\infty}^{\theta (\frac{r}{q})}$ with $\sigma = 2q'$, as in the proof of \eqref{eq:convolution-Hartree-Hs-Lq-theta}. Applying this to \eqref{eq: weighted multilin M1 in L2}-\eqref{eq: weighted multilin Lq in L2} or \eqref{eq: weighted multilin Lq in Hs} we obtain, respectively,  \eqref{eq:forcing-short-range-L1+Lq-weight-L2} and \eqref{eq:forcing-short-range-Leb-Hs-to-weight-Hs}. 
\end{proof}

	\section{Local existence}\label{sec:local}
	
	In this section we prove the local existence results stated in Theorems \ref{th:local-long-range-intro} and \ref{th: local_with_Strichartz-intro}. We actually establish more precise and more general versions of these theorems, stated in Theorems \ref{th: local H^s} and \ref{th: local_with_Strichartz}, respectively. The proofs rely on a standard fixed point argument in $L^\infty_tH^s_x$ for the first result, and in addition the Strichartz estimates of Proposition \ref{prop:Strichartz-Lorentz} for the second result. We give some details for the sake of completeness. 
	
    For $s\ge0$ and $T>0$, we say that $\psi$ is a solution to \eqref{eq:hartree_semi_intro}  on the time interval $[0,T]$ if $\psi$ belongs to $\mathcal{C}^{0}([0,T],H^{s})\cap \mathcal{C}^{1}([0,T],H^{s-1})$ and  satisfies \eqref{eq:hartree_semi_intro} in~$H^{s-1}$.
        In particular, a function~$\psi$ in $\mathcal{C}^{0}([0,T],H^{s})\cap \mathcal{C}^{1}([0,T],H^{s-1})$ is a solution to \eqref{eq:hartree_semi_intro}  if and only if it satisfies $\partial_t\big(e^{it \jnabla}\psi_t\big)=-ie^{it \jnabla}\big((w*|\psi_t|^2)\psi_t\big)$ in $H^{s-1}$ and thus if and only if it satisfies the Duhamel equation
\begin{equation}
\psi_{t}=e^{-it\langle\nabla\rangle}\psi_{0}-i\int_{0}^{t}e^{i(\tau-t)\langle\nabla\rangle}(w*|\psi_\tau|^{2})\psi_\tau\,\mathrm{d}\tau\label{eq:Hartree-Duhamel}
\end{equation}
in $\mathcal{C}^0([0,T],H^s)$.

We introduce the notation $X_T:=\mathcal{C}^0([0,T],H^s)$, endowed with the norm	
	\[
\|\psi\|_{X_{T}}:=\sup_{0\leq t\leq T}\|\psi_{t}\|_{H^s}.
\]
Moreover, for $w\in\mathcal{W}_{d,s}$ (see \eqref{eq:def_Wds}) we set
\begin{equation}
\|w\|:=
\begin{cases}
    \|w\|_{L^{\frac{d}{2s},\infty}+L^\infty} &\text{ if } s<\frac{d}{2},\\
     \inf_{q>1:w\in L^q+L^\infty}\|w\|_{L^q+L^\infty} &\text{ if } s= \frac{d}{2},\\
    \|w\|_{\meas+L^\infty} & \text{ if } s> \frac{d}{2}.
\end{cases}
\end{equation}

The next result readily implies Theorem \ref{th:local-long-range-intro}.

\begin{theo}
\label{th: local H^s} Let $s\geq0$ and $w$ in $\setW$. For every initial datum~$\psi_{0}$ in~$H^{s}$, there exists~$T_{\max}$ in~$(0,\infty]$ such that the Duhamel equation \eqref{eq:Hartree-Duhamel}
associated with \eqref{eq:hartree_semi_intro} admits a solution $\psi$ in~$\mathcal{C}^{0}([0,T_{\max}),H^{s})\cap \mathcal{C}^{1}([0,T_{\max}),H^{s-1})$ where either~$T_{\max}=\infty$ or~$\lim_{t\to T_{\max}}\|\psi_t\|_{H^s} = \infty$.

If~$T>0$ and~$\psi$,~$\tilde{\psi}$ are two solutions in~$X_{T}$ of \eqref{eq:Hartree-Duhamel} with initial data~$\psi_0$ and~$\tilde{\psi}_0$ in~$H^s$, then, for some $C_d>0$,
\begin{equation}\label{eqn:comparaison-solutions-Hartree}
    \|\psi-\tilde{\psi}\|_{X_{T}}
    \leq \|\psi_{0}-\tilde{\psi}_{0}\|_{H^{s}}\exp\big(C_d\|w\|(\|\psi\|_{X_{T}}+\|\tilde{\psi}\|_{X_{T}})^{2}\big).
\end{equation}
In particular, for any $0<T<T_{\max}$, the solution $\psi$ to \eqref{eq:hartree_semi_intro} on $[0,T]$ associated to the initial datum $\psi_0$ in $H^s$ is unique, and the map
\begin{equation*}
    H^s\ni\psi_0\mapsto\psi\in\mathcal{C}^{0}([0,T],H^{s})\cap \mathcal{C}^{1}([0,T],H^{s-1})
\end{equation*}
is continuous.
\end{theo}

\begin{proof}
We first show the existence of a  solution for short times $T$ by a
standard contraction argument. Let $T>0$. The operator $\langle\nabla\rangle$
is selfadjoint with form domain $H^{1/2}$ and hence generates a strongly
continuous unitary group $e^{-it\langle\nabla\rangle}$ acting on $L^{2}$. For~$\psi_{0}$ in~$H^{s}$
and~$\psi$ in $\mathcal{C}^{0}([0,T),H^{s})$ we write 
\[
\label{eq: integral op}
\psi^{(0)}_t=e^{-it\langle\nabla\rangle}\psi_{0}\,,\quad I(t)=\int_{0}^{t}e^{i(\tau-t)\langle\nabla\rangle}(w*|\psi(\tau)|^{2})\psi(\tau)\,\mathrm{d}\tau\,,\quad\mathcal{T}(\psi)=\psi^{(0)}-iI(\psi)\,.
\]
The Duhamel formula \eqref{eq:Hartree-Duhamel} then reads $\psi=\mathcal{T}(\psi)\,.$
The proof relies on usual contraction arguments.

We consider first the case $0\leq s<\frac{d}{2}~$. We define $\sigma_{q}$
by $1=\frac{1}{q}+\frac{2}{\sigma_{q}}$. For $q\in\{\infty,\frac{d}{2s}\}$,
as $\frac{1}{2}=\frac{s}{d}+\frac{1}{\sigma_{d/(2s)}}$~, we get
$\|u\|_{L^{\sigma_{q},2}}\lesssim\|u\|_{H^{s}}$ and thus, using \eqref{eq:convolution-Hartree-Hs-Lq},
we get, for $u_{j}$ in $H^{s}$, $j\in\mathbb{Z}_{3}=\mathbb{Z}/(3\mathbb{Z})$,
\begin{align}
\|(w*(u_{0}u_{1}))u_{2}\|_{H^{s}} & \leq\sum_{q\in\{\infty,\frac{d}{2s}\}}\|w\|_{L^{q,\infty}}\sum_{j\in\mathbb{Z}_{3}}\|u_{j}\|_{H^{s}}\|u_{j+1}\|_{L^{\sigma_{q},2}}\|u_{j+2}\|_{L^{\sigma_{q},2}}\nonumber \\
 & \lesssim\|w\|_{L^{\frac{d}{2s},\infty}+L^{\infty}}\prod_{j\in\mathbb{Z}_{3}}\|u_{j}\|_{H^{s}}\,.\label{eq:bound-non-linear-term-in-Hs-with-Ld/2s-Linfty}
\end{align}
Thus, for $R=\|\psi^{(0)}\|_{X_T}=\|\psi_0\|_{H^s}$ and $\psi$ in $\overline{B(\psi^{(0)},R)}$ in $X_{T}$, 
\eqref{eq:bound-non-linear-term-in-Hs-with-Ld/2s-Linfty} yields
\begin{align*}
\|\mathcal{T}(\psi)_{t}-\psi_{t}^{(0)}\|_{H^{s}}  =\|I(\psi)_{t}\|_{H^{s}}
 & \leq\int_{0}^{t}\|(w*|\psi_{\tau}|^{2})\psi_{\tau}\|_{H^{s}}\mathrm{d}\tau\\
 & \lesssim T \, \|w\|_{L^{\frac{d}{2s},\infty}+L^{\infty}} \sup_{0\leq \tau\leq T}\|\psi_{\tau}\|_{H^{s}}^{3}\\
 & \lesssim TR^{3}\|w\|_{L^{\frac{d}{2s},\infty}+L^{\infty}}\,.
\end{align*}
Hence, for $T$ sufficiently small, $\mathcal{T}$ sends $\overline{B(\psi^{(0)},R)}$
into itself.

Consider $\psi_{1}$ and $\psi_{2}$ in $\overline{B(\psi^{(0)},R)}$,
then again by \eqref{eq:bound-non-linear-term-in-Hs-with-Ld/2s-Linfty},
\begin{align*}
\|\mathcal{T}(\psi_{1})_{t}-\mathcal{T}(\psi_{2})_{t}\|_{H^{s}} & \leq\|I(\psi_{1})_{t}-I(\psi_{2})_{t}\|_{H^{s}}\\
 & \leq\int_{0}^{t}\|(w*|\psi_{1,\tau}|^{2})\psi_{1,\tau}-(w*|\psi_{2,\tau}|^{2})\psi_{2,\tau}\|_{H^{s}}\mathrm{d}\tau\\
 & \lesssim T\|w\|_{L^{\frac{d}{2s},\infty}+L^{\infty}}\|\psi_{1,\tau}-\psi_{2,\tau}\|_{H^{s}}(\|\psi_{1,\tau}\|_{H^{s}}+\|\psi_{2,\tau}\|_{H^{s}})^{2}\\
 & \lesssim TR^{2}\|w\|_{L^{\frac{d}{2s},\infty}+L^{\infty}}\|\psi_{1,\tau}-\psi_{2,\tau}\|_{H^{s}}\,.
\end{align*}
Hence, for $T$ sufficiently small, $\mathcal{T}$ is strictly contractive
on $\overline{B(\psi^{(0)},R)}$.

Similarly, for $s = \frac{d}{2}$ by assumption $ w\in L^{q, \infty} + L^\infty, q>1.$ We write $q = \frac{d}{2v}$ for some $v<\frac{d}{2}$. By definition of $v$ and  $\sigma_q$ (which we recall is such that $1=\frac{1}{q}+\frac{2}{\sigma_{q}}$), for $q\in\{\infty,\frac{d}{2v}\}$ it holds $\|u\|_{L^{\sigma_q,2}}\lesssim \|u\|_{H^v}\lesssim \|u\|_{H^{\frac{d}{2}}}$. Hence applying \eqref{eq:convolution-Hartree-Hs-Lq} again for $q\in\{\infty,\frac{d}{2v}\}$, we have 
\begin{align}
\label{eq:bound-non-linear-term-in-Hd/2-with-Lq-Linfty}
    \|(w*(u_{0}u_{1}))u_{2}\|_{H^{ \frac{d}{2}}} \lesssim\|w\|_{L^{\frac{d}{2v},\infty}+L^{\infty}}\prod_{j\in\mathbb{Z}_{3}}\|u_{j}\|_{H^{ \frac{d}{2}}}\,
\end{align}
for some $v<\frac{d}{2}$. We can then repeat the arguments of the case $s<\frac{d}{2}$ to obtain the stability and contraction properties of $\mathcal{T}$, where $\|w\|_{L^{\frac{d}{2s},\infty}+L^{\infty}}$ will be replaced by $\|w\|_{L^{q,\infty}+L^{\infty}}, q=\frac{d}{2v}$.

Finally, for $s>d/2$, using \eqref{eq:convolution-Hartree-Hs-Lq}, \eqref{eq:convolution-Hartree-Hs-M1} and $\|u\|_{L^{\infty}}\lesssim\|u\|_{H^{s}}$
we obtain
\begin{align}
\|&(w*(u_{0}u_{1}))u_{2}\|_{H^{s}} \nonumber \\ &\leq\|w\|_{\meas}\sum_{j\in\mathbb{Z}_{3}}\|u_{j}\|_{H^{s}}\|u_{j+1}\|_{L^{\infty}}\|u_{j+2}\|_{L^{\infty}} +\|w\|_{L^{\infty}}\sum_{j\in\mathbb{Z}_{3}}\|u_{j}\|_{H^{s}}\|u_{j+1}\|_{L^{2}}\|u_{j+2}\|_{L^{2}}\nonumber \\
 & \lesssim\|w\|_{\meas+L^{\infty}}\prod_{j\in\mathbb{Z}_{3}}\|u_{j}\|_{H^{s}}\,.\label{eq:eq:bound-non-linear-term-in-Hs-with-measure-Linfty}
\end{align}
The proofs of the stability and contraction properties then go along
the same lines as for~$0\leq s<\frac{d}{2}$ but with~$\|w\|_{L^{\frac{d}{2s},\infty}+L^{\infty}}$
replaced by~$\|w\|_{\meas+L^{\infty}}$.

By standard arguments a solution~$\psi$ of~\eqref{eq:Hartree-Duhamel} can be extended to a solution (again denoted by~$\psi$) in~$\mathcal{C}^0([0,T_{\max}),H^s)$ with a ``maximal'' time of existence~$T_{\max}$, such that either $T_{\max} = \infty$, or~$\lim_{t\to T_{\max}}\|\psi_t\|_{H^s}=\infty$.

We now prove the uniqueness of the solution
and the continuity with respect to the initial data. Consider
two initial data $\psi_{0}$ and $\tilde{\psi}_{0}$ and corresponding
solutions $\psi$ and $\tilde{\psi}$ with maximal times $T_{\max}$
and $\tilde{T}_{\max}$, and $0<T<\min\{T_{\max},\tilde{T}_{\max}\}$. Then, using the Duhamel formula, \eqref{eq:bound-non-linear-term-in-Hs-with-Ld/2s-Linfty} or \eqref{eq:eq:bound-non-linear-term-in-Hs-with-measure-Linfty},
\begin{align*}
\|\psi_{t}-\tilde{\psi}_{t}\|_{H^{s}} & \lesssim\|\psi_{0}-\tilde{\psi}_{0}\|_{H^{s}}+\|w\|\int_{0}^{t}(\|\psi_{\tau}\|_{H^{s}}+\|\tilde{\psi}_{\tau}\|_{H^{s}})^{2}\|\psi_{\tau}-\tilde{\psi}_{\tau}\|_{H^{s}}\mathrm{d}\tau\\
 & \lesssim\|\psi_{0}-\tilde{\psi}_{0}\|_{H^{s}}+\|w\|(\|\psi\|_{X_{T}}+\|\tilde{\psi}\|_{X_{T}})^{2}\int_{0}^{t}\|\psi_{\tau}-\tilde{\psi}_{\tau}\|_{H^{s}}\mathrm{d}\tau\,.
\end{align*}
Grönwall's inequality then yields for~$t$ in~$[0,T]$:
\[
\|\psi_{t}-\tilde{\psi}_{t}\|_{H^{s}}\leq\|\psi_{0}-\tilde{\psi}_{0}\|_{H^{s}}\exp\big(C_d\|w\|(\|\psi\|_{X_{T}}+\|\tilde{\psi}\|_{X_{T}})^{2}\big)\,.
\]
Hence \eqref{eqn:comparaison-solutions-Hartree} holds for any $0<T<\min\{T_{\max} ,\tilde{T}_{\max}\}$.

Then it makes sense to define the maximal time of existence $T_{\max}$ of the solution with a given initial datum.

Finally we show that $\psi$ belongs to $\mathcal{C}^{1}([0,T_{\max}),H^{s-1})$.
We first remark that for any $f\in H^{s}$ we have 
\[
\partial_{t}(e^{-i\langle\nabla\rangle t}f)=-ie^{-i\langle\nabla\rangle t}\langle\nabla\rangle f
\]
is continuous and bounded as a function of $t$ with values in $H^{s-1}$,
as $\langle\nabla\rangle f\in H^{s-1}$ and $t\mapsto e^{-i\langle\nabla\rangle t}$
is strongly continuous and bounded on $H^{s-1}$. We have then found $e^{-i\langle\nabla\rangle t}f\in \mathcal{C}^{1}([0,T),H^{s-1})$.
By assumption $\psi_{0}\in H^{s}$ and therefore $e^{-i\langle\nabla\rangle t}\psi_{0}\in \mathcal{C}^{1}([0,T],H^{s-1})$.
Similarly, we know that $\psi\in \mathcal{C}^{0}([0,T],H^{s})$, hence $(w*|\psi_{\tau}|^{2})\psi_{\tau}\in H^{s}$
uniformly in $\tau\in[0,T]$ for any $T<T_{\max}$ by \eqref{eq:bound-non-linear-term-in-Hs-with-Ld/2s-Linfty}, \eqref{eq:bound-non-linear-term-in-Hd/2-with-Lq-Linfty}
or \eqref{eq:eq:bound-non-linear-term-in-Hs-with-measure-Linfty},
and as $H^{s}\subset H^{s-1}$ and
\begin{align*}
\partial_{t}\int_{0}^{t}e^{i(\tau-t)\langle\nabla\rangle}(w*|\psi_{\tau}|^{2})\psi_{\tau}\,\mathrm{d}\tau= & (w*|\psi_{t}|^{2})\psi_{t}+\int_{0}^{t}\partial_{t}(\,e^{i(\tau-t)\langle\nabla\rangle}(w*|\psi_{\tau}|^{2})\psi_{\tau}\,)\,\mathrm{d}\tau\in H^{s-1}
\end{align*}
we get that $I(\psi)\in \mathcal{C}^{1}([0,T_{\max}),H^{s-1})$. We conclude
that both terms in \eqref{eq:Hartree-Duhamel} are then in $\mathcal{C}^{1}([0,T_{\max}),H^{s-1})$
and therefore $\psi$ belongs to~$\mathcal{C}^{1}([0,T_{\max}),H^{s-1})$ as
claimed.
\end{proof}

Now we turn to the proof of  Theorem \ref{th: local_with_Strichartz-intro}. The argument is similar to that used in the previous proof, using in addition the Strichartz estimates given in Proposition \ref{prop:Strichartz-Lorentz}. We have the following result involving any admissible pair $(a,b)$ for the Strichartz estimates (either wave admissible or Schrödinger admissible). For $s<(d+1)/(2d+2)$, it implies  Theorem \ref{th: local_with_Strichartz-intro} by taking $\beta=0$, $\frac1a=\frac{(d-1)s}{d+1}$ and $\frac1b=\frac12-\frac{2s}{d+1}$, and likewise for $s\ge(d+1)/(2d+2)$.

	\begin{theo}\label{th: local_with_Strichartz}
	    Let $0\le\beta\le1$ and $(a,b)$ be an admissible pair for the Strichartz estimates of Proposition \ref{prop:Strichartz-Lorentz} and $s\ge\frac12+\frac1a-\frac1b$. Let $w\in L^{b/(b-2),\infty}+L^\infty$. For every initial datum~$\psi_{0}$ in~$H^{s}$ there exists~$T_{\max}$ in~$(0,\infty]$ such that, for all $0<T<T_{\max}$, the Duhamel equation \eqref{eq:Hartree-Duhamel}
associated with \eqref{eq:hartree_semi_intro} admits a unique solution 
\begin{equation*}
\psi\in\mathcal{C}^{0}([0,T],H^{s})\cap \mathcal{C}^{1}([0,T],H^{s-1})\cap L^a([0,T],L^{b,2}),
\end{equation*}
where either~$T_{\max}=\infty$ or~$\lim_{t\to T_{\max}}\|\psi_t\|_{H^s}+\|\psi\|_{L^a([0,T_{\max}),L^{b,2})} = \infty$. Moreover, for any $0<T<T_{\max}$, the solution $\psi$ to \eqref{eq:hartree_semi_intro} on $[0,T]$ associated to the initial datum $\psi_0$ in $H^s$ is unique, and the map
\begin{equation*}
    H^s\ni\psi_0\mapsto\psi\in\mathcal{C}^{0}([0,T],H^{s})\cap \mathcal{C}^{1}([0,T],H^{s-1})\cap L^a([0,T],L^{b,2})
\end{equation*}
is continuous.
	\end{theo}
	
	\begin{proof}
	As in the proof of the previous theorem, it suffices to solve the Duhamel equation \eqref{eq:Hartree-Duhamel} using a fixed point argument. More precisely one verifies that the map $\mathcal{T}$ defined in the proof of Theorem \ref{th: local H^s} is a contraction in a suitable ball contained in the space
	\begin{equation*}
	    Y_T:=\mathcal{C}^{0}([0,T],H^{s})\cap L^a([0,T],L^{b,2})
	\end{equation*}
	endowed with the norm
	\begin{equation*}
	    \|\psi\|_{Y_T}:=\sup_{0\le t\le T}\|\psi_t\|_{H_s}+\Big[\int_0^T\|\psi_\tau\|^a_{L^{b,2}}d\tau\Big]^{1/a}.
	\end{equation*}
	Using the Strichartz estimates from Proposition \ref{prop:Strichartz-Lorentz}, we see that (with the notations used in the proof of Theorem \ref{th: local H^s})
	\begin{align*}
\|\mathcal{T}(\psi)-\psi^{(0)}\|_{L^a([0,T],L^{b,2})} \lesssim\int_{0}^{T}\|(w*|\psi_{\tau}|^{2})\psi_{\tau}\|_{H^{s}}\mathrm{d}\tau.
 \end{align*}
  We can then proceed exactly as in the proof of Theorem \ref{th: local H^s}. The only difference is the estimate of the part $w_{b/(b-2)}$ of $w$ belonging to $L^{b/(b-2),\infty}$. Using Hölder and Young inequalities in Lorentz spaces recalled in Section \ref{subsec:Lorentz-ineq}, it can be handled as follows:
	   \begin{align}\label{eq:w_ad/4}
	   \int_{0}^{T}\|(w_{b/(b-2)}*|\psi_{\tau}|^{2})\psi_{\tau}\|_{H^{s}}\mathrm{d}\tau&\lesssim\|w_{b/(b-2)}\|_{L^{b/(b-2),\infty}}\int_{0}^{T}\|\psi_\tau\|_{H^s}\|\psi_\tau\|^2_{L^{b,2}}\mathrm{d}\tau\lesssim T^{1-\frac2a}\|\psi\|_{Y_T}^3.
	   \end{align}
The contraction property of $\mathcal{T}$ can be proven similarly. The rest of the proof follows the arguments used at the end of the proof of Theorem \ref{th: local H^s}.   
	\end{proof}

	\section{Global existence for long-range potentials}\label{sec:global-long}

In this section we prove Theorems \ref{th:global-long-range-intro} and \ref{th:global-long-range-intro-small-energy-opt}. To do so, we use the conservation of the mass and energy (see \eqref{eq:def_mass_intro} and \eqref{eq:def_energy_intro} for the definitions).

The proof of the next lemma is analogous to that in \cite{lenzmann_semirel}*{Lemma 2}. We do not reproduce it here.

\begin{lemma}
\label{lemma: en conservation}
Let $s\ge0$, $w$ as in Theorem \ref{th: local H^s} or \ref{th: local_with_Strichartz}, and let $\psi$ be a local solution to \eqref{eq:hartree_semi_intro} given by Theorem \ref{th: local H^s} or \ref{th: local_with_Strichartz}. Then, for all $t \in  [0,T_{\max})$ it holds
\begin{equation*}
    M(\psi_t)= M(\psi_0). 
\end{equation*}
Moreover, if $s\ge\frac12$, for all $t \in  [0,T_{\max})$ it holds
\begin{equation*}
     E(\psi_t)= E(\psi_0). 
\end{equation*}
\end{lemma}

Now we can prove the global existence of solutions to \eqref{eq:hartree_semi_intro} stated in Theorem \ref{th:global-long-range-intro}.

\begin{proof}[Proof of Theorem \ref{th:global-long-range-intro}]
Let $s\ge\frac12$, $\psi_0\in H^s$ and $w\in \mathcal{W}_{d,1/2}=L^{d,\infty}+L^\infty$. Since $\mathcal{W}_{d,1/2}\subset\mathcal{W}_{d,s}$, Theorem \ref{th: local H^s} shows that \eqref{eq:hartree_semi_intro} admits a unique local solution $\psi \in \mathcal{C}^0 ([0,T_{\max}), H^s)\cap\mathcal{C}^1 ([0,T_{\max}), H^{s-1})$, for some $T_{\max}>0$, and that either $T_{\max}= \infty$ or $\|\psi_t\|_{H^s}\to \infty$ as $t\to T_{\max}$. Therefore is suffices to show that $\|\psi_t\|_{H^s}$ is uniformly bounded for $t\in[0,T_{\max})$.

Since $\psi_t$ is a local solution in $H^s$ with $s\ge\frac12$, we can apply Lemma \ref{lemma: en conservation}. This yields
\begin{align}
\label{comp: energy}
    E(\psi_0) = E(\psi_\tau)\geq  \frac12 \| \langle \nabla\rangle ^{\frac12} \psi_\tau\|^2_{L^2} - \frac14 \int (w_-*|\psi_\tau|^2 )|\psi_\tau|^2\,  .
\end{align}
Writing $w_-$ in the form $w_-= (w_d)_-  + (w_\infty)_- \in L^{d, \infty} + L^\infty$, we use the inequalities
\begin{align*}
    \|(w_\infty)_-*|\psi_\tau|^2 \|_{L^\infty} \leq \|(w_\infty)_-\|_{L^\infty} \|\psi_\tau\|_{L^2}^2 = \|(w_\infty)_-\|_{L^\infty}\|\psi_0\|_{L^2}^2 
\end{align*}
and 
\begin{align*}
    \|(w_d)_-*|\psi_\tau|^2 \|_{L^\infty} &\leq \|(w_d)_-\|_{L^{d,\infty }}\||\psi_\tau|^2\|_{L^{\frac{1}{1-\frac{1}{d}},1}} =  \|(w_d)_-\|_{L^{d,\infty }}\|\psi_\tau\|_{L^{\frac{2d}{d-1}, 2}}^2\\
    &\leq  C_S\|(w_d)_-\|_{L^{d,\infty }}\|\psi_\tau\|_{H^{\frac12}}^2
\end{align*}
to bound the integral
\begin{align*}
    \left| \int (w_-*|\psi_\tau|^2 )|\psi_\tau|^2\,   \right| \leq \|(w_\infty)_-\|_{L^\infty}\|\psi_0\|_{L^2}^4+ C_S\|(w_d)_-\|_{L^{d,\infty }}\|\psi_\tau\|_{H^{\frac12}}^2\|\psi_0\|_{L^2}^2,
\end{align*}
where $C_S$ is the constant in the Sobolev embedding $H^{\frac12}\hookrightarrow L^{\frac{2d}{d-1}, 2}$. 
Applying this bound in \eqref{comp: energy}, we obtain 
\begin{equation*}
    E(\psi_0) +\frac14\|(w_\infty)_-\|_{L^\infty}\|\psi_0\|_{L^2}^4 \geq \frac12 \|\psi_\tau\|_{H^{\frac12}}^2 (1- \frac12 C_S \|(w_d)_-\|_{L^{d,\infty }}\|\psi_0\|_{L^2}^2).
\end{equation*}
Hence, if $\|(w_d)_-\|_{L^{d,\infty }}\|\psi_0\|_{L^2}^2 < 2C_S^{-1}$, we have 
\begin{equation}\label{eq:unif-bound-H12}
    \sup_{\tau\ge 0} \|\psi_\tau\|_{H^{\frac12}}^2 \lesssim E(\psi_0)+\frac14\|(w_\infty)_-\|_{L^\infty}\|\psi_0\|_{L^2}^4. 
\end{equation}

Now, by Duhamel's formula \eqref{eq:Hartree-Duhamel} and the estimates as above with $w$ instead of $w_-$,  we have, for all $t \in [0, T]$ with $T< T_{\max}$
\begin{align}
\label{comp: Hs bound for glob long r}
    \|\psi_t\|_{H^s} \lesssim \|\psi_0\|_{H^s} + \|w\|_{L^{d, \infty}+ L^\infty} \int_0^t \|\psi_\tau\|_{H^{\frac12}}^2  \|\psi_\tau\|_{H^s}\, \mathrm{d}\tau.
\end{align}
Together with \eqref{eq:unif-bound-H12}, this yields
\begin{equation*}
    \|\psi_t\|_{H^s} \lesssim \|\psi_0\|_{H^s} + \|w\|_{L^{d, \infty}+ L^\infty} \int_0^t   \|\psi_\tau\|_{H^s}\, \mathrm{d}\tau,
\end{equation*}
and by Gronwall's inequality we obtain that for some $C>0$ and all $t\in[0,T]$, it holds
\begin{equation*}
    \|\psi_t\|_{H^s} \lesssim \|\psi_0\|_{H^s} e^{Ct}.
\end{equation*}
Therefore, we have a uniform bound on $\|\psi_t\|_{H^s}$ on any arbitrary time interval $[0,T]$. As argued above, this implies that the solution $\psi$ is globally defined. 
\end{proof}

For small energy initial data, the class of admissible potentials can be extended using the Strichartz estimates of Proposition \ref{prop:Strichartz-Lorentz}. The next result implies Theorem \ref{th:global-long-range-intro-small-energy-opt} by taking $\beta=0$ and  $\frac1a=\frac1b=\frac{d-1}{2d+2}$. 

\begin{theo}\label{th:global-long-range-intro-small-energy}
	Let $s=\frac12$. Let $0\le\beta\le1$ and let $(a,b)$ be an admissible pair for the Strichartz estimates of Proposition \ref{prop:Strichartz-Lorentz} with $b\le a$. There exists a universal constant $C_0>0$ such that, for all $w$ even of the form $w = w_{d/2} + w_\infty \in L^{b/(b-2),\infty} + L^\infty\subset L^{d/2}+L^\infty$ and $\psi_0 \in H^{\frac12}$ verifying
	\begin{equation}
	\label{eq: glob ex w_1 small-energy-intro}
	    \| (w_{d/2})_-\|_{L^{\frac{d}2}}(E(\psi_0)+\|(w_\infty)_-\|_{L^\infty}\|\psi_0\|^4_{L^{2}})<C_0 ,
	\end{equation}
	and
	\begin{equation}\label{eq: glob ex w_1 small-energy-intro2}
	    \| (w_{d/2})_-\|_{L^{\frac{d}2}}\|\psi_0\|_{H^{\frac12}}^2<C_0,
	\end{equation}
	Eq. \eqref{eq:hartree_semi_intro} admits a unique solution 
	\begin{equation*}
		\psi_t \in \mathcal{C}^0([0,\infty), H^{\frac12}) \cap \mathcal{C}^1 ([0,\infty), H^{-\frac12})\cap L^a_{\mathrm{loc}}([0,\infty),L^{b,2}).
	\end{equation*}
\end{theo}

\begin{proof}
   Let $w\in L^{b/(b-2),\infty}+L^\infty$ and let $\psi$ be a local solution to \eqref{eq:hartree_semi_intro} on $[0,T_{\max})$ given by Theorem \ref{th: local_with_Strichartz} (note that $\psi$ exists since $\psi_0\in H^{\frac12}$ and $\frac12\ge\frac12+\frac1a-\frac1b$). By the blow-up alternative stated in Theorem \ref{th: local_with_Strichartz}, it suffices to show that $t\mapsto\|\psi_t\|_{H^{1/2}}$ is uniformly bounded on $[0,T_{\max})$ and that $t\mapsto\|\psi_t\|_{L^b}$ belongs to $L^a([0,T_{\max}))$.
   
   Using the energy estimate \eqref{comp: energy} from the previous proof and Hölder's and Young's inequality, we obtain, for all $t\in[0,T_{\max})$,
   \begin{align*}
       \|\psi_t\|^2_{H^\frac12}&\le E(\psi_t)+\|(w_{d/2})_-\|_{L^{d/2}}\|\psi_t\|^4_{H^\frac12}+\|(w_\infty)_-\|_{L^\infty}\|\psi_t\|^4_{L^2}\\
       &=E(\psi_0)+\|(w_{d/2})_-\|_{L^{d/2}}\|\psi_t\|^4_{H^\frac12}+\|(w_\infty)_-\|_{L^\infty}\|\psi_0\|^4_{L^2},
   \end{align*}
where we used the conservation of the mass and energy (see Lemma \ref{lemma: en conservation}) in the equality.

Setting
\begin{equation*}
    X_t:=\|\psi_t\|^2_{H^\frac12}, \quad A:=\|(w_{d/2})_-\|_{L^{d/2}}\quad  C:=E(\psi_0)+\|(w_\infty)_-\|_{L^\infty}\|\psi_0\|^4_{L^2},
\end{equation*}
the previous inequality gives that $AX_t^2-X_t+C\ge0$ for all $t\in[0,T_{\max})$. If $(w_{d/2})_-=0$, then $\|\psi_t\|^2_{H^\frac12}$ is uniformly bounded on $[0,T_{\max})$ by $C$. Otherwise, letting $\delta:=1-4AC$ and choosing $C_0>0$ small enough, \eqref{eq: glob ex w_1 small-energy-intro} implies that $0<\delta<1$. Hence, for all $t\in[0,T]$, $X_t$ belongs to $[0,(1-\delta^{\frac12})/(2A)]\cup[(1+\delta^{\frac12})/(2A),\infty)$. Now if $C_0$ is small enough, one easily verifies using \eqref{eq: glob ex w_1 small-energy-intro2} that $X_0<(1-\delta^{\frac12})/(2A)$. By continuity of $t\mapsto X_t$, this implies that $X_t\in [0,(1-\delta^{\frac12})/(2A)]$ for all $t\in[0,T_{\max})$.

We have thus shown that $\|\psi_t\|_{H^{\frac12}}$ is uniformly bounded by $(2\|(w_{d/2})_-\|_{L^{d/2}})^{-1}$ on $[0,T_{\max})$ (or by $C$ if $(w_{d/2})_-=0$). To prove that $\|\psi_t\|_{L^{b,2}}$ belongs to $L^a([0,T_{\max}))$, we use again the Strichartz estimates from Proposition \ref{prop:Strichartz-Lorentz} together with \eqref{eq:w_ad/4}, obtaining for all $0<T<T_{\max}$,
\begin{align*}
    \|\psi\|_{L^a([0,T],L^{b,2})}&\lesssim\|\psi_0\|_{H^{\frac12}}+\int_0^T\|(w*|\psi_\tau|^2)\psi_\tau\|_{H^{\frac12}}\mathrm{d}\tau\\
    &\lesssim\|\psi_0\|_{H^{\frac12}}+\int_0^T\big(\|w_\infty\|_{L^\infty}\|\psi_0\|^2_{L^2}\|\psi_\tau\|_{H^{\frac12}}+\|w\|_{L^{b/(b-2),\infty}}\|\psi_\tau\|_{L^b}^2\|\psi_\tau\|_{H^{\frac12}}\big)\mathrm{d}\tau\\
    &\lesssim C_1(w,M(\psi_0),E(\psi_0))(1+T)+C_2(w,M(\psi_0),E(\psi_0))T^{1-\frac{2}{a}}\|\psi\|_{L^a([0,T],L^{b,2})}^2,
\end{align*}
for some positive constants $C_j(w,M(\psi_0),E(\psi_0))$ depending on $w$ and the mass and energy of $\psi_0$. Fixing $T_0>0$ small enough (depending only again on $w$ and $M(\psi_0)$, $E(\psi_0)$) and arguing as before, we deduce that
\begin{equation*}
    \|\psi\|_{L^ a([0,T_0],L^{b,2})}\lesssim C_3(w,M(\psi_0),E(\psi_0)).
\end{equation*}
Choosing an integer $n$ large enough so that $T_0>\frac1nT_{\max}=:T_1$ yields
\begin{equation*}
    \|\psi\|_{L^a([0,T_{\max}),L^{b,2})}^a=\sum_{j=0}^{n-1}\|\psi\|^a_{L^a([jT_1,(j+1)T_1),L^{b,2})}<\infty,
\end{equation*}
which by the blowup alternative stated in Theorem \ref{th: local_with_Strichartz}, conludes the proof.
\end{proof}

	\section{Global existence and pointwise time-decay for short range potentials}\label{sec:global-short}

In this section we prove the global well-posedness and time-decay properties of  \eqref{eq:hartree_semi_intro} stated in Theorem \ref{th:global-short-range-intro}. We use the time-decay properties of the free flow associated to \eqref{eq:hartree_semi_intro}. This will allow us to consider interacting potentials $w \in \meas + L^{q}$, $ 1\leq q <\frac{2d}{3}$ for $d\geq 3$.

In Appendix \ref{app: pointw dec} (see Lemma \ref{lemma: free t dec interp}), we recall that the solutions of the linear equation corresponding to the free dynamics associated to \eqref{eq:hartree_semi_intro} satisfy the following time-decay estimates for all $f \in H^{s,p' }\cap H^s$:
	\begin{equation}\label{eq:a3}
		\|e^{-it\langle \nabla \rangle} f\|_{L^{p}\cap L^\infty} \lesssim \langle t\rangle^{ -\frac{d}{2r}}\|f\|_{H^{s,p'}\cap H^s},
	\end{equation}
for any $2\le p \le\infty$, $s\geq \frac d2 +1$, with $\frac1{p}+\frac1{p'}=1$ and $\frac1p+\frac1{2r}=\frac12$.

Now we are ready to prove  Theorem \ref{th:global-short-range-intro}.

\begin{proof}[Proof of Theorem \ref{th:global-short-range-intro}]
We will construct a solution to the Duhamel equation
\begin{equation}\label{eq:Duhamel}
\psi_{t}=\psi^{(0)}_{t}-i\int_{0}^{t}e^{i(\tau-t)\langle\nabla\rangle}(w*|\psi_{\tau}|^{2}))\,\psi_{\tau}\,\mathrm{d}\tau\,.
\end{equation}
To simplify the notations, given $s,r,p$ as in the statement of the theorem, we set $S:=S^{s,r,p}$ where $S^{s,r,p}$ was defined in~\eqref{eq:def-Ssrp}.

Eq. \eqref{eq:a3} ensures that the free evolution $\psi_t^{(0)}$ is in $S$.
For $\varphi_{j}\in S$, $j\in\mathbb{Z}_{3}=\{0,1,2\}$, we set
\[
I(\varphi_{0},\varphi_{1},\varphi_{2})(t)=\int_{0}^{t}e^{i(\tau-t)\langle\nabla\rangle}(w*(\varphi_{0,\tau}\varphi_{1,\tau}))\,\varphi_{2,\tau}\,\mathrm{d}\tau\,.
\]
We first show that $I$ is well-defined on $S^{3}$, with values in $S$. It follows from  \eqref{eq:a3} that
\begin{align}
\|I(\varphi_{0},\varphi_{1},\varphi_{2})\|_{S} & \leq\sup_{t>0}\int_{0}^{t}\|(w*(\varphi_{0,\tau}\varphi_{1,\tau}))\,\varphi_{2,\tau}\|_{H^{s}}\,\mathrm{d}\tau\notag\\
 & \qquad+\sup_{t>0}\langle t\rangle^{\frac{d}{2r}}\int_{0}^{t}\langle t-\tau\rangle^{-\frac{d}{2r}}\|(w*(\varphi_{0,\tau}\varphi_{1,\tau}))\,\varphi_{2,\tau}\|_{H^{s}\cap H^{s,p'}}\,\mathrm{d}\tau\,. \label{eq:estim_I3}
\end{align}
Let $w_1\in\meas $ and $w_q\in L^q$ be such that $w=w_1+w_q$. Using~(\ref{eq:convolution-Hartree-Hs-M1}), we obtain
\begin{align*}
\int_{0}^{t}\|(w_{1}*(\varphi_{0,\tau}\varphi_{1,\tau}))\,\varphi_{2,\tau}\|_{H^{s}} \mathrm{d}\tau& \lesssim\|w_{1}\|_{\meas}\int_{0}^{t}\sum_{j\in\mathbb{Z}_{3}}\|\varphi_{j,\tau}\|_{H^{s}}\|\varphi_{j+1,\tau}\|_{L^{\infty}}\|\varphi_{j+2,\tau}\|_{L^{\infty}}\mathrm{d}\tau\\
 & \lesssim\|w_{1}\|_{\meas}\,\prod_{j\in\mathbb{Z}_{3}}\|\varphi_{j}\|_{S}\,\int_{0}^{t}\langle\tau\rangle^{-\frac{d}{r}}\mathrm{d}\tau
\end{align*}
and similarly, using~(\ref{eq:convolution-Hartree-Hs-Lq-theta}) and the fact that $\|\varphi_{j,\tau}\|_{H^s}^{1-\theta(\frac{r}{q})}\|\varphi_{j,\tau}\|^{\theta(\frac{r}{q})}_{L^{p}\cap L^{\infty}}\lesssim \langle \tau \rangle^{-\frac{d}{2r}\theta(\frac{r}{q})} \|\varphi_{j}\|_{S} $ we have
\begin{align*}
\int_{0}^{t}\|(w_{q}*(\varphi_{0,\tau}\varphi_{1,\tau}))\,\varphi_{2,\tau}\|_{H^{s}} \mathrm{d}\tau
 & \lesssim\|w_{q}\|_{L^{q}}\,\prod_{j\in\mathbb{Z}_{3}}\|\varphi_{j}\|_{S}\,\int_{0}^{t}\langle\tau\rangle^{-\frac{d}{r}\theta(\frac{r}{q})}\,\mathrm{d}\tau\,,
\end{align*}
where we recall that $\theta(u)=\min\{1,u\}$. We can proceed in the same way to estimate the second term in the right-hand side of \eqref{eq:estim_I3}. Using the bound~(\ref{eq:convolution-Hartree-Hspprime-M1-theta}) twice, we obtain
\begin{multline*}
\langle t\rangle^{\frac{d}{2r}}\int_{0}^{t}\langle t-\tau\rangle^{-\frac{d}{2r}}\|(w_{1}*(\varphi_{0,\tau}\varphi_{1,\tau}))\,\varphi_{2,\tau}\|_{H^{s}\cap H^{s,p'}}\mathrm{d}\tau\\
\lesssim\|w_{1}\|_{\meas}\,\prod_{j\in\mathbb{Z}_{3}}\|\varphi_{j}\|_{S}\,\langle t\rangle^{\frac{d}{2r}}\int_{0}^{t}\langle t-\tau\rangle^{-\frac{d}{2r}}\left(\langle\tau\rangle^{-\frac{d}{r}}+\langle\tau\rangle^{-\frac{d}{r}\theta(r-\frac{1}{2})}\right)\,\mathrm{d}\tau,
\end{multline*}
while~(\ref{eq:convolution-Hartree-Hs-Lq-theta})
and~(\ref{eq:convolution-Hartree-Hspprime-Lq-theta}) lead to
\begin{multline*}
\langle t\rangle^{\frac{d}{2r}}\int_{0}^{t}\|(w_{q}*(\varphi_{0,\tau}\varphi_{1,\tau}))\,\varphi_{2,\tau}\|_{H^{s}\cap H^{s,p'}}\mathrm{d}\tau\\
\lesssim\|w_{q}\|_{L^{q}}\,\prod_{j\in\mathbb{Z}_{3}}\|\varphi_{j}\|_{S}\,\langle t\rangle^{\frac{d}{2r}}\int_{0}^{t}\langle t-\tau\rangle^{-\frac{d}{2r}}\left(\langle\tau\rangle^{-\frac{d}{r}\theta(\frac{r}{q})}+\langle\tau\rangle^{-\frac{d}{r}\theta(\frac{r}{q}-\frac{1}{2})}\right)\,\mathrm{d}\tau\,.
\end{multline*}
We then observe that if $a\geq0,b>1$, then 
\begin{equation}
\sup_{t>0}\langle t\rangle^{a}\int_{0}^{t}\langle t-\tau\rangle^{-a}\langle\tau\rangle^{-b}\,\mathrm{d}\tau<\infty\,,\label{eq:integrale-finie}
\end{equation}
which we can apply with $a=\frac{d}{2r}$ and $b\in\left\{\frac dr, \frac{d}{r}\theta(\frac{r}{q}),\frac{d}{r}\theta(\frac{r}{q}-\frac{1}{2}),\frac{d}{r}\theta(r-\frac{1}{2})\right\} $.
Indeed, since $r\geq1$, we have
\begin{align*}
\min\left\{\frac dr, \frac{d}{r}\theta(\frac{r}{q}),\frac{d}{r}\theta(\frac{r}{q}-\frac{1}{2}),\frac{d}{r}\theta(r-\frac{1}{2})\right\} >1 & \Leftrightarrow\min\left\{ \frac{1}{r},\frac{1}{q},\frac{1}{q}-\frac{1}{2r},1-\frac{1}{2r}\right\} >\frac{1}{d}\\
 & \Leftrightarrow\min\left\{ \frac{1}{r},\frac{1}{q},\frac{1}{q}-\frac{1}{2r}\right\} >\frac{1}{d}\\
 & \Leftrightarrow\frac{3}{2d}<\frac{1}{2r}+\frac{1}{d}<\frac{1}{q}\,,
\end{align*}
which is true by assumption. Therefore we have shown that $I$ is well-defined on $S^{3}$, with values in $S$. Moreover, taking the infimum over all the decompositions $w=w_1+w_q$ with $w_1\in\meas $ and $w_q\in L^q$, we deduce that 
\begin{equation}
\|I(\varphi_{0},\varphi_{1},\varphi_{2})\|_{S}\leq C\|w\|_{\meas+L^{q}}\prod_{j\in\mathbb{Z}_{3}}\|\varphi_{j}\|_{S}\,,\label{eq:borne-integrale-trilineaire}
\end{equation}
for some positive constant $C$.

Now we seek for a fixed point of the map $\mathcal{T}:S\to S$ defined by
\[
\mathcal{T}(\psi)=\psi^{(0)}-iI(\psi,\bar{\psi},\psi).
\]
We consider the restriction of $\mathcal{T}$ to the closed ball
\begin{equation*}
    \bar{B}:=\overline{B(\psi^{(0)},\|\psi^{(0)}\|_{S})}.
\end{equation*}
The bound~(\ref{eq:borne-integrale-trilineaire}), for $\psi$ in~$\bar{B}$, gives
\[
\|\mathcal{T}(\psi)-\psi^{(0)}\|_S=\|I(\psi,\bar{\psi},\psi)\|_{S}\leq C\|w\|_{\meas+L^{q}}\|\psi\|_{S}^{3}\leq C\|w\|_{\meas+L^{q}}8\|\psi^{(0)}\|_{S}^{3}\,.
\]
It follows that if $\|w\|_{\meas+L^{q}}\|\psi\|_{S}^{2}\leq\frac{1}{8C}$~, then  $\mathcal{T}$ sends $\bar{B}$ into itself.

Now we verify that $\mathcal{T}$ is a contraction on $\bar{B}$. Using again~(\ref{eq:borne-integrale-trilineaire}) for $\psi_{1},\psi_{2}\in\bar{B}$ yields
\begin{align*}
\|\mathcal{T}(\psi_{1})-\mathcal{T}(\psi_{2})\|_{S} & \leq\|I(\psi_{1},\bar{\psi}_{1},\psi_{1}-\psi_{2})\|_{S}+\|I(\psi_{1},\bar{\psi}_{1}-\bar{\psi}_{2},\psi_{2})\|_{S}+\|I(\psi_{1}-\psi_{2},\bar{\psi}_{2},\psi_{2})\|_{S}\\
 & \leq3C\|w\|_{\meas+L^{q}}\|\psi_{1}-\psi_{2}\|_{S}(\|\psi_{1}\|_{S}+\|\psi_{2}\|_{S})^{2}\\
 & \leq 48C\|w\|_{\meas+L^{q}}\|\psi_{1}-\psi_{2}\|_{S}\|\psi^{(0)}\|_{S}^{2}.
\end{align*}
It is thus sufficient to have $\|w\|_{\meas+L^{q}}\|\psi^{(0)}\|_{S}^{2}\leq\frac{1}{48C}$
in order for $ \mathcal{T}$ to be a contraction on $\bar{B}$~,
this is true for $\|w\|_{\meas+L^{q}}\|\psi_{0}\|_{H^{s}\cap H^{s,p'}}^{2}$
sufficiently small as the assumption requires.

We have thus proven that $\mathcal{T}$ has a fixed point in $\bar{B}$, which gives the existence of a solution $\psi$ to Duhamel's equation \eqref{eq:Duhamel}. The fact that this solution $\psi$ belongs to $\mathcal{C}^0([0,\infty),H^{s})\cap \mathcal{C}^1([0,\infty),H^{s-1})$ follows from the previous estimates and standard arguments. Inequalities \eqref{eq:estim1_psit}--\eqref{eq:estim4_psit} are also direct consequences of the previous estimates.

Uniqueness of the solution follows from the blow-up alternative in Theorem \ref{th:local-long-range-intro}.
\end{proof}

\section{Maximal velocity estimates }\label{sec:max-vel}

In this section we prove our first result on the speed of propagation of the solution to~\eqref{eq:hartree_semi_intro}. We first prove Theorem~\ref{th:max-vel-intro}, which holds under a convexity assumption, in Subsection \ref{subsec:convex} and then Proposition~\ref{th:max-vel-general-intro1} in the case of general sets in Subsection \ref{subsec:general}.

\subsection{Convex sets}\label{subsec:convex}
In this section we prove Theorem~\ref{th:max-vel-intro}. If the initial data $\psi_0$ were in $H^s$ and the potential $w$ in $\mathcal{W}_{d,s}$ with $s\geq 1$, then Theorem~\ref{th:max-vel-intro} would be a corollary of our article on maximal velocity estimate for non-autonomous pseudo-relativistic Schrödinger equation \cite{BFG2a}, since in such case we could write a solution to \eqref{eq:hartree_semi_intro} as $\psi_t=U_t\psi_0$ with $U_t=U_{t,0}$ the propagator generated by $\langle\nabla\rangle+w*|\psi_t|^2$. In this section we prove that the maximal velocity estimate for \eqref{eq:hartree_semi_intro} as stated in Theorem~\ref{th:max-vel-intro} holds even for $s\geq 1/2$, assuming that the sets $X$ and $Y$ are convex and using several lemmata from  \cite{BFG2a} together with a final argument which avoids the use of the concept of propagator. 
Restricting to more regular initial states $\psi_0$ in $H^s$ with $s\geq1$ will allow us, in the next subsection, to prove the maximal velocity estimate for general subsets $X$ and $Y$ as stated in Proposition \ref{th:max-vel-general-intro1} by using the results of \cite{BFG2a}.

The main idea of the proof  of Theorem~\ref{th:max-vel-intro} is that, for a well chosen function $\ell$ it holds 
\begin{align}
\|\mathbf{1}_Y \psi_t\|_{L^2}
    \le \underset{\le\, \exp\big(-\frac{\dXY}{2}\big)}{\underbrace{\|\mathbf{1}_Ye^{\ell(x)}\|_{\mathcal{B}(L^2)}}} \quad
    \underset{\le\, \exp(t)\exp\big(-\frac{\dXY}{2}\big)\|\psi_0\|_{L^2}}{\underbrace{\|e^{-\ell(x)}\psi_t\|_{L^2}}}
   \,,\label{eq:idee-borne-de-vitesse-maximale}
\end{align}
where the $e^{-\ell(x)}\psi_t$ part is estimated through a Gronwall argument.
 To reach those estimates, we need several lemmata from \cite{BFG2a} that we recall without proof. The first one is a quantitative separation lemma which allows us to introduce the function $\ell$ appearing in \eqref{eq:idee-borne-de-vitesse-maximale}.
	\begin{lemma}[\cite{BFG2a}]\label{lm:convex}
		Let $X,Y$ two convex subsets of~$\mathbb{R}^d$ such that~$\dXY>0$. 
		There exist $x_0$ in $\mathbb{R}^d$ and a unit vector $n$  in $\mathbb{R}^d$ such that the affine functional~$\ell(x)=n\cdot (x-x_0)$ satisfies
		\begin{equation*}
		    \forall x\in X,\quad \ell(x)\ge\frac12 \dXY \qquad \text{and} \qquad \forall x\in Y,\quad \ell(x)\le-\frac12 \dXY\,.
		\end{equation*}
\end{lemma}
	
	From now on, in this section, we consider $\ell$ as in Lemma~\ref{lm:convex} and for all $\varepsilon>0$, we introduce a bounded regularization of $\ell$ by setting
	    \begin{equation*}
      \ell_\varepsilon(x):=f_\varepsilon(\ell(x))=f_\varepsilon(n\cdot(x-x_0)),
	    \end{equation*}
	    where $f_\varepsilon(r)=f(\varepsilon r)$, $f\in C^\infty(\mathbb{R})$, $f(r)=r$ on $[-1,1]$, $0\le f'\le1$ and $f'$ is compactly supported.

	For all $\varepsilon>0$, we define the operator $\Tepsilon$ on $H^1(\mathbb{R}^d)$ by
	    \begin{equation}\label{eq:defTeps}
	        \Tepsilon:=\mathrm{Im}\big(e^{\ell_\varepsilon(x)} \langle\nabla\rangle e^{-\ell_\varepsilon(x)}\big)=\frac{1}{2i}\big(e^{\ell_\varepsilon(x)} \langle\nabla\rangle e^{-\ell_\varepsilon(x)}-e^{-\ell_\varepsilon(x)} \langle\nabla\rangle e^{\ell_\varepsilon(x)}\big).
	    \end{equation}

	\begin{lemma}[\cite{BFG2a}]\label{lm:unif-bound}
	    For all $\varepsilon>0$, $\Tepsilon$ extends to a bounded operator on $L^2$, with
	    \begin{equation*}
	        \sup_{\varepsilon>0}\|\Tepsilon\|_{\mathcal{B}(L^2)}<\infty.
	    \end{equation*}
	\end{lemma}

For all $z\in\mathbb{C}\setminus(-\infty,0)$, we write $\sqrt{z}=\sqrt{|z|}e^{\frac{i}{2}\mathrm{Arg}(z)}$ with $-\pi<\mathrm{Arg}(z)<\pi$ and for all $\xi\in\mathbb{R}^d$, we set
	\begin{align}\label{eq:def_fpm}
    f_\pm(\xi):=\sqrt{|\xi\pm in|^2+1}=\sqrt{|\xi|^2\pm 2in\cdot\xi}.
\end{align}

\begin{lemma}[\cite{BFG2a}]\label{lm:bound_fpm}
    For all $\xi\in\mathbb{R}^d$, the bound
    \begin{equation*}
        \big|\mathrm{Im}f_\pm(\xi)\big|\le1
    \end{equation*}
    holds.
\end{lemma}

We define the operator $\Tzero$ on $L^2$ by
\begin{equation}\label{eq:defT0}
    \Tzero := \mathrm{Im}(f_+(-i\nabla)) = \mathcal{F} \, \mathrm{Im}(f_+(\xi)) \,  \mathcal{F}^{-1}.
\end{equation}
It then follows from Lemma \ref{lm:bound_fpm} that
\begin{equation}
    \|\Tzero\|_{\mathcal{B}(L^2)}\le1.\label{eq:borne-T0}
\end{equation}
The next lemma shows that $\Tzero$ is the weak limit of $\Tepsilon$ (defined in \eqref{eq:defTeps}) as $\varepsilon\to0$.

\begin{lemma}[\cite{BFG2a}]\label{lm:strongconv}
We have
\begin{equation*}
    \Tepsilon\to \Tzero, \quad \varepsilon\to0,
\end{equation*}
weakly in $\mathcal{B}(L^2)$.
\end{lemma}

We are now in the position to prove Theorem \ref{th:max-vel-intro}.

\begin{proof}[Proof of Theorem \ref{th:max-vel-intro}]
For $\tilde{\psi}_{0}$ in $\mathcal{C}_{0}^{\infty}$ whose support is included in $X$,
the integral form of Hartree's equation \eqref{eq:hartree_semi_intro} yields:
\begin{align}
\|e^{-\ell_{\varepsilon}}\tilde{\psi}_{t}\|_{L^{2}}^{2} & =\|e^{-\ell_{\varepsilon}}\tilde{\psi}_{0}\|_{L^{2}}^{2}+2\Im\int_{0}^{t}\langle e^{-\ell_{\varepsilon}}\tilde{\psi}_{\tau},e^{-\ell_{\varepsilon}}[\jnabla+w*|\tilde{\psi}_{\tau}|^{2}]\tilde{\psi}_{\tau}\rangle_{L^{2}}\diff\tau\nonumber \\
 & =\|e^{-\ell_{\varepsilon}}\tilde{\psi}_{0}\|_{L^{2}}^{2}+2\Im\int_{0}^{t}\langle e^{-\ell_{\varepsilon}}\tilde{\psi}_{\tau},e^{-\ell_{\varepsilon}}\jnabla e^{\ell_{\varepsilon}}e^{-\ell_{\varepsilon}}\tilde{\psi}_{\tau}\rangle_{L^{2}}\diff\tau\nonumber \\
 & =\|e^{-\ell_{\varepsilon}}\mathbf{1}_{X}\tilde{\psi}_{0}\|_{L^{2}}^{2}+2\int_{0}^{t}\langle e^{-\ell_{\varepsilon}}\tilde{\psi}_{\tau},\Tepsilon e^{-\ell_{\varepsilon}}\tilde{\psi}_{\tau}\rangle_{L^{2}}\diff\tau\,.\label{eq:borne-integrale-exp-ell-eps-psi}
\end{align}
Lemma~\ref{lm:unif-bound} provides a uniform control of $\|\Tepsilon \|_{\mathcal{B}(L^{2})}$, and thus there exists $C>0$ such that
\[
C=\sup_{\varepsilon\in(0,1)}\|\Tepsilon \|_{\mathcal{B}(L^{2})}<+\infty\,.
\]
This, along with Lemma~\ref{lm:convex} about the separation of the convex sets $X$ and $Y$ gives
\begin{align*}
\|e^{-\ell_{\varepsilon}}\tilde{\psi}_{t}\|_{L^2}^{2} & \leq e^{-\dXY}\|\tilde{\psi}_{0}\|_{L^2}^{2}+2C\int_{0}^{t}\|e^{-\ell_{\varepsilon}}\tilde{\psi}_{\tau}\|_{L^2}^{2}\diff\tau\,.
\end{align*}
It is then possible to use Gronwall's Lemma to get
\begin{equation}
\|e^{-\ell_{\varepsilon}}\tilde{\psi}_{t}\|_{L^2}^{2}\leq e^{2Ct-\dXY}\|\tilde{\psi}_{0}\|_{L^2}^{2}\, \label{eq:borne-unif-exp(-leps)-psi}
\end{equation}
which resembles the estimate in \eqref{eq:idee-borne-de-vitesse-maximale}. It remains to take the limit as $\varepsilon$ goes to $0$, obtain 1 instead of $C$ and replace $\tilde{\psi}_{0}$ by $\psi_{0}$ in $H^{s}$.

The square of the $L^{2}$ norm of $e^{-\ell_{\varepsilon}}\tilde{\psi}_{t}$ can be split into two terms:
\[
\int_{\mathbb{R}^{d}}e^{-2\ell_{\varepsilon}}|\tilde{\psi}_{t}|^{2}=\int_{\ell(x)\geq0}e^{-2\ell_{\varepsilon}}|\tilde{\psi}_{t}|^{2}+\int_{\ell(x)<0}e^{-2\ell_{\varepsilon}}|\tilde{\psi}_{t}|^{2}\,.
\]
The first term on the right-hand side converges to  $\int_{\ell(x)\geq0}e^{-2\ell}|\tilde{\psi}_{t}|^{2}$
by Lebesgue's dominated convergence Theorem while the second term on the right-hand side converges to $\int_{\ell(x)<0}e^{-2\ell}|\tilde{\psi}_{t}|^{2}$
by Beppo--Levi's monotone convergence Theorem. Hence, the uniform bound in~\eqref{eq:borne-unif-exp(-leps)-psi}
implies that~$e^{-\ell}\tilde{\psi}_{t}$ lies in $L^{2}$, with bound
\[
\|e^{-\ell}\tilde{\psi}_{t}\|_{L^2}^{2}\leq e^{2Ct-\dXY}\|\tilde{\psi}_{0}\|_{L^2}^{2}
\]
 and 
\begin{equation}
\|e^{-\ell_{\varepsilon}}\tilde{\psi}_{t}-e^{-\ell}\tilde{\psi}_{t}\|_{L^{2}}\xrightarrow[\varepsilon\to0]{}0\,. \label{eq:convergence-exp(-ell-eps)-psi}
\end{equation}
We wish to take the limit $\varepsilon$ to $0$ in \eqref{eq:borne-integrale-exp-ell-eps-psi}. We remark that
\begin{multline*}
    |  \langle e^{-\ell_{\varepsilon}}\tilde{\psi}_{\tau},\Tepsilon e^{-\ell_{\varepsilon}}\tilde{\psi}_{\tau}\rangle_{L^{2}}-\langle e^{-\ell}\tilde{\psi}_{\tau},\Tzero e^{-\ell}\tilde{\psi}_{\tau}\rangle_{L^{2}}|\\
  \leq\|e^{-\ell_{\varepsilon}}\tilde{\psi}_{\tau}-e^{-\ell}\tilde{\psi}_{\tau}\|_{L^{2}}\|\Tepsilon \|_{\mathcal{B}(L^{2})}\|e^{-\ell_{\varepsilon}}\tilde{\psi}_{\tau}\|_{L^{2}}
  +\|e^{-\ell}\tilde{\psi}_{\tau}\|_{L^{2}}\|\Tepsilon \|_{\mathcal{B}(L^{2})}\|e^{-\ell_{\varepsilon}}\tilde{\psi}_{\tau}-e^{-\ell}\tilde{\psi}_{\tau}\|_{L^{2}}\\
  +|\langle e^{-\ell}\tilde{\psi}_{\tau},(\Tepsilon -\Tzero )e^{-\ell}\tilde{\psi}_{\tau}\rangle_{L^{2}}|\,.
\end{multline*}
The first two terms on the right-hand side converge to $0$ as $\varepsilon$ goes to 0 using (\ref{eq:convergence-exp(-ell-eps)-psi}), 
Lemma~\ref{lm:unif-bound}, which provides a uniform bound on $\Tepsilon $,
and \eqref{eq:borne-unif-exp(-leps)-psi}, while the third term goes to 0 by  Lemma~\ref{lm:strongconv}, which states the convergence of $\Tepsilon $ towards $\Tzero $ for the weak operator topology.

 As a result, we can take the limit as $\varepsilon$ goes to 0 in~\eqref{eq:borne-integrale-exp-ell-eps-psi} and obtain
\[
\|e^{-\ell}\tilde{\psi}_{t}\|_{L^{2}}^{2}=\|e^{-\ell}\mathbf{1}_{X}\tilde{\psi}_{0}\|_{L^{2}}^{2}+2\int_{0}^{t}\langle e^{-\ell}\tilde{\psi}_{\tau},\Tzero e^{-\ell}\tilde{\psi}_{\tau}\rangle_{L^{2}}\diff\tau\,.
\]
A new application of Gronwall's Lemma along with~\eqref{eq:borne-T0} to control $\|\Tzero\|_{\mathcal{B}(L^2)}$ yields:
\[
\|e^{-\ell}\tilde{\psi}_{t}\|_{L^{2}}^{2}\leq e^{2t-\dXY}\|\tilde{\psi}_{0}\|_{L^{2}}^{2}\,.
\]
Using a second time  Lemma~\ref{lm:convex}, we obtain
\[
\|1_{Y}\tilde{\psi}_{t}\|_{L^{2}}\leq \|1_{Y}e^\ell\|_{\mathcal{B}(L^{2})}  \|e^{-\ell}\tilde{\psi}_{t}\|_{L^{2}}\leq e^{t-\dXY}\|\tilde{\psi}_{0}\|_{L^{2}}\,.
\]

We now consider $\psi_{0}$ in $H^{s}$ such that $\boldsymbol{1}_{X}\psi_{0}=\psi_{0}$.
Then, for all $\delta>0$ and all $T<T_{\max}(\psi_{0})$, Theorem~\ref{th: local H^s} implies that we can find $\tilde{\psi_{0}}$ in $\mathcal{C}_{0}^{\infty}$
such that $T_{\max}(\tilde{\psi}_{0})>T$ and $\sup_{t\in[0,T]}\|\psi_{t}-\tilde{\psi}_{t}\|_{H^{s}}\leq\delta$,
hence 
\[
\|\boldsymbol{1}_{Y}\psi_{t}\|_{L^{2}}\leq\|\boldsymbol{1}_{Y}\tilde{\psi}_{t}\|_{L^{2}}+\delta\leq e^{t-\dXY}\|\tilde{\psi}_{0}\|_{L^{2}}+\delta\leq e^{t-\dXY}(\|\psi_{0}\|_{L^{2}}+\delta)+\delta,
\]
for any $\delta>0$ which implies the bound in Theorem~\ref{th:max-vel-intro}.
\end{proof}

\subsection{General sets}\label{subsec:general}

We now turn to the proof of Proposition~\ref{th:max-vel-general-intro1} giving a maximal velocity estimate in the case of non-convex sets $X$ and $Y$. As mentioned before, we need to impose more regularity on the initial states $\psi_0$ in order to apply results from our companion paper \cite{BFG2a}. More precisely we must work in a setting where the family of operators $(\jnabla +w*|\psi_t|^2)_t$ defines a unitary propagator in the following sense.

\begin{definition} Let $I$ a compact interval of $\mathbb{R}$ and $(A_t)_{t\in I}$ a family of self-adjoint operators on $L^2$ such that $\mathcal{D}(A_t)\cap H^{1/2}$ is dense in $H^{1/2}$ and $A_t$ are continuously extendable to~$\mathcal{B}(H^{1/2},H^{-1/2})$. 
The map $I\times I \ni (t,s)\mapsto U(t,s)$ is a \emph{unitary propagator} associated to
\begin{equation}\label{Cauchy-pb-propagateur}
                i\partial_t \psi_t = A_t \psi_t\,, \quad t\in I,
\end{equation}
if and only if
\begin{enumerate}
    \item $U(t,s)$ is unitary on $L^2$ for all $t,s$ in $I$.
    \item $U(t,t)=\mathbf{1}_{L^2}$ for all $t$ in $I$ and $U(t,s)U(s,r)=U(t,r)$ for all $t,s,r$ in $I$.
    \item For all $s$ in $I$, the map $t\ni I \mapsto U(t,s)$ belongs to $\mathcal{C}^0(I,\mathcal{B}(H^{1/2})_{\mathrm{str}}) \cap \mathcal{C}^1(I,\mathcal{B}(H^{1/2},H^{-1/2})_{\mathrm{str}})$ and satisfies
    \[\forall t,s\in I\,,\,\forall \psi \in H^{1/2}\,, \quad i\partial_t U(t,s)\psi = A_t U(t,s)\psi\,,\]
\end{enumerate}
as an equality in $H^{-1/2}$. The index ``$\mathrm{str}$'' indicates that the considered topology  is the strong operator topology.
\end{definition}

Proposition~\ref{th:max-vel-general-intro1} will then be a consequence of the following theorem from~\cite{BFG2a}.

\begin{theo}[Corollary 1.4 in \cite{BFG2a}]\label{cor:max-vel-bounds-general-subsets}There exists $C_d>0$ such that, if   $(V_t)_{0\leq t\leq T}$, with $T>0$, is a family of real-valued potentials such that $(U(t,s))_{0\leq t,s\leq T}$ is a unitary propagator associated to $(\langle\nabla\rangle+V_t)_{0\leq t\leq T}$, then, for any measurable subsets  $X$ and $Y$ of $\mathbb{R}^{d}$, it holds
\[
\forall t\in[0,T]\,,\qquad \|\mathbf{1}_{Y}U(t,0)\mathbf{1}_{X}\|_{\mathcal{B}(L^{2})}\leq C_{d} \, e^{t-\dXY} \,  \langle\dXY\rangle^{d}\,,
\]
with $\langle r \rangle = \sqrt{1+r^2}$.
\end{theo}

To apply Theorem~\ref{cor:max-vel-bounds-general-subsets} with $V_t=w*|\psi_t|^2$, it is sufficient to check that it satisfies the assumptions of the following result from~\cite{BFG2a}. (Note that this result remains true in dimension $d=2$.)

\begin{prop}[Remark 1.8 in \cite{BFG2a}]\label{cor:existence-unitary-propagator-for-time-dependent-potential} There exists $c_d>0$, such that for any $T>0$, the following holds: Suppose that for all $t$ in $[0,T]$,  $V_t=V_{\infty,t}+V_{d,t}$ with $V_{\infty,t}$ in $L^\infty$, $V_{d,t}$ in $L^{d,\infty}$ and
\begin{enumerate}
    \item \label{item:borne-Vd} for all $t$ in $[0,T]$, $\|V_{d,t}\|_{L^{d,\infty}}<c_{d}$\,,
    \item \label{item:borne-Vinfty} $\sup_{t\in [0,T]}\|V_{\infty,t}\|_{L^\infty}<\infty$\,,
    \item \label{item:borne-dtV}$ \sup_{t\in [0,T]}\|\partial_t V_t\|_{L^{d,\infty}+L^\infty} < \infty$\,.
\end{enumerate}
Then the non-autonomous equation \eqref{Cauchy-pb-propagateur}, with the self-adjoint operator $A_t = \jnabla +V_t$ defined through the KLMN theorem,  admits a unique unitary propagator.
\end{prop}

It is thus sufficient to prove the following estimates.

\begin{prop}\label{cor:estimates-conv-potential-for-unitary-propagator}Let $s\geq 1$, $\psi_0$ in $H^{s}$ and $w$ in $\mathcal{W}_{d,s}$. Let  $\psi$ in $\mathcal{C}^0([0,T_{\mathrm{max}}),H^s) \cap \mathcal{C}^1([0,T_{\mathrm{max}}),H^{s-1})$ be the local solution to \eqref{eq:hartree_semi_intro} given by Theorem \ref{th:local-long-range-intro}. Then, for any $0<T<T_{\mathrm{max}}$, it holds
\begin{enumerate}
    \item  $\sup_{t\in [0,T]}\|w*|\psi_t|^2\|_{L^\infty}<\infty$\,,
    \item $ \sup_{t\in [0,T]}\|\partial_t (w*|\psi_t|^2)\|_{L^{d,\infty}+L^\infty} < \infty$\,.
\end{enumerate}
As a consequence, for all $T$ in $[0,T_{\mathrm{max}})$,  $(\langle\nabla\rangle + w*|\psi_t|^2)_{t\in[0,T]}$ generates a unitary propagator.
\end{prop}
\begin{proof}
   Let us first remark that if $w_\infty$ is in $L^\infty$, then:
   \begin{itemize}
       \item $\|w_\infty*|\psi_t|^2\|_{L^\infty}$ is uniformly bounded on $[0,T]$ as
   \begin{equation}
       \|w_\infty*|\psi_t|^2\|_{L^\infty} 
        \lesssim \|w_\infty\|_{L^\infty}  \|\psi_t\|_{L^2}^2=\|w_\infty\|_{L^\infty}  \|\psi_0\|_{L^2}^2\,.
    \end{equation}
        \item $\|\partial_t(w_\infty*|\psi_t|^2)\|_{L^\infty}$ is uniformly bounded on $[0,T]$ as
   \begin{align*}
       \|\partial_t(w_\infty*|\psi_t|^2)\|_{L^\infty} 
       &\lesssim \|w_\infty*(\bar\psi_t \partial_t\psi_t)\|_{L^\infty} \\
       & \lesssim \|w_\infty\|_{L^\infty} \|\psi_t\|_{L^2} \|\partial_t\psi_t\|_{L^2} \\
       & \lesssim \|w_\infty\|_{L^\infty} \|\psi_t\|_{H^s} \|\partial_t\psi_t\|_{H^{s-1}}
   \end{align*}
   and the right-hand side is uniformly bounded on the compact interval $[0,T]$ using the continuity of $t\mapsto\|\psi_t\|_{H^s}$ and $t\mapsto\|\partial_t\psi_t\|_{H^{s-1}}$ on $[0,T]$ given by Theorem \ref{th:local-long-range-intro}.
   \end{itemize}
   
   We will now treat the singular part of $w$ distinguishing three cases, depending on  the parameter~$s$, as in the definition of $\setW$. In any case we will have that $w*|\psi_t|^2$ belongs to $L^\infty$ and therefore the first point of Proposition \ref{cor:existence-unitary-propagator-for-time-dependent-potential} obviously holds. 
   
   \paragraph{Case $1\leq s< \frac{d}{2}$.}
   We consider a decomposition $w=w_{d/(2s)}+w_\infty$ with $w_{d/(2s)}$ in $L^{d/(2s),\infty}$ and $w_\infty$ in $L^\infty$. The convolution product with $w_\infty$ can be estimated as above while for the convolution involving $w_{d/(2s)}$ we use the Young inequality in Lorentz spaces (Proposition \ref{pro:young-lorentz}) in the special case of a $L^\infty$ function, followed by the Hölder and Sobolev inequalities in Lorentz spaces (Propositions \ref{pro:holder-lorentz} and \ref{pro:sobolev-lorentz}):
   \begin{align*}
       \|w_{\frac{d}{2s}}*|\psi_t|^2\|_{L^{\infty}} 
        &\lesssim \|w_{\frac{d}{2s}}\|_{L^{\frac{d}{2s},\infty}} \| |\psi_t|^2 \|_{L^{\frac{d}{d-2s},1}} \\
        &\lesssim \|w_{\frac{d}{2s}}\|_{L^{\frac{d}{2s},\infty}} \|\psi_t\|_{L^{\frac{2d}{d-2s},2}}^2 \\
        &\lesssim \|w_{\frac{d}{2s}}\|_{L^{\frac{d}{2s},\infty}} \|\psi_t\|_{H^s}^2
   \end{align*}  
   and the right hand side is uniformly bounded on $[0,T]$ by Theorem \ref{th:local-long-range-intro}.

   We proceed similarly for the time derivative of the term involving the $L^{d/(2s),\infty}$ part of the potential:
   \begin{align*}
       \|\partial_t(w_{\frac{d}{2s}}*|\psi_t|^2)\|_{L^{d,\infty}}
       &\lesssim \|w_{\frac{d}{2s}}*(\bar\psi_t \partial_t\psi_t)\|_{L^{d,\infty}} \\
       & \lesssim \|w_{\frac{d}{2s}}\|_{L^{\frac{d}{2s},\infty}} \|\bar\psi_t \partial_t\psi_t\|_{L^{\frac{2d}{2d-4s+2},1}} \\
       & \lesssim \|w_{\frac{d}{2s}}\|_{L^{\frac{d}{2s},\infty}} \|\psi_t\|_{L^{\frac{2d}{d-2s},2}} \|\partial_t\psi_t\|_{L^{\frac{2d}{d-2(s-1)},2}} \\
       & \lesssim \|w_{\frac{d}{2s}}\|_{L^{\frac{d}{2s},\infty}} \|\psi_t\|_{H^{s}} \|\partial_t\psi_t\|_{H^{s-1}}
   \end{align*}
and the last line is uniformly bounded on $[0,T]$ by Theorem \ref{th:local-long-range-intro}.

\paragraph{Case $s=\frac{d}{2}$.}
    In this case $w$ is in $L^q+L^\infty$ for some $q>1$ and we can assume without loss of generality that $q$ is in $(1,d/2]$. Thus we can consider  $1\le s'<\frac{d}{2}$ such that $\frac{d}{2s'}=q$. In particular $\psi_0$ is in $H^{s'}$ and $w$ is in $\mathcal{W}_{d,s'}$. Since $s'$ is in $[1,\frac{d}{2})$, it suffices to apply the previous case.
   
\paragraph{Case $\frac{d}{2}<s$.}
      For the $\meas$ part $w_1$ of the potential $w$, we have
   \begin{equation}
       \|w_{1}*|\psi_t|^2\|_{L^{\infty}} 
        \lesssim \|w_{1}\|_{\meas} \|\psi_t\|_{L^{\infty}}^2 \lesssim \|w_{1}\|_{\meas} \|\psi_t\|_{H^s}^2
   \end{equation}
   and the right hand side is uniformly bounded in $[0,T]$ by Theorem \ref{th:local-long-range-intro}.

    Moreover, with $0<\varepsilon<\min\{s-\frac{d}{2},1\}$, we deduce from Propositions \ref{pro:young-lorentz-measure}, \ref{pro:holder-lorentz} and \ref{pro:sobolev-lorentz} that
       \begin{align*}
       \|\partial_t(w_{1}*|\psi_t|^2)\|_{L^{d,\infty}+L^\infty}
       &\lesssim \|\partial_t(w_{1}*|\psi_t|^2)\|_{L^{\frac{d}{1-\varepsilon},\infty}} \\
       &\lesssim \|w_{1}*(\bar\psi_t \partial_t\psi_t)\|_{L^{\frac{d}{1-\varepsilon},\infty}} \\
       & \lesssim \|w_{1}\|_{\meas} \|\bar\psi_t \partial_t\psi_t\|_{L^{\frac{d}{1-\varepsilon},\infty}} \\
       & \lesssim \|w_{1}\|_{\meas} \|\psi_t\|_{L^{\infty}} \|\partial_t\psi_t\|_{L^{\frac{d}{1-\varepsilon},\infty}} \\
       & \lesssim \|w_{1}\|_{\meas} \|\psi_t\|_{H^{s}} \|\partial_t\psi_t\|_{H^{\frac{d}{2} -1+\varepsilon}} \\
       & \lesssim \|w_{1}\|_{\meas} \|\psi_t\|_{H^{s}} \|\partial_t\psi_t\|_{H^{s-1}}\,.
   \end{align*}
Again the last line is uniformly bounded on $[0,T]$ by Theorem \ref{th:local-long-range-intro}.

Applying Proposition~\ref{cor:existence-unitary-propagator-for-time-dependent-potential} with $V_t=w*|\psi_t|^2$ concludes the proof.
\end{proof}

We can now conclude the proof of Proposition~\ref{th:max-vel-general-intro1}.
\begin{proof}[Proof of Proposition~\ref{th:max-vel-general-intro1}]
Let $\psi$ be the local solution to \eqref{eq:hartree_semi_intro} on $[0,T_{\mathrm{max}})$ given by Theorem \ref{th:local-long-range-intro}. By Proposition~\ref{cor:estimates-conv-potential-for-unitary-propagator}, $(\jnabla + w*|\psi_t|^2)_{t\in[0,T]}$ generates a unitary propagator $U(t,s)_{t,s\in[0,T]}$, for any $T$ in $(0,T_{\max})$. By uniqueness of the solution $\psi$ to \eqref{eq:hartree_semi_intro}, we have
\begin{equation*}
    \psi_t=U(t,0)\psi_0=U(t,0)\mathbf{1}_X\psi_0,
\end{equation*}
for all $t$ in $[0,T]$. 
Applying then Theorem~\ref{cor:max-vel-bounds-general-subsets} yields Proposition~\ref{th:max-vel-general-intro1}.
\end{proof}

\section{Scattering and asymptotic minimal velocity estimate}\label{sec:min-vel}

In this section we prove Theorems \ref{th:scattering-short-range-intro}, \ref{thm:invertibility-intro}, \ref{th:propag-est-intro} and \ref{th:min-vel-intro} on the scattering and long-time behavior of solutions to the pseudo-relativistic Hartree equation \eqref{eq:hartree_semi_intro} in the case of short-range interaction potentials. 

\subsection{Estimates on the solution  in weighted Sobolev spaces}\label{subsec:estimates_growth}

We begin with a preliminary subsection where we prove several estimates on the norm of the solutions to \eqref{eq:hartree_semi_intro} in weighted Sobolev spaces. These estimates subsequently play an important role in our analysis.

We first estimate the growth of the first moment of a solution in the $L^2$-norm. The proof is rather straightforward, see \cite{frohlich_lenzm07}*{Lemma A.1} for a slightly different argument. We recall the notation $\theta(x)=\min\{1,x\}$.

\begin{lemma}\label{lem:x-evol}
Let $0\le\gamma\le 2$. Under the conditions of Theorem \ref{th:global-short-range-intro}, assuming in addition that $\psi_0\in L^2_\gamma$, 
the global solution $\psi$ to \eqref{eq:hartree_semi_intro} given by Theorem \ref{th:global-short-range-intro} satisfies
\begin{equation}\label{eq:boundL2-0}
\|\psi_t\|_{L^2_\gamma}\lesssim\|\psi_0\|_{L^2_\gamma}+\langle t\rangle^\gamma\|\psi_0\|_{L^2},
\end{equation}
uniformly in $t\in[0,\infty)$.
\end{lemma}

\begin{proof}
From Duhamel's formula
\begin{equation}\label{eq:Duhamel_secscatt}
\psi_{t}=e^{-it\langle\nabla\rangle}\psi_{0}-i\int_{0}^{t}e^{i(\tau-t)\langle\nabla\rangle}(w*|\psi_{\tau}|^{2}))\,\psi_{\tau}\,\mathrm{d}\tau\,,
\end{equation}
and Lemma \ref{lm:xgamma-lin}, we obtain
\begin{align*}
\|\psi_t\|_{L^2_\gamma}&\le \big\| e^{-it\langle\nabla\rangle}\psi_0\big\|_{L^2_\gamma}+ \int_0^t \big\| e^{-i(t-\tau)\langle\nabla\rangle} (w * |\psi_\tau|^2)\psi_\tau\big\|_{L^2_\gamma}\, \mathrm{d}\tau \\
&\lesssim\|\psi_0\|_{L^2_\gamma}+\langle t\rangle^\gamma\|\psi_0\|_{L^2}\\
&\quad+\int_0^t \Big(\big\| (w * |\psi_\tau|^2)\psi_\tau\big\|_{L^2_\gamma}+\langle t-\tau\rangle^\gamma\big\| (w * |\psi_\tau|^2)\psi_\tau\big\|_{L^2}\Big)\, \mathrm{d}\tau.
\end{align*}
Using \eqref{eq:forcing-short-range-L1+Lq-weight-L2} and Theorem \ref{th:global-short-range-intro}, we estimate the first term in the integral as \begin{align}
\big\|(w * |\psi_\tau|^2)\psi_\tau\big\|_{L^2_\gamma}&\lesssim \|w\|_{\meas +L^{q}} \|\psi_\tau\|_{L^{p}\cap L^\infty}^{2\theta(\frac{r}{q})} \|\psi_\tau\|_{L^{2}_\gamma} \|\psi_0\|_{H^s\cap H^{s,p'}}^{2-2\theta(\frac{r}{q})}\notag \\
&\lesssim \langle \tau\rangle^{-\frac dr\theta(\frac{r}{q})} \|w\|_{\meas +L^{q}}\|\psi_0\|_{H^s\cap H^{s,p'}}^2 \|\psi_\tau\|_{L^{2}_\gamma}\notag\\
&\lesssim \varepsilon_0\langle \tau\rangle^{-\frac dr \theta(\frac{r}{q}) }  \|\psi_\tau\|_{L^{2}_\gamma},\label{eq:boundL2-1}
\end{align}
and likewise, applying Lemma \ref{lem:bounds-forcing-short-range} and Theorem \ref{th:global-short-range-intro} to the second term we obtain
\begin{align}
\big\| (w * |\psi_\tau|^2)\psi_\tau\big\|_{L^2}
&\lesssim \varepsilon_0\langle \tau\rangle^{-\frac dr\theta(\frac{r}{q})}  \|\psi_\tau\|_{L^2} = \varepsilon_0\langle \tau\rangle^{-\frac dr\theta(\frac{r}{q})}  \|\psi_0\|_{L^2}. \label{eq:boundL2-2}
\end{align}
Since $r<d$, we thus obtain that
\begin{align*}
\|\psi_t\|_{L^2_\gamma}&\lesssim\|\psi_0\|_{L^2_\gamma}+\langle t\rangle^\gamma(1+\varepsilon_0)\|\psi_0\|_{L^2} +\varepsilon_0\int_0^t\langle \tau\rangle^{-\frac dr \theta(\frac{r}{q}) } \|\psi_\tau\|_{L^2_\gamma}\mathrm{d}\tau. 
\end{align*}
Using again that $r<d$ and $q<d$ imply $\frac dr \theta(\frac{r}{q}) >1$, the statement of the lemma follows from Gronwall's inequality.
\end{proof}

We now need to estimate the norm of a solution to \eqref{eq:hartree_semi_intro}  in the weighted Sobolev space $H^s_\gamma$ defined in \eqref{eq:def_Hsgamma0}--\eqref{eq:def_Hsgamma}. 

 In the remainder of this subsection we restrict the class of admissible interaction potentials to $w\in \meas +L^q$ with $1\le q<\frac d2$.
Recall that, for $\gamma>\frac{d}{2r}$, $H^s_\gamma\hookrightarrow H^s\cap H^{s,p'}$ (where, as before $p=\frac{2r}{r-1}$, $p'=\frac{2r}{r+1}$). Recall also that Lemma \ref{lemma: weight sob norm} implies $  \|\langle\nabla\rangle^s \langle x\rangle^\gamma\varphi\|_{L^2}\lesssim \|\varphi\|_{H^s_\gamma}$.

\begin{lemma}\label{lm:estim1_S}
    Let $s\ge\frac{d}{2}+1$, $\max \{\frac d4, 1\}\le r <\frac d2$ and $1 \le q< \frac d2$. Let  $\gamma\ge0$ be such that $\frac{d}{2r}<\gamma\le2$. There exists $\varepsilon_0>0$ such that, for all  $w \in \meas+L^{q}$ and $\psi_0\in H^s_\gamma\hookrightarrow H^s\cap H^{s,p'}$ (with $p=\frac{2r}{r-1}$) satisfying
		\begin{equation*}
			\|w\|_{\meas +L^{q}} \|\psi_0\|^2_{H^s\cap H^{s,p'}}\le\varepsilon_0,
		\end{equation*} 
		the global solution $\psi$ to \eqref{eq:hartree_semi_intro} given by Theorem \ref{th:global-short-range-intro} satisfies
		\begin{equation*}
		  \|\psi_t\|_{H^s_\gamma}\lesssim\|\psi_0\|_{H^s_\gamma}+\langle t\rangle^\gamma\|\psi_0\|_{H^s\cap H^{s,p'}},
		\end{equation*}
		uniformly in $t\in[0,\infty)$.
\end{lemma}

\begin{proof}

Duhamel's formula \eqref{eq:Duhamel_secscatt} and Lemma \ref{lm:xgamma-lin} give
\begin{align*}
\|\psi_t\|_{H^s_\gamma}&=\big\|\langle x\rangle^\gamma\langle\nabla\rangle^s\psi_t\big\|_{L^2}\\
&\le \big\| e^{-it\langle\nabla\rangle}\langle\nabla\rangle^s\psi_0\big\|_{L^2_\gamma}
+ \int_0^t \big\| e^{-i(t-\tau)\langle\nabla\rangle} \langle\nabla\rangle^s\big[(w * |\psi_\tau|^2)\psi_\tau\big]\big\|_{L^2_\gamma}\, \mathrm{d}\tau \\
&\lesssim  \|\psi_0\|_{H^s_\gamma}+t^\gamma\|\psi_0\|_{H^s}\\
&\quad+ \int_0^t \Big(\big\|(w * |\psi_\tau|^2)\psi_\tau\big\|_{H^s_\gamma} +  \langle t-\tau\rangle^\gamma\big\| (w * |\psi_\tau|^2)\psi_\tau\big\|_{H^s}\Big)\mathrm{d}\tau. 
\end{align*}
Using \eqref{eq: weighted multilin M in Hs} from Lemma \ref{lemma: weight multilin} and \eqref{eq:forcing-short-range-Leb-Hs-to-weight-Hs}  from Lemma \ref{lemma: weight multilin 2} and the embeddings $H^s \hookrightarrow L^2$, $H^s \hookrightarrow L^{2q'}$, $H^s\hookrightarrow L^p\cap L^\infty \hookrightarrow L^\infty$ (since $p, 2q'\geq 2$), we estimate the first term in the integral as
\begin{align*}
&\big\|(w * |\psi_\tau|^2)\psi_\tau\big\|_{H^s_\gamma}\lesssim \|w\|_{\meas+L^{q}} \|\psi_\tau\|_{H^s}^{2-\theta(\frac{r}{q})} \|\psi_\tau\|_{L^p\cap L^\infty}^{\theta(\frac{r}{q})}  \|\psi_\tau\|_{H^s_\gamma} .
\end{align*}
Applying \eqref{eq:estim1_psit} and \eqref{eq:estim_psit_energy} from Theorem \ref{th:global-short-range-intro} therefore implies that
    \begin{align*}
\big\|(w * |\psi_\tau|^2)\psi_\tau\big\|_{H^s_\gamma}&\lesssim \langle \tau\rangle^{-\frac{d}{2r}\theta(\frac{r}{q})}\|w\|_{\meas+L^{q}} \|\psi_0\|_{H^s\cap H^{s,p'}}^2 \|\psi_\tau\|_{H^s_\gamma}\\
&\lesssim \varepsilon_0\langle \tau\rangle^{-\frac{d}{2r}\theta(\frac{r}{q})}  \|\psi_\tau\|_{H^s_\gamma}.
\end{align*}
Likewise, applying Lemma \ref{lem:bounds-forcing-short-range} and Theorem \ref{th:global-short-range-intro} to the second term in the integral above we have
\begin{align}
\big\| (w * |\psi_\tau|^2)\psi_\tau\big\|_{H^s}&\lesssim\|w\|_{\meas+L^{q}} \|\psi_\tau\|_{L^{p}\cap L^\infty}^{2\theta(\frac{r}{q})} \|\psi_\tau\|_{H^{s}}^{3-2\theta(\frac{r}{q})}\nonumber\\
&\lesssim \langle\tau\rangle^{-\frac dr \theta(\frac{r}{q})}\|w\|_{\meas+L^{q}} \|\psi_0\|_{H^s\cap H^{s,p'}}^3. \label{comp:for bound of G gamma}
\end{align}
Hence, since $r<d$ and $q<d$, it holds
\begin{align*}
    \int_0^t \langle t-\tau\rangle^\gamma\big\|(w * |\psi_\tau|^2)\psi_\tau\big\|_{H^s}\mathrm{d}\tau &\lesssim  \langle t\rangle^\gamma   \|w\|_{\meas+L^{q}} \|\psi_0\|_{H^s\cap H^{s,p'}}^3 \int_0^t\langle\tau\rangle^{-\frac dr\theta(\frac{r}{q})}\mathrm{d}\tau \\
    &\lesssim \varepsilon_0 \langle t\rangle^\gamma  \|\psi_0\|_{H^s\cap H^{s,p'}}. 
\end{align*}
Therefore we have shown that 
\begin{align*}
    \|\psi_t\|_{H^s_\gamma} &\lesssim \|\psi_0\|_{H^s_\gamma} +(1+\varepsilon_0)\langle t\rangle^\gamma\|\psi_0\|_{H^s\cap H^{s,p'}}
      +\varepsilon_0 \int_0^t \langle \tau\rangle^{-\frac{d}{2r}\theta(\frac{r}{q})} \|\psi_\tau\|_{H^s_\gamma} \mathrm{d}\tau. 
\end{align*}
Since $r<\frac d2$ and $q<\frac d2$ imply $\frac{d}{2r}\theta(\frac{r}{q})>1$, the statement of the lemma follows from Gronwall's inequality.
\end{proof}

We will need yet another related, auxiliary result, which will be used in our proof of the invertibility of the wave operators in Theorem \ref{thm:invertibility-intro}.

\begin{lemma}\label{eq:estim_2S}

   Let $s\ge\frac{d}{2}+1$, $\max \{ \frac{d}{4}, 1\}\le r <\frac d2$ and $1\le q\le r$.   Let  $\gamma\ge0$ be such that $\frac{d}{2r}<\gamma\le2$. There exists $\varepsilon_0>0$ such that, for all  $w \in \meas+L^{q}$ and $\psi_0\in H^s_\gamma\hookrightarrow H^s\cap H^{s,p'}$ (with $p=\frac{2r}{r-1}$) satisfying
		\begin{equation*}
			\|w\|_{\meas+L^{q}} \| \psi_0\|^2_{H^s\cap H^{s,p'}}\le\varepsilon_0,
		\end{equation*} 
		the global solution $\psi$ to \eqref{eq:hartree_semi_intro} given by Theorem \ref{th:global-short-range-intro} satisfies
		\begin{equation}\label{eq:estimate-Lpgamma}
		  \|\psi_t\|_{L^{p}_\gamma\cap L^\infty_\gamma}\lesssim\langle t\rangle^{\gamma-\frac{d}{2r}}\| \psi_0\|_{H^s_\gamma},
		\end{equation}
		uniformly in $t\in[0,\infty)$.
\end{lemma}

\begin{proof}
Starting with Duhamel's formula \eqref{eq:Duhamel_secscatt} and applying Lemma \ref{lm:linear-bound-weighted-Lp}, we estimate
\begin{align}
\big\|\psi_t\big\|_{L^{p}_\gamma\cap L^\infty_\gamma}&\le\big\| e^{-it\langle\nabla\rangle}\psi_0\big\|_{L^{p}_\gamma\cap L^\infty_\gamma}+ \int_0^t \big\| e^{-i(t-\tau)\langle\nabla\rangle} \big[(w * |\psi_\tau|^2)\psi_\tau\big]\big\|_{L^{p}_\gamma\cap L^\infty_\gamma}\, \mathrm{d}\tau\notag \\
&\lesssim \|\psi_0\|_{H^s_\gamma}+\langle t\rangle ^{\gamma-\frac{d}{2r}}\|\psi_0\|_{H^s\cap H^{s,p'}}\notag\\
&\quad+ \int_0^t \Big(\big\|(w * |\psi_\tau|^2)\psi_\tau\big\|_{H^s_\gamma} +\langle t-\tau\rangle^{\gamma-\frac{d}{2r}}\big\| (w * |\psi_\tau|^2)\psi_\tau\big\|_{H^s\cap H^{s,p'}}\Big)\, \mathrm{d}\tau .\label{eq:estim1-weigted-Lp}
\end{align}
Using \eqref{eq: weighted multilin M in Hs} from Lemma \ref{lemma: weight multilin} and \eqref{eq:forcing-short-range-Leb-Hs-to-weight-Hs} from Lemma \ref{lemma: weight multilin 2} we estimate the first term in the integral as
\begin{align*}
\big\|(w * |\psi_\tau|^2)\psi_\tau\big\|_{H^s_\gamma}\lesssim \|w\|_{\meas+L^{q}} \|\psi_\tau\|_{L^p\cap L^\infty}\big( \|\psi_\tau\|_{L^p\cap L^\infty} \|\psi_\tau\|_{H^s_\gamma}+\|\psi_\tau\|_{H^s}\|\psi_\tau\|_{L^p_\gamma\cap L^\infty_\gamma}\big) ,
\end{align*}
where we used that $r\ge q$ and $2q'\ge p$ under our assumptions. Now, as mentioned above, $\gamma>\frac{d}{2r}$ implies $\|\psi_0\|_{H^{s,p'}\cap H^s}\lesssim \|\psi_0\|_{H^s_\gamma}$ and hence Lemma \ref{lm:estim1_S} yields $ \|\psi_\tau\|_{H^s_\gamma} \lesssim \langle \tau \rangle^\gamma \|\psi_0\|_{H^s_\gamma}$. Applying this property and \eqref{eq:estim1_psit} and \eqref{eq:estim_psit_energy} from Theorem \ref{th:global-short-range-intro} to the previous inequality we obtain 
\begin{align}
   \big\|(w * |\psi_\tau|^2)\psi_\tau\big\|_{H^s_\gamma} &\lesssim  \|w\|_{\meas+L^{q}}\|\psi_0\|^2_{H^s\cap H^{s,p'}}\big(\langle \tau\rangle^{\gamma-\frac{d}{r}} \|\psi_0\|_{H^s_\gamma}+\langle\tau\rangle^{-\frac{d}{2r}}\|\psi_\tau\|_{L^p_\gamma\cap L^\infty_\gamma}\big)\notag\\
  &\lesssim \varepsilon_0\big(\langle \tau\rangle^{\gamma-\frac{d}{r}} \|\psi_0\|_{H^s_\gamma}+\langle\tau\rangle^{-\frac{d}{2r}}\|\psi_\tau\|_{L^p_\gamma\cap L^\infty_\gamma}\big) \label{eq:bound G gamma,s-G gamma,s}.
\end{align}

We can proceed similarly to estimate the second term of the integral in \eqref{eq:estim1-weigted-Lp}. Indeed, applying \eqref{eq:convolution-Hartree-Hs-Lq-theta}-\eqref{eq:convolution-Hartree-Hspprime-Lq-theta} from Lemma \ref{lem:bounds-forcing-short-range} we have 
\begin{align}
\|(w * |\psi_\tau|^2)\psi_\tau\|_{H^s} &\lesssim  \|w\|_{\meas+L^{q}} \|\psi_\tau\|_{H^s} \|\psi_\tau\|_{L^p\cap L^\infty}^2 \nonumber\\
&\lesssim \|w\|_{\meas+L^{q}} \langle \tau \rangle^{-\frac{d}{r}} \|\psi_0 \|_{H^s\cap H^{s,p'}}^3\notag\\
&\lesssim \varepsilon_0\langle \tau \rangle^{-\frac{d}{r}} \|\psi_0 \|_{H^s_\gamma} \label{eq:bound Hs-G gamma,s},
\end{align}
where we applied again \eqref{eq:estim1_psit} and \eqref{eq:estim_psit_energy} from Theorem \ref{th:global-short-range-intro} in the second inequality, and used the embedding $H^s_\gamma \hookrightarrow H^s \cap H^{s, p'}$ in the last one. Analogously, thanks to \eqref{eq:convolution-Hartree-Hspprime-M1-theta}-\eqref{eq:convolution-Hartree-Hspprime-Lq-theta} from Lemma \ref{lem:bounds-forcing-short-range} and Theorem \ref{th:global-short-range-intro} we obtain 
\begin{align*}
    \|(w * |\psi_\tau|^2)\psi_\tau\|_{H^{s, p'}} &\lesssim  \|w\|_{\meas+L^{q}} \|\psi_0 \|_{H^s\cap H^{s,p'}}^3 (\langle \tau \rangle^{-\frac{d}{r} \theta(r-\frac{1}{2})}  + \langle \tau \rangle^{-\frac{d}{r} \theta(\frac{r}{q}- \frac12)} ).
\end{align*}
Now we remark that, for $d>3$, we have $r-\frac12>1$ and $\frac{r}{q}-\frac12> 1$, while if $d=3$, then by assumption it holds $\frac12\le r-\frac12<1$ and $\frac12\le \frac{r}{q}-\frac12\leq 1$ and therefore $\langle \tau \rangle^{-\frac{d}{r} \theta(r-\frac{1}{2})}  +\langle \tau \rangle^{-\frac{d}{r} \theta(\frac{r}{q}- \frac12)} \lesssim  \langle \tau \rangle^{-\frac{d}{2r}} $. With these properties and the embedding $H^s_\gamma \hookrightarrow H^s \cap H^{s, p'}$ we conclude that 
\begin{align}
     \|(w * |\psi_\tau|^2)\psi_\tau\|_{H^{s, p'}} 
    & \lesssim \varepsilon_0 \langle \tau \rangle^{-\frac{d}{2r}} \|\psi_0 \|_{H^s_\gamma} .\label{eq:bound Hsprime-G gamma,s}
\end{align}

Combining \eqref{eq:estim1-weigted-Lp} with \eqref{eq:bound G gamma,s-G gamma,s}-\eqref{eq:bound Hsprime-G gamma,s}, we obtain
\begin{align*}
    \big\|\psi_t\big\|_{L^{p}_\gamma\cap L^\infty_\gamma} &\lesssim  \langle t\rangle^{\gamma-\frac{d}{2r}}  \|\psi_0\|_{H^s_\gamma}+ \varepsilon_0 \|\psi_0\|_{H^s_\gamma}\Big( \int_0^t\langle t-\tau\rangle^{\gamma-\frac{d}{2r}}\langle\tau\rangle^{-\frac {d}{2r}} \mathrm{d}\tau+\int_0^t\langle\tau\rangle^{\gamma-\frac{d}{r}}\mathrm{d}\tau\Big)\\
    &\quad+\varepsilon_0 \int_0^t\langle\tau\rangle^{-\frac{ d}{2r}}\|\psi_\tau\|_{L^p_\gamma\cap L^\infty_\gamma}\mathrm{d}\tau\\
    &\lesssim  \langle t\rangle^{\gamma-\frac{d}{2r}}  \|\psi_0\|_{H^s_\gamma}+\varepsilon_0 \int_0^t\langle\tau\rangle^{-\frac{d}{2r}}\|\psi_\tau\|_{L^p_\gamma\cap L^\infty_\gamma}\mathrm{d}\tau,
\end{align*}
where in the second inequality, we used the property 
\begin{equation*}
    \sup_{t>0} \langle t \rangle^a \int_0^t \langle t-\tau\rangle^{-a} \langle\tau \rangle^{-b}\, \mathrm{d}\tau<\infty \  \quad \textnormal{for any }a\geq0,\ b>1
\end{equation*}
with $b=\frac d{2r}$.
Using again that $\frac{d}{2r}>1$, the statement of the lemma follows from Gronwall's inequality.
\end{proof}

\subsection{Scattering states and wave operators}
\label{sect: wave op}

In this subsection we establish the scattering results stated in Theorems \ref{th:scattering-short-range-intro} and \ref{thm:invertibility-intro}. 
We begin with showing that, for short-range interaction potentials, any global solutions to \eqref{eq:hartree_semi_intro} given by Theorem \ref{th:global-short-range-intro} scatters to a free solution. The proof straightforwardly follows from the multinilear estimates stated in Lemma \ref{lem:bounds-forcing-short-range} and the regularity and decay properties of a solution obtained in \eqref{eq:estim1_psit}--\eqref{eq:estim4_psit}.

	\begin{proof}[Proof of Theorem \ref{th:scattering-short-range-intro}]
	We recall that $\psi_+$ is defined as 
	\begin{equation*}
	    	\psi_+ := \psi_0 - i \int_0^\infty e^{i\tau\langle\nabla\rangle} (w * |\psi_\tau|^2)\psi_\tau\, \mathrm{d}\tau.
	\end{equation*}
		First we note that $\psi_+$ is well-defined in $H^s$ since $\psi_0\in H^s$ and, using Lemma \ref{lem:bounds-forcing-short-range} and Theorem \ref{th:global-short-range-intro}, we have
		\begin{align}\label{eq:decay-scat}
		 \|e^{i\tau\langle\nabla\rangle}(w * |\psi_\tau|^2)\psi_\tau\|_{H^s}\lesssim \|w\|_{\meas+L^{q}} \|\psi_\tau\|_{L^{p}\cap L^\infty}^{2 \theta(\frac{r}{q})} \|\psi_\tau\|_{H^{s}}^{3-2 \theta(\frac{r}{q})}\lesssim\varepsilon_0\langle \tau\rangle^{-\frac dr\theta (\frac{r}{q})}\|\psi_0\|_{H^s\cap H^{s,p'}},
		\end{align}	
		which is integrable over $[0,\infty)$ since $r<d$ and $q<d$ imply $\frac{d}{r}\theta (\frac{r}{q})>1$.
		Next it follows from Duhamel's formula that
		\begin{align*}
			\psi_t - e^{-it\langle\nabla\rangle} \psi_+=i\int_t^\infty e^{i\tau\langle\nabla\rangle} (w * |\psi_\tau|^2)\psi_\tau\, \mathrm{d}\tau.
		\end{align*}
		Applying \eqref{eq:decay-scat} gives \eqref{eq:decay-scat0-intro} (or more precisely \eqref{eq:decay-scat0-intro-rate} in Remark \ref{rk:decay-scat0-intro-rate}). 
		
		To prove the continuity of $W_+$, let $\psi_0$, $\tilde\psi_0$ in $\mathcal{B}_{H^s\cap H^{s,p'}}(\delta_w)$ and let $\psi_t$, $\tilde\psi_t$ be the corresponding solutions to \eqref{eq:hartree_semi_intro} given by Theorem \ref{th:global-short-range-intro}. Then applying again Lemma \ref{lem:bounds-forcing-short-range} and arguing as in the proof of Theorem \ref{th:global-short-range-intro}, we obtain
		\begin{align*}
		& \big\|e^{i\tau\langle\nabla\rangle}(w * |\psi_\tau|^2)\psi_\tau-e^{i\tau\langle\nabla\rangle}(w * |\tilde\psi_\tau|^2)\tilde\psi_\tau\big\|_{H^s}\\
		&\lesssim\varepsilon_0\langle \tau\rangle^{-\frac{d}{r}\theta(\frac{r}{q}) }\|\psi_0-\tilde\psi_0\|_{H^s\cap H^{s,p'}},
		\end{align*}
		which in turn implies that
		\begin{align*}
		    \|W_+\psi_0-W_+\tilde\psi_0\|_{H^s}\lesssim\|\psi_0-\tilde\psi_0\|_{H^s\cap H^{s,p'}}.
		\end{align*}
		This implies the continuity of $W_+$.
	\end{proof}

In order to prove the invertibility properties of $W_+$ stated in Theorem \ref{thm:invertibility-intro}, we first have to justify the existence of the wave operator $\Omega_+$. As mentioned in the introduction, the Cauchy problem at $\infty$ for the pseudo-relativistic Hartree equation,  Eq. \eqref{eq: hartree semir infinity}, can be solved exactly as in Theorem \ref{th:global-short-range-intro}, considering the integral equation  
	\begin{equation*}
		\psi_t = e^{-it\langle \nabla \rangle}\psi_+ +i \int_t^\infty e^{-i(t-\tau)\langle \nabla \rangle} (w * |\psi_\tau|^2)\psi_\tau\, \mathrm{d}\tau,
	\end{equation*}
and then applying a fixed point argument in $S^{s,r,p}$ (recall that the norm in $S^{s,r,p}$ has been defined in \eqref{eq:def-Ssrp}). This leads to the following result.

	\begin{prop}[Solutions to \eqref{eq: hartree semir infinity}]
		\label{th:global-short-range-infinity}
	Let $s\ge\frac{d}{2}+  1$, $1\le r <d$ and $1\le q<\frac{2d}{3}$ be such that $\frac{1}{d}+\frac{1}{2r}<\frac{1}{q}$. Define $p$ by the relation $1=\frac{1}{r}+\frac{2}{p}$.

There exists $\varepsilon_0>0$ such that the following holds: for all $w\in\meas+L^{q}$ and~$\psi_{+}\in H^{s}\cap H^{s,p'}$ satisfying
\begin{equation}\label{eq:cond_w_psi0}
\|w\|_{\meas+L^{q}}\,\|\psi_{+}\|_{H^{s}\cap H^{s,p'}}^{2}\le\varepsilon_0,
\end{equation}
 Eq.~\eqref{eq: hartree semir infinity} has a unique solution $\psi$ in $S^{s,r,p}$.
	\end{prop}

Given this theorem one can mimic the proof of Theorem \ref{th:scattering-short-range-intro} to obtain the existence and continuity of the wave operator mapping any scattering state $\psi_+$ to the corresponding initial state $\psi_0$.

	\begin{theo}[Wave operator]\label{cor:scat2}
Under the conditions of Proposition \ref{th:global-short-range-infinity},  there exists $\delta_w>0$ such that the (non-linear) wave operator 
		\begin{align*}
    \Omega_+:\mathcal{B}_{H^s\cap H^{s,p'}}(\delta_w)&\to H^s, \\ 
    \psi_+&\mapsto  \Omega_+\psi_+= \psi_+ + i \int_0^\infty e^{i\tau\langle\nabla\rangle} (w * |\psi_\tau|^2)\psi_\tau\, \mathrm{d}\tau,
\end{align*}
is well-defined and continuous.
	\end{theo}

Our next purpose is to identify a subset \begin{equation*}
    \mathcal{A}\subset\mathcal{B}_{H^s\cap H^{s,p'}}(\delta_w)  
\end{equation*} 
such that $W_+:\mathcal{A}\to\mathcal{B}_{H^s\cap H^{s,p'}}(\delta_w)$ and $\Omega_+:\mathcal{A}\to\mathcal{B}_{H^s\cap H^{s,p'}}(\delta_w)$, which allows us to give a meaning to the expressions $\Omega_+W_+\varphi$ and $W_+\Omega_+\varphi$ for $\varphi\in\mathcal{A}$. The next theorem shows that, provided we restrict the class of potentials to $w \in \meas+L^{q}$ with $1\le q<\frac d2$, it suffices to consider for $\mathcal{A}$ a closed ball with sufficiently small radius in $H^s_\gamma$, where we recall that
\begin{equation*}
    \|\varphi\|_{H^s_\gamma}=\|\langle x\rangle^\gamma\langle\nabla\rangle^s\varphi\|_{L^2}.
\end{equation*}
Note that Theorem \ref{thm:invertibility} readily implies Theorem \ref{thm:invertibility-intro}. The proof relies on the estimates obtained in Section \ref{subsec:estimates_growth}.

\begin{theo}\label{thm:invertibility}
    Let $s\ge\frac d2 +1$, $ \max \{1, \frac d4\} \leq r <\frac d2$, $1\le q\le r $ and $\frac{d}{2r}<\gamma<\min \{2, \frac dr-1\}$. There exists $\varepsilon_0>0$ such that, for all $w \in \meas+L^{q}$ and $\psi_0\in H^s_\gamma$ such that
		\begin{equation*}
			\|w\|_{\meas+L^{q}} \| \psi_0\|_{H^s\cap H^{s,p'}}^2\le\varepsilon_0,
		\end{equation*}
	we have
	\begin{equation}\label{eq:boundW+}
	    \|W_+\psi_0-\psi_0\|_{H^s_\gamma}\lesssim\varepsilon_0\|\psi_0\|_{H^s_\gamma}.
	\end{equation}
	Likewise, for all $\psi_+\in H^s_\gamma$ such that
		\begin{equation*}
			\|w\|_{\meas+L^{q}} \| \psi_+\|_{H^s\cap H^{s,p'}}^2\le\varepsilon_0,
		\end{equation*}
	we have
	\begin{equation}\label{eq:boundOmega+}
	    \|\Omega_+\psi_+-\psi_+\|_{H^s_\gamma}\lesssim\varepsilon_0\|\psi_+\|_{H^s_\gamma}.
	\end{equation}
	In particular, there exists $\delta_w>0$ such that 
	\begin{equation*}
	\Omega_+W_+\varphi=W_+\Omega_+\varphi \quad\text{ for all }\quad\varphi\in\mathcal{B}_{H^s_\gamma}(\delta_w).
	\end{equation*}
\end{theo}

\begin{proof}
Since $1\le q\le r<\frac d2$, the conditions in  Theorem \ref{th:scattering-short-range-intro} are satisfied and hence the expression of $W_+\psi_0$ gives
    \begin{align*}
			\|W_+\psi_0-\psi_0\|_{H^s_\gamma} \le  \int_0^\infty \big\| e^{i\tau\langle\nabla\rangle} \langle\nabla\rangle^s\big[(w * |\psi_\tau|^2)\psi_\tau\big]\big\|_{L^2_\gamma}\, \mathrm{d}\tau.
    \end{align*}
  Using first Lemma \ref{lm:xgamma-lin}, we estimate the integrated norm as
    \begin{align*}
        \big\| e^{i\tau\langle\nabla\rangle} \langle\nabla\rangle^s\big[(w * |\psi_\tau|^2)\psi_\tau\big]\big\|_{L^2_\gamma}&\lesssim \big\| (w * |\psi_\tau|^2)\psi_\tau\big\|_{H^s_\gamma} + \langle\tau\rangle^{\gamma} \big\|(w * |\psi_\tau|^2)\psi_\tau\big\|_{H^s}.
    \end{align*}
Combining \eqref{eq:estimate-Lpgamma} and \eqref{eq:bound G gamma,s-G gamma,s}, we obtain
\begin{align*}
   \big\|(w * |\psi_\tau|^2)\psi_\tau\big\|_{H^s_\gamma} &\lesssim  \varepsilon_0\langle \tau\rangle^{\gamma-\frac{d}{r}} \|\psi_0\|_{H^s_\gamma}.
\end{align*}
Similarly, by \eqref{eq:bound Hs-G gamma,s} we have
\begin{align*}
\langle\tau\rangle^{\gamma}\|(w * |\psi_\tau|^2)\psi_\tau\|_{H^s} \lesssim \varepsilon_0\langle \tau\rangle^{\gamma-\frac{d}{r}} \|\psi_0\|_{H^s_\gamma}.
\end{align*}
Since $\gamma<\frac dr-1$, this implies \eqref{eq:boundW+}. The proof of \eqref{eq:boundOmega+} is identical.

Now since $H^s_\gamma\hookrightarrow H^s\cap H^{s,p'}$, we deduce from \eqref{eq:boundW+} and Theorem \ref{cor:scat2} that, for $\delta_w>0$ small enough and for all $\varphi\in\mathcal{B}_{H^s_\gamma}(\delta_w)$, $\Omega_+W_+\varphi$ is well-defined. The fact that $\Omega_+W_+\varphi=\varphi$ then follows from the definitions of $W_+$ and $\Omega_+$. The same holds if one inverts the roles of $W_+$ and $\Omega_+$.
\end{proof}

\subsection{Average velocity and instantaneous velocity}

This subsection contains first justifications that the average velocity and the ``instantaneous'' velocity converge to each other as $t\to\infty$. The results obtained here will then be used in the next subsection to prove the propagation estimates stated in Theorems \ref{th:propag-est-intro} and \ref{th:min-vel-intro}.

We recall the definition of the instantaneous velocity operator $\Theta:=-i\nabla\langle\nabla\rangle^{-1}$. Note the following identity 
\begin{equation}\label{eq:b1}
e^{it\langle\nabla\rangle}x^2e^{-it\langle\nabla\rangle}=(x+t\Theta)^2,
\end{equation}
which follows from direct computations, and  recall that $L^2_1$ is the weighted $L^2$ space, with weight $\langle x \rangle$.

\begin{prop}\label{prop:asympt1}
Under the conditions of Theorem \ref{th:global-short-range-intro}, assuming in addition that $\psi_0\in L^2_1$,  
the global solution $\psi$ to \eqref{eq:hartree_semi_intro} given by Theorem \ref{th:global-short-range-intro} satisfies
\begin{equation*}
\Big\langle\psi_t,\Big(\frac{x}{t}-\Theta\Big)^2\psi_t\Big\rangle_{L^2}\lesssim \big(\langle t\rangle^{-2}+\langle t\rangle^{2-\frac{ 2d}{r}\theta (\frac{r}{q})}\big)\|\psi_0\|^2_{L^2_1} \to 0 , \quad \text{as } t\to\infty.
\end{equation*}
\end{prop}

\begin{proof}
Using \eqref{eq:b1}, we rewrite
\begin{equation}\label{eq:b2}
\Big\langle\psi_t,\Big(\frac{x}{t}-\Theta\Big)^2\psi_t\Big\rangle_{L^2}=\frac{1}{t^2}\big\langle e^{it\langle\nabla\rangle}\psi_t,x^2e^{it\langle\nabla\rangle}\psi_t\big\rangle_{L^2}.
\end{equation}
Duhamel's formula gives
\begin{equation}\label{eq:b3}
e^{it\langle\nabla\rangle}\psi_t=\psi_0-i \int_0^t e^{i\tau\langle\nabla\rangle} (w * |\psi_\tau|^2)\psi_\tau\, \mathrm{d}\tau ,
\end{equation}
and we can therefore estimate
\begin{equation}\label{eq:b4}
\big\| e^{it\langle\nabla\rangle}\psi_t\big\|_{L^2_1}\le \|\psi_0\|_{L^2_1}+ \int_0^t \big\| e^{i\tau\langle\nabla\rangle} (w * |\psi_\tau|^2)\psi_\tau\big\|_{L^2_1}\, \mathrm{d}\tau .
\end{equation}
Applying Lemma \ref{lm:xgamma-lin} gives
\begin{align*}
\big\| e^{i\tau\langle\nabla\rangle} (w * |\psi_\tau|^2)\psi_\tau\big\|_{L^2_1}\lesssim \big\|(w * |\psi_\tau|^2)\langle x\rangle\psi_\tau\big\|_{L^2}+\langle \tau\rangle\big\|(w * |\psi_\tau|^2)\psi_\tau\big\|_{L^2}.
\end{align*}
Combining \eqref{eq:boundL2-0}, \eqref{eq:boundL2-1} and \eqref{eq:boundL2-2}, we then obtain 
\begin{align*}
\big\| e^{i\tau\langle\nabla\rangle} (w * |\psi_\tau|^2)\psi_\tau\big\|_{L^2_1}\lesssim\varepsilon_0\langle\tau\rangle^{1-\frac dr\theta (\frac{r}{q})}\|\psi_0\|_{L^2_1}.
\end{align*}
Since $\frac dr\theta (\frac{r}{q})>1$, this implies that
\begin{equation*}
 \int_0^t \big\| e^{i\tau\langle\nabla\rangle} (w * |\psi_\tau|^2)\psi_\tau\big\|_{L^2_1}\, \mathrm{d}\tau \lesssim\varepsilon_0\langle t\rangle^{2-\frac dr \theta (\frac{r}{q})}\|\psi_0\|_{L^2_1}.
\end{equation*}
Together with \eqref{eq:b2}--\eqref{eq:b4}, this gives the statement of the proposition.
\end{proof}

The previous proposition, combined with functional calculus, implies the following result.

\begin{prop}\label{prop:diff-f}
Under the conditions of Theorem \ref{th:global-short-range-intro}, assuming in addition that $\psi_0\in L_1^2$,  
the global solution $\psi$ to \eqref{eq:hartree_semi_intro} given by Theorem \ref{th:global-short-range-intro} satisfies
\begin{equation*}
\Big\|\Big(f\Big(\frac{x^2}{t^2}\Big)-f(\Theta^2)\Big)\psi_t\Big\|_{L^2}\lesssim (\langle t \rangle^{-1} + \langle  t\rangle^{1-\frac dr \theta (\frac{r}{q})}) \|\psi_0\|_{L^2_1} \to 0 , \quad \text{as } t\to\infty,
\end{equation*}
for all $f\in \mathcal{C}_0^\infty(\mathbb{R})$.
\end{prop}

\begin{proof}
Let $F$ be an almost analytic extension of $f$. This means that $F \in \mathcal{C}_0^\infty(\mathbb{C})$, $F|_{\mathbb{R}}=f$ and, for all $n\in\mathbb{N}$, $|\frac{\partial F}{\partial \bar z}(z)|\le C_n|\mathrm{Im}(z)|^n$, with $C_n>0$. Using the Helffer-Sj\"ostrand representation (see e.g. \cite{derezinski_ger}), we write
\begin{align*}
	&	f\Big(\frac{x^2}{t^2}\Big)-f(\Theta^2)\\
	&=\frac1\pi\int \frac{\partial F}{\partial \bar z}(z) \Big(\frac{x^2}{t^2}-z\Big)^{-1}\Big(\Theta^2-\frac{x^2}{t^2}\Big)(\Theta^2-z)^{-1} \mathrm{d}\mathrm{Re}(z)\mathrm{d}\mathrm{Im}(z) \\
		&=\int \frac{\partial F}{\partial \bar z}(z) \Big(\frac{x^2}{t^2}-z\Big)^{-1}\Big(\Theta+\frac{x}{t}\Big)\cdot\Big(\Theta-\frac{x}{t}\Big)(\Theta^2-z)^{-1} \mathrm{d}\mathrm{Re}(z)d\mathrm{Im}(z)\\
		&\quad-\frac{1}{t}\int \frac{\partial F}{\partial \bar z}(z) \Big(\frac{x^2}{t^2}-z\Big)^{-1}[\Theta,x](\Theta^2-z)^{-1} \mathrm{d}\mathrm{Re}(z)\mathrm{d}\mathrm{Im}(z),
\end{align*}
where $[\Theta,x]:=\sum_j[\Theta_j,x_j]$ is a bounded operator on $L^2$. In particular, from the properties of the almost analytic extension $F$, we see that the last term is a bounded operator, with bound $C_ft^{-1}$. For the first term on the right-hand side, we commute $\Theta-\frac{x}{t}$ through $(\Theta^2-z)^{-1}$, obtaining
\begin{align*}
		&\int \frac{\partial F}{\partial \bar z}(z) \Big(\frac{x^2}{t^2}-z\Big)^{-1}\Big(\Theta+\frac{x}{t}\Big)\cdot\Big(\Theta-\frac{x}{t}\Big)(\Theta^2-z)^{-1} \mathrm{d}\mathrm{Re}(z)\mathrm{d}\mathrm{Im}(z)\\
		&=\int \frac{\partial F}{\partial \bar z}(z) \Big(\frac{x^2}{t^2}-z\Big)^{-1}\Big(\Theta+\frac{x}{t}\Big)\cdot(\Theta^2-z)^{-1}\Big(\Theta-\frac{x}{t}\Big) \mathrm{d}\mathrm{Re}(z)\mathrm{d}\mathrm{Im}(z)\\
		&\quad+\frac{1}{t}\int \frac{\partial F}{\partial \bar z}(z) \Big(\frac{x^2}{t^2}-z\Big)^{-1}\\
		&  \qquad \qquad\Big(\Theta+\frac{x}{t}\Big)\cdot(\Theta^2-z)^{-1}(\Theta[\Theta,x]+[\Theta,x]\Theta)(\Theta^2-z)^{-1} \mathrm{d}\mathrm{Re}(z)\mathrm{d}\mathrm{Im}(z).
\end{align*}
As before, using that $\Theta$ and $[\Theta,x]$ are bounded operators on $L^2$, one easily sees that the last term is a bounded operator, with bound $C_ft^{-1}$. Taking the expectation in $\psi_t$ and using again the properties of the almost analytic extension $F$, it follows that
\begin{align*}
\Big\|\Big(f\Big(\frac{x^2}{t^2}\Big)-f(\Theta^2)\Big)\psi_t\Big\|_{L^2}\le C_f t^{-1} \|\psi_t\|_{L^2} + C\Big\|\Big(\Theta-\frac{x}{t}\Big)\psi_t\Big\|_{L^2}.
\end{align*}
Applying Proposition \ref{prop:asympt1} concludes the proof.
\end{proof}

\subsection{Phase-space and minimal velocity estimates}

We are now in position to prove Theorems \ref{th:propag-est-intro} and \ref{th:min-vel-intro}. We begin with the following proposition, which implies Theorem \ref{th:propag-est-intro} and provides a further result.

\begin{prop}\label{th:propag-estim}
Let $f,g\in \mathcal{C}_0^\infty(\mathbb{R})$ be such that $\mathrm{supp}(g)\,\cap\,\mathrm{supp}(f)=\emptyset$. Under the conditions of Theorem \ref{th:global-short-range-intro}, assuming in addition that $\psi_0\in L_1^2$,  
the global solution $\psi$ to \eqref{eq:hartree_semi_intro} given by Theorem \ref{th:global-short-range-intro} satisfies
\begin{equation}\label{eq:propag1}
\Big\|g\Big(\frac{x^2}{t^2}\Big)f(\Theta^2) \psi_t\Big\|_{L^2} \lesssim(\langle t \rangle^{-1} + \langle  t\rangle^{1-\frac dr \theta(\frac{r}{q})})\|\psi_0\|_{L_1^2}.
\end{equation}
Additionally, if $\psi_0$ is chosen such that the associated scattering state $\psi_+$ given by Theorem \ref{th:scattering-short-range-intro} satisfies $f(\Theta^2)\psi_+=\psi_+$, then
\begin{equation}\label{eq:propag2}
\Big\|g\Big(\frac{x^2}{t^2}\Big) \psi_t\Big\|_{L^2} \lesssim (\langle t \rangle^{-1} + \langle  t\rangle^{1-\frac dr})\|\psi_0\|_{L^2_1} + \langle t\rangle^{1-\frac dr \theta(\frac{r}{q})} \|\psi_0\|_{H^s\cap H^{s,p'}}.
\end{equation}
\end{prop}

\begin{proof}
To prove \eqref{eq:propag1}, it suffices to use that $\mathrm{supp}(g)\cap\mathrm{supp}(f)=\emptyset$, which gives
\begin{equation*}
\Big\|g\Big(\frac{x^2}{t^2}\Big)f(\Theta^2) \psi_t\Big\|_{L^2}=\Big\| g\Big(\frac{x^2}{t^2}\Big)\Big(f(\Theta^2)-f\Big(\frac{x^2}{t^2}\Big)\Big)\psi_t\Big\|_{L^2},
\end{equation*}
and then apply Proposition \ref{prop:diff-f}.

To prove \eqref{eq:propag2}, we consider a function $\tilde f\in \mathcal{C}_0^\infty(\mathbb{R})$ such that $\tilde f=1$ on $\mathrm{supp}(f)$ and $\mathrm{supp}(\tilde f)\cap\mathrm{supp}(g)=\emptyset$. Then we write
\begin{equation*}
\Big\|g\Big(\frac{x^2}{t^2}\Big) \psi_t\Big\|_{L^2} \le\Big\|g\Big(\frac{x^2}{t^2}\Big)\tilde f(\Theta^2) \psi_t\Big\|_{L^2}+\Big\|g\Big(\frac{x^2}{t^2}\Big)(1-\tilde f)(\Theta^2) \psi_t\Big\|_{L^2}.
\end{equation*}
The first term is estimated by \eqref{eq:propag1}, while for the second term, we have
\begin{align*}
\Big\|g\Big(\frac{x^2}{t^2}\Big)(1-\tilde f)(\Theta^2) \psi_t\Big\|_{L^2}&\lesssim \big\|(1-\tilde f)(\Theta^2) \psi_t\big\|_{L^2} \\
&=\big\|(1-\tilde f)(\Theta^2) e^{it\langle\nabla\rangle}\psi_t\big\|_{L^2} \\
&= \big\|(1-\tilde f)(\Theta^2) (e^{it\langle\nabla\rangle}\psi_t-\psi_+)\big\|_{L^2},
\end{align*}
where we used the unitarity of $e^{it\langle\nabla\rangle}$, and that $f(\Theta^2)\psi_+=\psi_+$ in the last inequality. This gives the result by the unitarity of $e^{it\langle\nabla\rangle}$ and Theorem \ref{th:scattering-short-range-intro}.
\end{proof}

Using the previous proposition, we can now easily prove Theorem \ref{th:min-vel-intro}.

\begin{proof}[Proof of Theorem \ref{th:min-vel-intro}]
Let $\psi_0$ be as in the statement of Theorem \ref{th:min-vel-intro}. By Theorem \ref{thm:invertibility}, we have $\psi_0\in \mathcal{B}_{H^s_\gamma}(2\delta_w)$ provided that $\delta_w$ is small enough. Since $H^s_\gamma\hookrightarrow H^s\cap H^{s,p'}$, we can apply Theorem \ref{th:global-short-range-intro} to obtain a solution $\psi$ to \eqref{eq:hartree_semi_intro} associated to the initial state $\psi_0$. Now we can consider $f,g\in \mathcal{C}_0^\infty(\mathbb{R})$ such that $\mathrm{supp}(f)\cap\mathrm{supp}(g)=\emptyset$, $\mathbf{1}_{[\alpha,1]}=\mathbf{1}_{[\alpha,1]}f$ and $\mathbf{1}_{[0,\alpha)}=\mathbf{1}_{[0,\alpha)}g$. Theorem \ref{th:min-vel-intro} then follows from Proposition \ref{th:propag-estim}.
\end{proof}

	\appendix
	
	\section{Commutation of weights and derivatives}

\label{app: commute}

In this appendix we prove the boundedness in $L^p$ spaces of some Fourier multipliers used in the main text, and we next deduce from it the equivalence of the norms $\|f\|_{H^{s,p}_\gamma}=\|\langle x \rangle ^\gamma \langle \nabla \rangle^s f \|_{L^p}$ and $\|\langle \nabla \rangle^s \langle x \rangle ^\gamma f \|_{L^p},$ for all $0\leq \gamma\leq 2$, $s\geq 0 $ and $1\le p \le \infty$. Most of the results recalled here are well-known (see e.g. \cite{Loefstroem}), especially in the case where $1<p<\infty$. We provide some details for self-containedness.

We will first need the following representation formulas.

\begin{lemma}
		\label{lemma: expr sqrt}
		Let $0\leq \alpha <2 $. Then there exists $c_\alpha>0$ such that 
		\begin{align*}
		    \langle \nabla \rangle^\alpha =&\, \frac{1}{c_\alpha} \int_0^\infty \big(\, u^{\frac\alpha2-1}- u^{\frac\alpha2} (-\Delta + 1 + u )^{-1}\, \big) \mathrm{d}u ,\\
		     \langle x \rangle^\alpha= &\,\frac{1}{c_\alpha} \int_0^\infty \big(\,u^{\frac\alpha2-1}- u^{\frac\alpha2} (x^2 +1+ u )^{-1}\, \big) \mathrm{d}u,
		\end{align*}
as operators on $\mathcal{S}'(\mathbb{R}^d)$.
	\end{lemma}
	\begin{proof}
		Let $ b>0$, then it holds
		\begin{align*}
			\int_0^\infty \frac{1}{u^{1-\frac\alpha2}(b+u)} \mathrm{d}u = \int_0^\infty \frac{1}{u^{1-\frac\alpha2}(1+u/b)}\frac{\mathrm{d}u}{b}= \frac{1}{b^{1-\frac\alpha2}}\int_0^\infty \frac{1}{t^{1-\frac\alpha2}(1+t)}\mathrm{d}t
		\end{align*}
		from which we obtain 
		\begin{equation*}
			b^{\alpha/2} = \frac{b}{c_\alpha} \int_0^\infty \frac{1}{u^{1-\frac\alpha2}(b+u)} \mathrm{d}u =  \frac{1}{c_\alpha} \int_0^\infty \big(\,u^{\frac\alpha2-1}- u^{\frac\alpha2} (b+ u )^{-1}\, \big) \mathrm{d}u 
		\end{equation*}
		where $c_\alpha=\int_0^\infty \frac{1}{t^{1-\alpha/2}(1+t)}\mathrm{d}t$. The statement follows replacing $b$ by $\langle x \rangle$ and $\langle \nabla \rangle$ (working in Fourier space for $\langle\nabla\rangle$).
	\end{proof}

	\begin{lemma}
	\label{lemma: resolvent Lp}
	    Let $1\leq p\leq \infty$ and $\lambda,s >0$, then it holds 
	    \begin{equation}\label{eq:bound-resolv_Lp}
	        \|(-\Delta + \lambda )^{-1} \|_{\mathcal{B}(L^p)} \lesssim \lambda^{-1}, \quad  \|\partial_k^m(-\Delta + \lambda )^{-1} \|_{\mathcal{B}(L^p)} \lesssim \lambda^{-m/2}, \quad m =1,2 
	    \end{equation}
	    and 
	    \begin{equation}\label{eq:bound-bessel_Lp}
	        \|\langle \nabla \rangle^{-s}\|_{\mathcal{B}(L^p)} \lesssim 1 , \quad \|\partial_k \langle \nabla \rangle^{-s-1}\|_{\mathcal{B}(L^p)} \lesssim 1
	    \end{equation}
	    for all $k =1,\dots,d$. If $1<p<\infty$, we also have
	    \begin{equation}\label{eq:bound-bessel2_Lp}
	        \|\partial_k \langle \nabla \rangle^{-1}\|_{\mathcal{B}(L^p)} \lesssim 1.
	    \end{equation}
	\end{lemma}
	
\begin{proof}
The properties of the resolvent can be derived by direct computations, using the explicit expression of its kernel, namely \begin{equation*}
    (-\Delta + \lambda )^{-1} f= G_\lambda*f, \qquad G_\lambda(x)=\int_0^\infty(4\pi t)^{-\frac{d}2}\mathrm{exp}\Big(-\frac{x^2}{4t}-\lambda t\Big)\mathrm{d}t,
\end{equation*} 
see e.g. \cite{LiebLoss}*{Theorem 6.23}. The first inequality in \eqref{eq:bound-resolv_Lp} then directly follows from Young's inequality, since
\begin{equation*}
    \|G_\lambda\|_{L^1}=\int_{\mathbb{R}^d}G_\lambda(x)\mathrm{d}x=(2\pi)^{\frac{d}{2}}\mathcal{F}(G_\lambda)(0)=(2\pi)^{\frac{d}{2}}\lambda^{-1}.
\end{equation*}
When $m=1$, to prove the second bound in \eqref{eq:bound-resolv_Lp}, we estimate
\begin{equation*}
    \|\partial_kG_\lambda\|_{L^1}\lesssim\int_{\mathbb{R}^d}\int_0^\infty|x_k|t^{-\frac{d}{2}}\mathrm{exp}\Big(-\frac{x^2}{4t}-\lambda t\Big)\frac{\mathrm{d}t}{t}.
\end{equation*}
Changing variables $x=\lambda^{-\frac12}y$ next $t=\lambda^{-1}u$, we obtain $\|\partial_kG_\lambda\|_{L^1}\lesssim\lambda^{-\frac12}$, which implies the second bound in \eqref{eq:bound-resolv_Lp} by Young's inequality. The case $m=2$ can be proved similarly by estimating $\|\partial_k^2 G_\lambda\|_{L^1}$.

Now we prove \eqref{eq:bound-bessel_Lp}. For $1 \leq p \leq \infty $, the operator $\langle \nabla \rangle^{-s}$ is bounded on $L^p$ by \cite{grafakos_modern}*{Corollary 6.1.6}. The proof is based as before on the explicit expression
\begin{equation*}
    \langle\nabla\rangle^{-s} f= J_s*f, \qquad J_s(x)=\frac{1}{\Gamma(\frac{s}{2})}\int_0^\infty (4\pi t)^{-\frac{d}{2}}\mathrm{exp}\Big(-\frac{|x|^2}{4t}-t\Big)t^{\frac{s}{2}-1}\mathrm{d}t.
\end{equation*} 
Using
\begin{equation*}
    \|J_s\|_{L^1}=\int_{\mathbb{R}^d}J_s(x)\mathrm{d}x=(2\pi)^{\frac{d}{2}}\mathcal{F}(J_s)(0)=(2\pi)^{\frac{d}{2}},
\end{equation*}
we obtain the first bound in \eqref{eq:bound-bessel_Lp} by Young's inequality. Moreover we can estimate
\begin{equation*}
    \|\partial_kJ_{s+1}\|_{L^1}\lesssim\int_{\mathbb{R}^d}\int_0^\infty|x_j|t^{-\frac{d}{2}}\mathrm{exp}\Big(-\frac{x^2}{4t}-t\Big)t^{\frac{s+1}{2}-2}\mathrm{d}t.
\end{equation*}
Changing variables $x=t^{\frac12}y$ in the integral over $x$ and using that $s>0$, we deduce that $\|\partial_kJ_{s+1}\|_{L^1}\lesssim1$, from which the second bound in \eqref{eq:bound-bessel_Lp} follows by Young's inequality.

Finally, since $1<p<\infty$, \eqref{eq:bound-bessel2_Lp} follows from the Mihlin–Hörmander Multiplier Theorem, see e.g. \cite{grafakos}*{Theorem 5.2.7}.
\end{proof}

\begin{lemma}
\label{lemma: [x^gamma, grad^s]}
    Let $0\leq \gamma\leq 2$,  $s\geq 0$ and $1\leq p \leq \infty$. Then it holds 
    \begin{equation*}
        \|[\langle x \rangle^\gamma , \langle \nabla \rangle^s]\langle x \rangle^{-\gamma} \langle \nabla \rangle^{-s} \|_{\mathcal{B}(L^p)}\lesssim 1, \quad \|[\langle x \rangle^\gamma , \langle \nabla \rangle^s]\langle \nabla \rangle^{-s}   \langle x \rangle^{-\gamma} \|_{\mathcal{B}(L^p)}\lesssim 1 .
    \end{equation*}
\end{lemma}

As a direct consequence of the previous lemma we have the following.
 
\begin{cor}\label{cor:weighted-sobolev}
 Let $0\leq \gamma\leq 2$,  $s\geq 0$ and $1\leq p \leq \infty$. There exist $c, c'>0$ such that 
 \begin{equation*}
    c'\,  \| \langle \nabla \rangle^s \langle x \rangle^\gamma f \|_{L^p} \leq   \| f \|_{H_\gamma^{s,p}} \leq c\,  \| \langle \nabla \rangle^s \langle x \rangle^\gamma f \|_{L^p} . 
 \end{equation*}
\end{cor}

\begin{proof}[Proof of Lemma \ref{lemma: [x^gamma, grad^s]}]
We will prove the lemma for the operator $[\langle x \rangle^\gamma , \langle \nabla \rangle^s]\langle x \rangle^{-\gamma} \langle \nabla \rangle^{-s} $, the other case can be proved analogously. 
We distinguish different cases: 
\begin{itemize}
    \item $s=2$, $\gamma \in \R^+$.

 We compute explicitly the commutator yielding
    \begin{align*}
        [\langle x \rangle^\gamma, \langle \nabla \rangle^2]\langle x \rangle^{-\gamma} \langle \nabla \rangle^{-2} =& [\langle x \rangle^\gamma, -\Delta]\langle x \rangle^{-\gamma} \langle \nabla \rangle^{-2}\\
        = &  \Delta (\langle x \rangle^\gamma) \langle x \rangle^{-\gamma} \langle \nabla \rangle^{-2} +  \langle x \rangle^{-\gamma} \nabla (\langle x \rangle^{\gamma})\cdot \nabla \langle \nabla \rangle^{-2} \\
        & + \nabla ( \langle x \rangle^{\gamma})\cdot \nabla ( \langle x \rangle^{-\gamma})  \langle \nabla \rangle^{-2}, 
    \end{align*}
    which is a sum of bounded operators on $L^p$ thanks to Lemma \ref{lemma: resolvent Lp}.
    
    \item $s \in [0,2)$,  $\gamma = 2$. 
    
    Applying Lemma \ref{lemma: expr sqrt} we obtain the expression
    \begin{align}
    \label{eq: comm x^2}
        [\langle x \rangle^2 , \langle \nabla \rangle^s] \langle x \rangle^{-2} \langle \nabla \rangle^{-s} =& \int_0^\infty u^{s/2}R_{-\Delta} (u) (\Delta (x^2) + 2\nabla(x^2)\cdot \nabla)R_{-\Delta} (u)\,  \langle x \rangle^{-2} \langle \nabla \rangle^{-s} \, \mathrm{d}u
    \end{align}
    where for a non negative selfadjoint operator $A$ we define $R_{A} (u) = (A+1 + u)^{-1}$. For the second term in the integral we can argue as follows: 
    \begin{align*}
        \Big\| \int_0^\infty u^{s/2}R_{-\Delta} (u)& \nabla(x^2)\cdot \nabla R_{-\Delta} (u)\,  \langle x \rangle^{-2} \langle \nabla \rangle^{-s} \, \mathrm{d}u \Big\|_{\mathcal{B}(L^p)} \\
         =& 2 \Big\|\sum_{k=1}^d\int_0^\infty u^{s/2}  R_{-\Delta} (u)^2\, x_k \langle x \rangle^{-2}\partial_k \langle \nabla \rangle^{-s} \, \mathrm{d}u\\
          & \hspace{1cm}+  \int_0^\infty u^{s/2}R_{-\Delta} (u)^2\,x_k \partial_k( \langle x \rangle^{-2}) \langle \nabla \rangle^{-s} \, du\Big\|_{\mathcal{B}(L^p)}\\
          \lesssim & \int_0^\infty  \frac{u^{s/2}}{(1+u)^2} \, \mathrm{d}u \lesssim 1
    \end{align*}
    where we applied Lemma \ref{lemma: resolvent Lp}. Since $\Delta (x^2) = 2d$ the corresponding term in \eqref{eq: comm x^2} is treated applying directly Lemma \ref{lemma: resolvent Lp}. This and the previous inequality, together with expression \eqref{eq: comm x^2} prove the statement for $s \in [0, 2)$ and $\gamma =2$. 
    
    \item $s, \gamma \in [0, 2)$. 
    
    We reason analogously, applying Lemma \ref{lemma: expr sqrt} to obtain the expression
    \begin{align}
        [\langle x \rangle^\gamma , \langle \nabla \rangle^s] =& \int_0^\infty t^{\gamma/2}\,u^{s/2}[ (x^2 + 1+t)^{-1}, (-\Delta + 1 + u)^{-1}]\, \mathrm{d}u\, \mathrm{d}t \nonumber\\
         =& \int_0^\infty t^{\gamma/2}\,u^{s/2} R_{x^2}(t)R_{-\Delta} (u) (\Delta (x^2) + 2\nabla(x^2)\cdot \nabla)R_{-\Delta} (u)R_{x^2}(t)\, \mathrm{d}u\, \mathrm{d}t \label{comp: int commut}. 
    \end{align}
     If $0\leq s<1$ we rewrite the contribution of $\nabla(x^2)\cdot \nabla = 2x\cdot \nabla $ as 
     \begin{align}
        R_{x^2}(t)R_{-\Delta} (u) \, x\cdot \nabla R_{-\Delta} (u)R_{x^2}(t) & \nonumber\\
        =R_{x^2}(t)R_{-\Delta} (u)^2  \big(\,& \sum_{k=1}^d    \partial_k  R_{x^2}(t) x_k + dR_{x^2}(t)-2R_{-\Delta} (u) \Delta R_{x^2}(t) \big). \label{comp: commut x, grad}
    \end{align}
    Using the inequality
    \begin{equation*}
        \|R_{x^2}(t)x_k \langle x \rangle^{-\gamma}\|_{\mathcal{B}(L^p)} \lesssim \begin{cases}
            (1 + t)^{-1} & 1\leq \gamma <2\\
            (1 + t)^{-1/2} & 0 \leq \gamma <1
        \end{cases}
    \end{equation*}
as well as Lemma \ref{lemma: resolvent Lp} to estimate \eqref{comp: commut x, grad} we obtain
    \begin{align}
        \Big\| \Big (\int_0^\infty t^{\gamma/2}\,u^{s/2} R_{x^2}(t)R_{-\Delta} (u) \, \nabla(x^2)\cdot \nabla&\,R_{-\Delta} (u)R_{x^2} (t)\, \mathrm{d}u\, \mathrm{d}t\Big ) \langle x \rangle^{-\gamma} \langle \nabla \rangle^{-s}\Big\|_{\mathcal{B}(L^p)}\nonumber\\
        \lesssim  & \begin{cases}
            \int_0^\infty \frac{t^{\gamma/2}}{(1+t)^2} \mathrm{d}t  \int_0^\infty \frac{u^{s/2}}{(1+u)^{3/2}}  \mathrm{d}u & 1\leq \gamma <2\nonumber\\
            \int _0^\infty \frac{t^{\gamma/2}}{(1+t)^{3/2}} \mathrm{d}t  \int_0^\infty \frac{u^{s/2}}{(1+u)^{3/2}}  \mathrm{d}u & 0\leq \gamma <1\nonumber\\
        \end{cases}\nonumber\\
        \lesssim &\, 1  \label{comp: bound [x^g, grad^s]}
    \end{align}
    since we assumed $s<1$. If $1\leq s <2$ we commute $\partial_k$ to the right in \eqref{comp: commut x, grad} obtaining 
    \begin{align*}
       R_{x^2}(t)R_{-\Delta} (u) \, x\cdot &\nabla R_{-\Delta} (u)R_{x^2}(t) \langle x \rangle^{-\gamma} \langle \nabla \rangle^{-s}\\
        =  R_{x^2}(t) R_{-\Delta} (u)^2 \big( \,& \sum_{k=1}^d R_{x^2}(t) x_k \langle x \rangle^{-\gamma} \partial_k \langle \nabla \rangle^{-s}  +  \partial_k (R_{x^2}(t) x_k \langle x \rangle^{-\gamma} )  \langle \nabla \rangle^{-s}\\
         & + d R_{x^2}(t) \langle x \rangle^{-\gamma} \langle \nabla \rangle^{-s} -2R_{-\Delta} (u) \Delta R_{x^2}(t)  \langle x \rangle^{-\gamma} \langle \nabla \rangle^{-s}\big). 
  \end{align*}
    Since $ \|\partial_k (R_{x^2}(t) x_k \langle x \rangle^{-\gamma} ) \|_{\mathcal{B}(L^p)}\lesssim (1+t)^{-1}$ and applying again Lemma \ref{lemma: resolvent Lp} as before we obtain 
    \begin{align*}
      \Big\| \Big (\int_0^\infty u^{\gamma/2}\,t^{s/2} R_{x^2}(u)R_{-\Delta} (t) \, &\nabla(x^2)\cdot \nabla\,R_{-\Delta} (t)R_{x^2} (u)\, \mathrm{d}u\, \mathrm{d}t\Big ) \langle x \rangle^{-\gamma} \langle \nabla \rangle^{-s}\Big\|_{\mathcal{B}(L^p)}\\
      \lesssim  & \begin{cases}
            \int_0^\infty \frac{t^{\gamma/2}}{(1+t)^2} \mathrm{d}t  \int_0^\infty \frac{u^{s/2}}{(1+u)^{2}}  \mathrm{d}u & 1\leq \gamma <2\\
            \int _0^\infty \frac{t^{\gamma/2}}{(1+t)^{3/2}} \mathrm{d}t  \int_0^\infty \frac{u^{s/2}}{(1+u)^{2}}  \mathrm{d}u & 0\leq \gamma <1\\
        \end{cases}\\
        \lesssim & 1 .
    \end{align*}
    The term in \eqref{comp: int commut} involving $\Delta (x^2)$ can be bounded in an analogous way.
    \end{itemize}
     We have therefore proved the statement holds for $\gamma, s \in [0, 2]$, we now use this fact to prove the case $s \in (2, 4]$. Indeed, we need to prove that 
     \begin{equation*}
         [\langle x \rangle^\gamma , \langle \nabla \rangle^s]\langle x \rangle^{-\gamma} \langle \nabla \rangle^{-s} = \langle x \rangle^\gamma \langle \nabla \rangle^s \langle x \rangle^{-\gamma }\langle \nabla \rangle^{-s} - \mathrm{Id}
     \end{equation*}
     is bounded for $\gamma \in [0,2]$ and $s \in (2, 4]$. We rewrite the first operator in the sum as 
     \begin{align}
     \label{comp: comm s>2}
         \langle x \rangle^\gamma \langle \nabla \rangle^s \langle x \rangle^{-\gamma }\langle \nabla \rangle^{-s} =& \langle x \rangle^\gamma \langle \nabla \rangle^{s-2} \langle x \rangle^{-\gamma }\langle \nabla \rangle^{-s+2} + \langle x \rangle^\gamma \langle \nabla \rangle^{s-2} [-\Delta , \langle x \rangle^{-\gamma}]\langle \nabla \rangle^{-s}
     \end{align}
     where the first operator is bounded, since $\gamma, s-2 \in [0,2]$. For the second operator, let us first consider $ \Delta (\langle x \rangle^{-\gamma}) = c \langle x \rangle^{-\gamma-2} + c' \langle x \rangle^{-\gamma-4} $. Let $ j =2,4$, then we have
     \begin{align*}
          \langle x \rangle^\gamma \langle \nabla \rangle^{s-2} \langle x \rangle^{-\gamma-j} \langle \nabla \rangle^{-s} = \big(\langle x \rangle^\gamma \langle \nabla \rangle^{s-2} \langle x \rangle^{-\gamma}\langle \nabla \rangle^{-s +2}\big)\,\big( \langle \nabla \rangle^{s -2} \langle x \rangle^{-j}  \langle \nabla \rangle^{-s} \big)
     \end{align*}
     where again $\langle x \rangle^\gamma \langle \nabla \rangle^{s-2} \langle x \rangle^{-\gamma}\langle \nabla \rangle^{-s +2}$ is bounded from the previous step and it holds 
     \begin{align*}
         \langle \nabla \rangle^{s -2} \langle x \rangle^{-2}  \langle \nabla \rangle^{-s}  = &  \langle x \rangle^{-2}(\langle x \rangle^{2}\langle \nabla \rangle^{s -2} \langle x \rangle^{-2}  \langle \nabla \rangle^{-s+2}) \langle \nabla \rangle^{-2} ,
     \end{align*}
     which is again a composition of bounded operators on $L^p$, and similarly for 
     \begin{align*}
         \langle \nabla \rangle^{s -2} \langle x \rangle^{-4}  \langle \nabla \rangle^{-s}  = &[\, \langle x \rangle^{-2}(\langle x \rangle^{2}\langle \nabla \rangle^{s -2} \langle x \rangle^{-2}  \langle \nabla \rangle^{-s+2})\,]^2 \langle \nabla \rangle^{-2}.
     \end{align*}
          Hence $ \langle x \rangle^\gamma \langle \nabla \rangle^{s-2} \Delta (\langle x \rangle^{-\gamma})\langle \nabla \rangle^{-s}$ is a bounded operator on $L^p$ and the same can be proved for $ \langle x \rangle^\gamma \langle \nabla \rangle^{s-2} \nabla (\langle x \rangle^{-\gamma})\cdot \nabla\, \langle \nabla \rangle^{-s}$ with similar computations. These facts, combined with the expression given in \eqref{comp: comm s>2} prove the statement for $s \in (2, 4]$. For larger $s$ we reason analogously. 
    \end{proof}

	\section{Time decay for the free propagator}
	\label{app: pointw dec} In this appendix we 
	consider the linear flow $e^{-it\langle\nabla \rangle }$. First, in Subsection \ref{subsec:Lp-time-decay}, we recall the pointwise time-decay estimates in $L^p$-spaces that were used in Section \ref{sec:global-short}. We also formulate pointwise time-decay estimates in Besov-Lorentz spaces. In Subsection \ref{subsec:Strichartz}, we deduce from the latter the Strichartz estimates in Lorentz spaces that were used in Section \ref{sec:global-long}. Finally, in Subsection \ref{subsec:pointwiseLp}, we derive the pointwise estimates in weighted $L^p$ spaces that were used in Section \ref{sec:min-vel}.
	
	\subsection{Pointwise estimates in $L^p$ spaces}\label{subsec:Lp-time-decay}
	We will use the following lemma due to H\"ormander \cite{Hormander_book_hyperbolic}. In what follows we denote by $\mathcal{F}$ the unitary Fourier transform on $L^2$.
	
	\begin{lemma}
		\label{lemma: ancona f}
		Let $\phi :\R^ d \to \mathbb{C}$ be such that $\mathcal F \phi  \in C^\infty (\R^d)$ and 
		\begin{equation*}
			|\nabla^\alpha\mathcal F \phi(\xi) |\leq C_\alpha \langle \xi \rangle^{-\frac{d}2 - 1 - |\alpha|}	\quad \forall\, \xi \in \R^d, \alpha\in \N^d,
		\end{equation*}
		for some $C_\alpha>0$. Then 
		\begin{equation}\label{eq:bound-Hormander}
			\big|e^{-it\langle \nabla \rangle } \phi(x)\big| \leq c_m(|t| + |x|)^{-\frac{d}2}\big(1+(|x|^2-t^2)\mathbf{1}_{|x|\ge t}\big)^{-m},
		\end{equation}
		for all $ t \in \R$, $x\in\R^d$, $m\in\mathbb{N}$ and for some $c_m>0$. In particular, for all $2\le p\le\infty$, $\gamma\ge0$ and $t\neq0$,
		\begin{equation}\label{eq:bound-phi-Lp}
		    \| e^{-it\langle\nabla\rangle}\phi\|_{L^p_\gamma}\lesssim\langle t\rangle^{\gamma-\frac{d}{2r}},
		\end{equation}
		with $p=\frac{2r}{r-1}$.
	\end{lemma}
	
	\begin{proof}
	    The estimate \eqref{eq:bound-Hormander} is proven in \cite{Hormander_book_hyperbolic}*{Theorem 7.2.1}. 
	    \eqref{eq:bound-phi-Lp} directly follows from \eqref{eq:bound-Hormander}.
	\end{proof}

	We introduce the notations
	\begin{equation*}
	    \phi_d:=\mathcal{F}^{-1}\big(\langle\xi\rangle^{-s_d}\big),\qquad s_d:=\frac{d}{2}+1.
	\end{equation*}
	The following provides a useful representation of $e^{-it\langle\nabla\rangle}f$ for $f$ regular enough. The proof is straightforward.
	
	\begin{lemma}\label{lm:repr_with_phid}
	    Let $f\in H^{s_d,p'}$ for some $1\le p'\le 2$. Then, for all $t\in\mathbb{R}$,
	    \begin{equation}\label{eq:convol_nabla}
	        e^{-it\langle\nabla\rangle}f=\big(e^{-it\langle\nabla\rangle}\phi_d\big)*\big(\langle\nabla\rangle^{s_d}f\big).
	    \end{equation}
	\end{lemma}
	
	\begin{proof}
	    We rewrite $f$ as
		\begin{align*}
			f  &= \langle \nabla \rangle^{-s_d}\langle \nabla \rangle^{s_d} f= \mathcal{F}^{-1}[\langle \xi \rangle^{-s_d}\mathcal{F}(\langle \nabla\rangle^{s} f )] 
			= \phi_d * (\langle \nabla\rangle^{s_d} f ).
		\end{align*}
		Hence,
		observing that the convolution product is well-defined by Lemma \ref{lemma: ancona f}, we obtain \eqref{eq:convol_nabla}.
	\end{proof}
	
	Using the previous two lemmas, we have the following estimate.

	\begin{lemma}
	\label{lemma: free t dec interp}
		Let $s\ge s_d=d/2+1$ and $1\le r \le\infty$. Then, for all $f \in H^{s,p'}\cap H^{s}$ and $t\ge0$, 
		\begin{equation*}
			\|e^{-it\langle \nabla \rangle} f\|_{L^{p}\cap L^\infty} \lesssim \langle t\rangle^{-\frac{d}{2r}}\| f\|_{H^{s,p'}\cap H^{s}},
		\end{equation*}
		with $p=\frac{2r}{r-1}$.
		\end{lemma}

	\begin{proof}
Since $H^{s,p'}\cap H^{s}\subset H^{s_d,p'}\cap H^{s_d}$, it suffices to prove the lemma for $s=s_d$. For $t\le1$, using Sobolev's embedding $H^{s}\hookrightarrow L^p\cap L^\infty$ (since $p\ge2$ and $s>\frac{d}2$) and the unitarity of $e^{-it\langle\nabla\rangle}$ in $H^{s}$, we have
\begin{equation}\label{eq:estim_Lp_1}
    \|e^{-it\langle \nabla \rangle} f\|_{L^{p}\cap L^\infty}\lesssim \|f\|_{H^{s}}.
\end{equation}
		For $t>1$, using Lemma \ref{lm:repr_with_phid} and the fact that $\phi_d$ satisfies the conditions of Lemma \ref{lemma: ancona f}, we obtain
		\begin{align}
		    \|e^{-it\langle \nabla \rangle} f\|_{ L^\infty}&=\big\|\big(e^{-it\langle\nabla\rangle}\phi_d\big)*\big(\langle\nabla\rangle^{s}f\big)\big\|_{L^\infty}\notag\\
		    &\lesssim\big\|e^{-it\langle\nabla\rangle}\phi_d\big\|_{L^p}\big\|\langle\nabla\rangle^{s}f\big\|_{L^{p'}}\lesssim| t|^{-\frac{d}{2r}}\|f\|_{H^{s,p'}}.\label{eq:estim_Lp_2}
		\end{align}
Interpolating the inequality  
		\begin{equation*}
			\|e^{-it\langle \nabla \rangle} f\|_{L^\infty} \lesssim |t|^{ -d/2 }\|\langle \nabla\rangle^{s} f\|_{L^{1}},
		\end{equation*}
(obtained from the previous bound for $p=\infty$), with 
		\begin{align*}
			\|e^{-it\langle \nabla \rangle} f\|_{L^2} = \|f\|_{L^2}\leq  \|\langle \nabla\rangle^{s} f\|_{L^{2}},
		\end{align*}
		we obtain, using e.g. \cite{grafakos}*{Theorem 1.3.4}, 
		\begin{equation*}
			\|e^{-it\langle \nabla \rangle} f\|_{L^a} \lesssim |t|^{ -d/2(1-\theta) } \|\langle \nabla\rangle^{s} f\|_{L^{b}}, \quad \frac{1}{a} = \frac{\theta}{2}, \quad\frac{1}{b} = \frac{1-\theta}{1} + \frac{\theta}{2}, \quad 0\le\theta\le1.
		\end{equation*}
		Setting $\theta = 1-\frac{1}{r}$ yields
		\begin{align}
			\label{comp: bound 1/t}
			\|e^{-it\langle \nabla \rangle} f\|_{L^{p}} \lesssim |t|^{-\frac{d}{2r}}  \|\langle \nabla \rangle^{s} f\|_{L^{p'}}
		\end{align}
	Combining \eqref{eq:estim_Lp_1}, \eqref{eq:estim_Lp_2} and \eqref{comp: bound 1/t} leads to the statement of the lemma.
	\end{proof}

One can formulate a more precise time-decay estimate in Besov-Lorentz spaces as follows. For $s\in\mathbb{R}$, $1\le p,q<\infty$ and $1\le r\le\infty$, we denote by $B^s_{(p,r),q}$ the Besov-Lorentz space associated to the norm
\begin{equation*}
    \|f\|_{B^s_{(p,r),q}}=\Big(\sum_{k\in\mathbb{N}_0}2^{ksq}\|\Lambda_kf\|^q_{L^{p,r}}\Big)^{\frac1q},
\end{equation*}
where $(\Lambda_k)_{k\in\mathbb{N}_0}$ stands for a Littlewood-Paley decomposition, see \cite{SeegerTrebels} for more details.

\begin{lemma}\label{lem:decay_Besov}
Let $2\le b\le\infty$ and $0\le\beta\le1$. Then
\begin{align*}
\|e^{-it\langle\nabla\rangle}f\|_{B'}\lesssim |t|^{-(d-1+\beta)(\frac12-\frac1b)}\|f\|_B,
\end{align*}
where $B$ stands for the Besov-Lorentz space $B=B_{(b',2),2}^s$, and $B'=B_{(b,2),2}^{-s}$, with $s=\frac12(d+1+\beta)(\frac12-\frac1b)$.
\end{lemma}

\begin{proof}
   It suffices to modify the end of the proof of \cite{GinibreVelo_Timedecay}*{Lemma 2.1} (see also \cite{Brenner})  in a straightforward way, using real interpolation $L^{b,2}=[L^1,L^\infty]_{1/b',2}$.
\end{proof}

\subsection{Strichartz estimates}\label{subsec:Strichartz}

Lemma \ref{lem:decay_Besov} implies the following Strichartz estimates. Compared to the estimates previously used in the literature, see e.g. \cites{Machihara_et_al,dancona_fan08}, a difference here is that we need a Strichartz estimate in Lorentz space. On the other hand, we do not consider the endpoint case $a=2$, $b=2(d-1+\beta)(d-3+\beta)^{-1}$ since we do not need it in our application.

In the following statement, $\beta=0$ corresponds to the usual wave admissible pairs and $\beta=1$ to Schrödinger admissible pairs. We only consider here the estimates that were used in Section \ref{sec:global-long}, we do not state more general estimates that could be deduced from Lemma \ref{lem:decay_Besov}.

\begin{prop}\label{prop:Strichartz-Lorentz}
    Let $a<2\le\infty$, $0\le\beta\le1$ and $2\le b< 2(d-1+\beta)(d-3+\beta)^{-1}$ be such that
    \begin{equation*}
		\frac{2}{a} + \frac{d-1+\beta}{b} = \frac{d-1+\beta}{2}.
	\end{equation*}
	Then
	\begin{equation}\label{eq:Strichartz1}
       \|e^{-it\langle\nabla\rangle}f\|_{L^a_tL^{b,2}_x}\lesssim \|f\|_{H^s},
   \end{equation}
   for any $s\ge\frac12+\frac1a-\frac1b$    and
   	\begin{equation}\label{eq:Strichartz2}
		\Big\|\int_{0}^t e^{-i(t-\tau)\langle \nabla\rangle} f(\tau,x) \mathrm{d}\tau \Big\|_{L^a_t\, L^{b,2}_x} \lesssim  \|f\|_{L^{1}_t H^{s}_x}.
	\end{equation}
\end{prop}

\begin{proof}
   Recall that $B$ stands for the Besov-Lorentz space $B_{(b',2),2}^s$.  Applying \cite{KeelTao}*{Theorem 10.1} (with $B_0=H=L^2$ and $B_1=B$), we obtain from Lemma \ref{lem:decay_Besov} that
   \begin{equation*}
       \|e^{-it\langle\nabla\rangle}f\|_{L^a_tB'}\lesssim \|f\|_{L^2}, \qquad a=2\Big[(d-1+\beta)(\frac12-\frac1b)\Big]^{-1},
   \end{equation*}
   provided that $a>2$. Equivalently
   \begin{equation*}
       \|e^{-it\langle\nabla\rangle}f\|_{L^a_tB_{(b,2),2}^0}\lesssim \|f\|_{H^s},
   \end{equation*}
   with $s=\frac12(d+1+\beta)(\frac12-\frac1b)$. Since, by \cite{SeegerTrebels}*{Theorem 1.1}, we have that $B^0_{(b,2),2}\hookrightarrow L^{b,2}$, we deduce from the previous inequality that \eqref{eq:Strichartz1} holds. The estimate \eqref{eq:Strichartz2} follows in the same way.
\end{proof}

\subsection{Pointwise estimates in weighted $L^p$ spaces}\label{subsec:pointwiseLp}

The next lemma was used in Section \ref{sect: wave op}. Recall the notation $\Theta:=-i\nabla\langle\nabla\rangle^{-1}$.

\begin{lemma}\label{lm:xgamma-lin}
Let $0\leq \gamma\leq 2$. Then for all $\varphi\in L^2_\gamma$,
\begin{equation}\label{eq:xgamma-lin}
    \big\| e^{-it\langle\nabla\rangle}\varphi\|_{L^2_\gamma}\lesssim \|\varphi\|_{L^2_\gamma}+t^\gamma\|\varphi\|_{L^2}.
\end{equation}
\end{lemma}	
\begin{proof}
From an explicit computation we obtain the expression
\begin{equation*}
e^{it\langle\nabla\rangle}xe^{-it\langle\nabla\rangle}=x+t\Theta,
\end{equation*}
which gives the statement when $\gamma = 1$. Moreover, this relation further implies $x^2e^{-it\langle\nabla\rangle} = e^{-it\langle\nabla\rangle} (x+t\Theta)^2$. Since $\langle x \rangle^2 = 1 + x^2$, we have
\begin{align}
\label{comp: x^2 free prop}
    \big\|e^{-it\langle\nabla\rangle}\varphi\|_{L^2_2}\le \|\varphi\|_{L^2} + \|(x+t\Theta)^2\varphi\|_{L^2}, 
\end{align}
with $(x+t\Theta)^2 = x^2  + t^2 \Theta^2 + t(\Theta\cdot x  + x\cdot \Theta)=x^2  + t^2 \Theta^2 + 2t\Theta\cdot x  + t[x,\Theta]$,
where we use the notation $[x,\Theta] = \sum_{k=1}^d [x_k,\Theta_k]$. 
Since $\Theta_k$ and $[x_k,\Theta_k]$ are bounded operators in $L^2$, we deduce that
\begin{align}
    \big\|e^{-it\langle\nabla\rangle}\varphi\|_{L^2_2} \lesssim \|\langle x\rangle^2\varphi\|_{L^2}+t^2\|\varphi\|_{L^2}\lesssim\|(\langle x \rangle^2 + t^2 ) \varphi\|_{L^2}. 
\end{align}
Hence we have the bounded operators 
\begin{align*}
    e^{-it\langle\nabla\rangle} :L^2(\mathrm{d}x) \to L^2(\mathrm{d}x) \qquad e^{-it\langle\nabla\rangle}: L^2 ((\langle x \rangle^2 + t^2 )^2 \mathrm{d}x) \to  L^2 (\langle x \rangle^4 \mathrm{d}x)
\end{align*}
By an application of the Riesz-Thorin theorem (see \cite{stein_weiss}*{Theorem 2.11}) we obtain that $e^{-it\langle\nabla\rangle}: L^2 ((\langle x \rangle^2 + t^2 )^{2\theta} \mathrm{d}x) \to  L^2 (\langle x \rangle^{4\theta} \mathrm{d}x)$ is a bounded operator for all $\theta \in (0,1)$, from which the statement follows since $(\langle x \rangle^2 + t^2 )^{\theta} \lesssim \langle x \rangle^{2\theta} + t^{2\theta}$. 
\end{proof}	

\begin{lemma}\label{lm:linear-bound-weighted-Lp}
Let $s\ge d/2+1$, $1\le r \le\infty$ and $0\le\gamma\le2$. Then, for all $f \in H^s_{\gamma}\cap H^{s,p'}$ and $t\ge0$,
    \begin{align*}
     \| e^{-it\langle\nabla\rangle} f\|_{L^p_\gamma\cap L^\infty_\gamma} \lesssim  \| f \|_{H^{s}_\gamma}  + \langle t \rangle^{\gamma-\frac{d}{2r}} \|  f\|_{H^{s,p'}\cap H^s},
\end{align*}
		with $p=\frac{2r}{r-1}$.
\end{lemma}

\begin{proof}
For $t\le 1$, it suffices to use Sobolev's embedding $H^s\hookrightarrow L^p\cap L^\infty$ together with Corollary \ref{cor:weighted-sobolev} and Lemma \ref{lm:xgamma-lin}. Suppose that $t>1$. Using Lemma \ref{lm:repr_with_phid}, we write
\begin{align*}
        \big|\langle x\rangle^\gamma e^{-it\langle\nabla\rangle}f\big|&=\langle x\rangle^\gamma\big|\big(e^{-it\langle\nabla\rangle}\phi_d\big)*\big(\langle\nabla\rangle^{s}f\big)\big|\\
        &\lesssim\big(\langle x\rangle^\gamma\big| e^{-it\langle\nabla\rangle}\phi_d\big|\big)*\big|\langle\nabla\rangle^{s}f\big|+\big| e^{-it\langle\nabla\rangle}\phi_d\big|*\big(\langle x\rangle^\gamma\big|\langle\nabla\rangle^{s}f\big|\big).
\end{align*}
This implies, for the $L^\infty$-norm,
\begin{align*}
        \big\| e^{-it\langle\nabla\rangle}f\big\|_{L^\infty_\gamma}
        &\lesssim\big\| e^{-it\langle\nabla\rangle}\phi_d\big\|_{L^p_\gamma}\big\|f\big\|_{H^{s,p'}}+\big\| e^{-it\langle\nabla\rangle}\phi_d\big\|_{L^2}\big\|f\big\|_{H^s_\gamma}\\
        &\lesssim\langle t\rangle^{\gamma-\frac{d}{2r}}\big\|f\big\|_{H^{s,p'}}+\big\|f\big\|_{H_\gamma^s},
\end{align*}
where we used Lemma \ref{lemma: ancona f} in the second inequality. 

Now we consider the $\|\cdot\|_{L^p}$ norm for $2\le p<\infty$.  First, for $\gamma=2$, we have 
\begin{align*}
    \|x^2f\|_{L^p}&\le\big\|(x^2+t^2\Theta^2)f\|_{L^p}+t^2\|\Theta^2f\|_{L^p}\\
    &\le\big\|(x-t\Theta)^2f\big\|_{L^p}+2t\|(x\cdot\Theta+\Theta\cdot x)f\|_{L^p}+t^2\|\Theta^2f\|_{L^p}\\
    &\le\big\|(x-t\Theta)^2f\big\|_{L^p}+2t\|(\Theta\cdot x)f\|_{L^p}+t\|[x,\Theta]f\|_{L^p}+ t^2\|\Theta^2f\|_{L^p}\\
    &\lesssim\big\|(x-t\Theta)^2f\big\|_{L^p}+t^2\|f\|_{L^p},
\end{align*}
where in the last inequality we used that $\Theta$, $[x,\Theta]$ and $\Theta^2=\mathrm{Id}-(-\Delta+1)^{-1}$ are bounded operators from $L^p$ to $L^p$ by Lemma \ref{lemma: resolvent Lp} (using that  $2\le p<\infty$ for $\Theta$). Using, similarly as before, that $\langle x-t\Theta\rangle^2e^{-it\langle\nabla\rangle} = e^{-it\langle\nabla\rangle} \langle x\rangle^2$, this yields
\begin{align*}
        \big\| e^{-it\langle\nabla\rangle}f\big\|_{L^p_2}&\lesssim   \big\|\langle x-t\Theta \rangle^2 e^{-it\langle\nabla\rangle}f\big\|_{L^p}+t^2   \big\| e^{-it\langle\nabla\rangle}f\big\|_{L^p}\\
        &=\big\|e^{-it\langle\nabla\rangle}\langle x\rangle^2f\big\|_{L^p}+t^2\big\| e^{-it\langle\nabla\rangle}f\big\|_{L^p}\\
        &\lesssim\big\|f\big\|_{H_2^s}+|t|^{2-\frac{d}{2r}}\big\| f\big\|_{H^{s,p'}},
\end{align*}
where we used Sobolev's embedding $H^s\hookrightarrow L^p$, Lemma \ref{lemma: weight sob norm} and \eqref{comp: bound 1/t} in the last inequality. In order to interpolate later on, we rewrite this as
\begin{align}\label{eq:first_ineq_interpo}
        \big\|\langle x\rangle^2 e^{-it\langle\nabla\rangle}\langle\nabla\rangle^{-s}f\big\|_{L^p}&\lesssim \|f\|_{B_0\cap B_1},
\end{align}
where we have set
\begin{equation*}
    B_0:=L^2_2= L^2(\langle x\rangle^4\mathrm{d}x),\qquad B_1:=L^{p'}(| t|^{p'(2-\frac{d}{2r})}\mathrm{d}x).
\end{equation*}

Next we interpolate \eqref{eq:first_ineq_interpo} with the inequality obtained for $\gamma=0$. The argument may be well-known, we provide some details for the convenience of the reader. From the proof of Lemma \ref{lemma: free t dec interp} and Sobolev's embedding $H^s\hookrightarrow L^p$, we deduce that 
		\begin{equation*}
			\|e^{-it\langle \nabla \rangle} f\|_{L^{p}} \lesssim \min\big( |t|^{-\frac{d}{2r}}\| f\|_{H^{s,p'}} , \|f\|_{H^s}\big),
		\end{equation*}
which in turn implies that
\begin{equation}\label{eq:2nd_ineq_interpo}
			\|e^{-it\langle \nabla \rangle}\langle\nabla\rangle^{-s} f\|_{L^{p}} \lesssim \|f\|_{A_0+A_1},
		\end{equation}
where we have set
\begin{equation*}
    A_0:=L^2(\mathrm{d}x),\qquad A_1:=L^{p'}(| t|^{-\frac{d}{2r}p'}\mathrm{d}x).
\end{equation*}
By real interpolation, we deduce from \eqref{eq:first_ineq_interpo}, \eqref{eq:2nd_ineq_interpo} and \cite{bergh_lofstrom}*{Theorem 5.4.1} that, for all $0\le\gamma\le2$,
\begin{align}\label{eq:result_interpo}
        \big\|\langle x\rangle^\gamma e^{-it\langle\nabla\rangle}\langle\nabla\rangle^{-s}f\big\|_{L^p}&\lesssim \|f\|_{[B_0\cap B_1;A_0+A_1]_{\gamma,p}}.
\end{align}
Now since $A_0$, $A_1$, $B_0$, $B_1$ are all Banach function lattices, we have
\begin{align*}
    [B_0\cap B_1;A_0+A_1]_{\gamma,p}&=[B_0;A_0+A_1]_{\gamma,p}\cap [B_1;A_0+A_1]_{\gamma,p}\\
    &=\big([B_0;A_0]_{\gamma,p}+[B_0;A_1]_{\gamma,p}\big)\cap \big([B_1;A_0]_{\gamma,p}+[B_1;A_1]_{\gamma,p}\big),
\end{align*}
(see \cite{Maligranda} for the commutativity between interpolation and intersection, and use duality \cite{bergh_lofstrom}*{Theorem 2.7.1} for the commutativity between interpolation and sum). Hence
\begin{equation}
    \|f\|_{[B_0\cap B_1;A_0+A_1]_{\gamma,p}}\lesssim \|f\|_{[B_0;A_0]_{\gamma,p}}+\|f\|_{[B_1;A_1]_{\gamma,p}}. \label{eq:interpo-a}
\end{equation}
Since $p\ge2$, using again \cite{bergh_lofstrom}*{Theorem 5.4.1},
\begin{equation}
    \|f\|_{[B_0;A_0]_{\gamma,p}}\lesssim\|f\|_{[B_0;A_0]_{\gamma,2}}\lesssim\|\langle x\rangle^\gamma f\|_{L^2}, \label{eq:interpo-b}
\end{equation}
and likewise, since $p\ge p'$,
\begin{equation}
    \|f\|_{[B_1;A_1]_{\gamma,p}}\lesssim\|f\|_{[B_1;A_1]_{\gamma,p'}}\lesssim|t|^{\gamma-\frac{d}{2r}}\|f\|_{L^{p'}}.\label{eq:interpo-c}
\end{equation} 
Putting together \eqref{eq:result_interpo} and \eqref{eq:interpo-a}--\eqref{eq:interpo-c} concludes the proof.
\end{proof}
	\bibliography{biblio_eff}

\begin{bibdiv}
\begin{biblist}

\bib{ArbunichFaupinPusateriSigal23}{article}{
      author={Arbunich, J.},
      author={Faupin, J.},
      author={Pusateri, F.},
      author={Sigal, I.~M.},
       title={Maximal speed of quantum propagation for the {H}artree equation},
        date={2023},
        ISSN={0360-5302,1532-4133},
     journal={Comm. Partial Differential Equations},
      volume={48},
      number={4},
       pages={542\ndash 575},
         url={https://doi.org/10.1080/03605302.2023.2183408},
}

\bib{ArbunichPusateriSigalSoffer21}{article}{
      author={Arbunich, J.},
      author={Pusateri, F.},
      author={Sigal, I.~M.},
      author={Soffer, A.},
       title={Maximal speed of quantum propagation},
        date={2021},
        ISSN={1573-0530},
     journal={Letters in Mathematical Physics},
      volume={111},
      number={3},
         url={http://dx.doi.org/10.1007/s11005-021-01397-y},
}

\bib{bergh_lofstrom}{book}{
      author={Bergh, J.},
      author={L\"ofstr\"om, J.},
       title={Interpolation spaces. {A}n introduction},
      series={Grundlehren der Mathematischen Wissenschaften},
   publisher={Springer-Verlag, Berlin-New York},
        date={1976},
      volume={No. 223},
}

\bib{BonyFaupinSigal12}{article}{
      author={Bony, J.-F.},
      author={Faupin, J.},
      author={Sigal, I.~M.},
       title={Maximal velocity of photons in non-relativistic {QED}},
        date={2012},
        ISSN={0001-8708,1090-2082},
     journal={Adv. Math.},
      volume={231},
      number={5},
       pages={3054\ndash 3078},
         url={https://doi.org/10.1016/j.aim.2012.07.019},
}

\bib{Brenner}{article}{
      author={Brenner, P.},
       title={On scattering and everywhere defined scattering operators for
  nonlinear {Klein}-{Gordon} equations},
        date={1985},
        ISSN={0022-0396},
     journal={J. Differ. Equations},
      volume={56},
       pages={310\ndash 344},
}

\bib{BFG2a}{article}{
      author={Breteaux, S.},
      author={Faupin, J.},
      author={Grasselli, V.},
       title={Exponential decay outside of the light cone for the
  pseudo-relativistic non-autonomous {S}chrödinger equation},
        date={2025},
     journal={arXiv:},
         url={https://arxiv.org/abs/},
}

\bib{BreteauxFaupinLemmSigalYangZhang24}{article}{
      author={Breteaux, S.},
      author={Faupin, J.},
      author={Lemm, M.},
      author={Ou~Yang, D.~H.},
      author={Sigal, I.~M.},
      author={Zhang, J.},
       title={Light cones for open quantum systems in the continuum},
        date={2024},
        ISSN={0129-055X,1793-6659},
     journal={Rev. Math. Phys.},
      volume={36},
      number={9},
       pages={Paper No. 2460004, 49},
         url={https://doi.org/10.1142/S0129055X24600043},
}

\bib{BreteauxFaupinLemmSigal22}{incollection}{
      author={Breteaux, S.},
      author={Faupin, J.},
      author={Lemm, M.},
      author={Sigal, I.~M.},
       title={Maximal speed of propagation in open quantum systems},
        date={2022},
   booktitle={The physics and mathematics of {E}lliott {L}ieb---the 90th
  anniversary. {V}ol. {I}},
   publisher={EMS Press, Berlin},
       pages={109\ndash 130},
}

\bib{cazenaveWeissler_H1}{article}{
      author={Cazenave, T.},
      author={Weissler, F.~B.},
       title={The {C}auchy problem for the nonlinear {S}chr\"odinger equation
  in {$H^1$}},
        date={1988},
        ISSN={0025-2611,1432-1785},
     journal={Manuscripta Math.},
      volume={61},
      number={4},
       pages={477\ndash 494},
         url={https://doi.org/10.1007/BF01258601},
}

\bib{cazenaveWeissler_Hs}{article}{
      author={Cazenave, T.},
      author={Weissler, F.~B.},
       title={The {C}auchy problem for the critical nonlinear {S}chr\"odinger
  equation in {$H^s$}},
        date={1990},
        ISSN={0362-546X,1873-5215},
     journal={Nonlinear Anal.},
      volume={14},
      number={10},
       pages={807\ndash 836},
         url={https://doi.org/10.1016/0362-546X(90)90023-A},
}

\bib{ChoOzawa1}{article}{
      author={Cho, Y.},
      author={Ozawa, T.},
       title={On the semirelativistic {Hartree}-type equation},
        date={2006},
        ISSN={0036-1410},
     journal={SIAM J. Math. Anal.},
      volume={38},
      number={4},
       pages={1060\ndash 1074},
         url={hdl.handle.net/2115/69581},
}

\bib{ChoOzawa2}{article}{
      author={Cho, Y.},
      author={Ozawa, T.},
       title={Global solutions of semirelativistic {Hartree} type equations},
        date={2007},
        ISSN={0304-9914},
     journal={J. Korean Math. Soc.},
      volume={44},
      number={5},
       pages={1065\ndash 1078},
}

\bib{ChoOzawaSazaki}{article}{
      author={Cho, Y.},
      author={Ozawa, T.},
      author={Sasaki, H.},
      author={Shim, Y.},
       title={Remarks on the semirelativistic {Hartree} equations},
        date={2009},
        ISSN={1078-0947},
     journal={Discrete Contin. Dyn. Syst.},
      volume={23},
      number={4},
       pages={1277\ndash 1294},
}

\bib{dancona_fan08}{article}{
      author={D'Ancona, P.},
      author={Fanelli, L.},
       title={Strichartz and smoothing estimates of dispersive equations with
  magnetic potentials},
        date={2008},
        ISSN={0360-5302,1532-4133},
     journal={Comm. Partial Differential Equations},
      volume={33},
      number={4-6},
       pages={1082\ndash 1112},
         url={https://doi.org/10.1080/03605300701743749},
}

\bib{Derezinski93}{article}{
      author={Derezi{\'n}ski, J.},
       title={Asymptotic completeness of long-range {{\(N\)}}-body quantum
  systems},
        date={1993},
        ISSN={0003-486X},
     journal={Ann. Math. (2)},
      volume={138},
      number={2},
       pages={427\ndash 476},
}

\bib{derezinski_ger}{book}{
      author={Derezi\'{n}ski, J.},
      author={G\'{e}rard, C.},
       title={Scattering theory of classical and quantum {$N$}-particle
  systems},
      series={Texts and Monographs in Physics},
   publisher={Springer-Verlag, Berlin},
        date={1997},
        ISBN={3-540-62066-4},
         url={https://doi.org/10.1007/978-3-662-03403-3},
}

\bib{DerezinskiGerard99}{article}{
      author={Derezi\'nski, J.},
      author={G\'erard, C.},
       title={Asymptotic completeness in quantum field theory. {M}assive
  {P}auli-{F}ierz {H}amiltonians},
        date={1999},
        ISSN={0129-055X,1793-6659},
     journal={Rev. Math. Phys.},
      volume={11},
      number={4},
       pages={383\ndash 450},
         url={https://doi.org/10.1142/S0129055X99000155},
}

\bib{ElgartSchlein}{article}{
      author={Elgart, A.},
      author={Schlein, B.},
       title={Mean field dynamics of boson stars},
        date={2007},
        ISSN={0010-3640,1097-0312},
     journal={Comm. Pure Appl. Math.},
      volume={60},
      number={4},
       pages={500\ndash 545},
         url={https://doi.org/10.1002/cpa.20134},
}

\bib{Enss83}{article}{
      author={Enss, V.},
       title={Propagation properties of quantum scattering states},
        date={1983},
        ISSN={0022-1236},
     journal={J. Funct. Anal.},
      volume={52},
      number={2},
       pages={219\ndash 251},
         url={https://doi.org/10.1016/0022-1236(83)90083-6},
}

\bib{FaupinLemmSigal22_maxv}{article}{
      author={Faupin, J.},
      author={Lemm, M.},
      author={Sigal, I.~M.},
       title={Maximal speed for macroscopic particle transport in the
  {B}ose-{H}ubbard model},
        date={2022},
        ISSN={0031-9007,1079-7114},
     journal={Phys. Rev. Lett.},
      volume={128},
      number={15},
       pages={Paper No. 150602, 6},
         url={https://doi.org/10.1103/physrevlett.128.150602},
}

\bib{FaupinLemmSigal22_LR}{article}{
      author={Faupin, J.},
      author={Lemm, M.},
      author={Sigal, I.~M.},
       title={On {L}ieb-{R}obinson bounds for the bose-{H}ubbard model},
        date={2022},
        ISSN={0010-3616,1432-0916},
     journal={Comm. Math. Phys.},
      volume={394},
      number={3},
       pages={1011\ndash 1037},
         url={https://doi.org/10.1007/s00220-022-04416-8},
}

\bib{FaupinLemmSigalZhang25}{article}{
      author={Faupin, J.},
      author={Lemm, M.},
      author={Sigal, I.~M.},
      author={Zhang, J.},
       title={Macroscopic suppression of supersonic quantum transport},
        date={2025},
     journal={Phys. Rev. Lett.},
      volume={135},
       pages={160405},
         url={https://link.aps.org/doi/10.1103/27qs-vlrn},
}

\bib{FaupinSigal14}{article}{
      author={Faupin, J.},
      author={Sigal, I.~M.},
       title={Minimal photon velocity bounds in non-relativistic quantum
  electrodynamics},
        date={2014},
        ISSN={0022-4715,1572-9613},
     journal={J. Stat. Phys.},
      volume={154},
      number={1-2},
       pages={58\ndash 90},
         url={https://doi.org/10.1007/s10955-013-0862-1},
}

\bib{FrohlichGriesemerSchlein02}{article}{
      author={Fr\"ohlich, J.},
      author={Griesemer, M.},
      author={Schlein, B.},
       title={Asymptotic completeness for {R}ayleigh scattering},
        date={2002},
        ISSN={1424-0637,1424-0661},
     journal={Ann. Henri Poincar\'e},
      volume={3},
      number={1},
       pages={107\ndash 170},
         url={https://doi.org/10.1007/s00023-002-8614-9},
}

\bib{FrohlichJonssonLenz1}{article}{
      author={Fr\"ohlich, J.},
      author={Jonsson, B. L.~G.},
      author={Lenzmann, E.},
       title={Boson stars as solitary waves},
        date={2007},
        ISSN={0010-3616,1432-0916},
     journal={Comm. Math. Phys.},
      volume={274},
      number={1},
       pages={1\ndash 30},
         url={https://doi.org/10.1007/s00220-007-0272-9},
}

\bib{FrohlichJonssonLenz2}{article}{
      author={Fr\"ohlich, J.},
      author={Jonsson, B. L.~G.},
      author={Lenzmann, E.},
       title={Effective dynamics for boson stars},
        date={2007},
        ISSN={0951-7715,1361-6544},
     journal={Nonlinearity},
      volume={20},
      number={5},
       pages={1031\ndash 1075},
         url={https://doi.org/10.1088/0951-7715/20/5/001},
}

\bib{frohlich_lenzm07}{article}{
      author={Fröhlich, J.},
      author={Lenzmann, E.},
       title={Blowup for nonlinear wave equations describing boson stars},
        date={2007},
        ISSN={0010-3640},
     journal={Comm. Pure Appl. Math.},
      volume={60},
       pages={1691\ndash  1705},
         url={https://doi.org/10.1002/cpa.20186},
}

\bib{Gerard92}{article}{
      author={G\'erard, C.},
       title={Sharp propagation estimates for {$N$}-particle systems},
        date={1992},
        ISSN={0012-7094,1547-7398},
     journal={Duke Math. J.},
      volume={67},
      number={3},
       pages={483\ndash 515},
         url={https://doi.org/10.1215/S0012-7094-92-06719-6},
}

\bib{GinibreVelo}{article}{
      author={Ginibre, J.},
      author={Velo, G.},
       title={On a class of non linear schrödinger equations with non local
  interaction},
        date={1980},
        ISSN={1432-1823},
     journal={Math. Z.},
      volume={170},
       pages={109\ndash  136},
         url={https://doi.org/10.1007/BF01214768},
}

\bib{GinibreVelo_Timedecay}{article}{
      author={Ginibre, J.},
      author={Velo, G.},
       title={Time decay of finite energy solutions of the nonlinear
  {Klein}-{Gordon} and {Schr{\"o}dinger} equations},
        date={1985},
        ISSN={0246-0211},
     journal={Ann. Inst. Henri Poincar{\'e}, Phys. Th{\'e}or.},
      volume={43},
       pages={399\ndash 442},
         url={https://eudml.org/doc/76307},
}

\bib{Graf90}{article}{
      author={Graf, G.~M.},
       title={Asymptotic completeness forn-body short-range quantum systems: A
  new proof},
        date={1990},
        ISSN={1432-0916},
     journal={Commun. Math. Phys.},
      volume={132},
       pages={73\ndash  101},
         url={https://doi.org/10.1007/BF02278000},
}

\bib{grafakos_modern}{book}{
      author={Grafakos, L.},
       title={Modern {F}ourier analysis},
     edition={Second},
      series={Graduate Texts in Mathematics},
   publisher={Springer, New York},
        date={2009},
      volume={250},
        ISBN={978-0-387-09433-5},
         url={https://doi.org/10.1007/978-0-387-09434-2},
}

\bib{grafakos}{book}{
      author={Grafakos, L.},
       title={\textnormal{Classical Fourier Analysis}},
      series={Graduate Texts in Mathematics},
   publisher={Springer New York},
        date={2014},
        ISBN={9781493911943},
         url={https://books.google.it/books?id=94FxBQAAQBAJ},
}

\bib{grafakos_oh_2014}{article}{
      author={Grafakos, L.},
      author={Oh, S.},
       title={The {Kato}-{Ponce} inequality},
        date={2014},
        ISSN={0360-5302},
     journal={Commun. Partial Differ. Equations},
      volume={39},
      number={6},
       pages={1128\ndash 1157},
}

\bib{HerrLenzmann}{article}{
      author={Herr, S.},
      author={Lenzmann, E.},
       title={The {B}oson star equation with initial data of low regularity},
        date={2014},
        ISSN={0362-546X,1873-5215},
     journal={Nonlinear Anal.},
      volume={97},
       pages={125\ndash 137},
         url={https://doi.org/10.1016/j.na.2013.11.023},
}

\bib{HerrTesfahun}{article}{
      author={Herr, S.},
      author={Tesfahun, A.},
       title={Small data scattering for semi-relativistic equations with
  {H}artree type nonlinearity},
        date={2015},
        ISSN={0022-0396,1090-2732},
     journal={J. Differential Equations},
      volume={259},
      number={10},
       pages={5510\ndash 5532},
         url={https://doi.org/10.1016/j.jde.2015.06.037},
}

\bib{Hormander_book_hyperbolic}{book}{
      author={H\"ormander, L.},
       title={Lectures on nonlinear hyperbolic differential equations},
      series={Math\'ematiques \& Applications (Berlin)},
   publisher={Springer-Verlag, Berlin},
        date={1997},
      volume={26},
        ISBN={3-540-62921-1},
}

\bib{HunzikerSigalSoffer99}{article}{
      author={Hunziker, W.},
      author={Sigal, I.~M.},
      author={Soffer, A.},
       title={Minimal escape velocities},
        date={1999},
        ISSN={0360-5302,1532-4133},
     journal={Comm. Partial Differential Equations},
      volume={24},
      number={11-12},
       pages={2279\ndash 2295},
         url={https://doi.org/10.1080/03605309908821502},
}

\bib{KeelTao}{article}{
      author={Keel, M.},
      author={Tao, T.},
       title={Endpoint {Strichartz} estimates},
        date={1998},
        ISSN={0002-9327},
     journal={Am. J. Math.},
      volume={120},
      number={5},
       pages={955\ndash 980},
  url={muse.jhu.edu/journals/american_journal_of_mathematics/toc/ajm120.5.html},
}

\bib{Knowles-Pizzo}{article}{
      author={Knowles, A.},
      author={Pickl, P.},
       title={Mean-field dynamics: {Singular} potentials and rate of
  convergence},
        date={2010},
        ISSN={0010-3616},
     journal={Commun. Math. Phys.},
      volume={298},
      number={1},
       pages={101\ndash 138},
         url={hdl.handle.net/20.500.11850/20676},
}

\bib{Lee}{article}{
      author={Lee, J.~O.},
       title={Rate of convergence towards semi-relativistic {Hartree}
  dynamics},
        date={2013},
        ISSN={1424-0637},
     journal={Ann. Henri Poincar{\'e}},
      volume={14},
      number={2},
       pages={313\ndash 346},
}

\bib{LemarieRieusset}{book}{
      author={Lemari{\'e}-Rieusset, P.~G.},
       title={Recent developments in the {Navier}-{Stokes} problem},
      series={Chapman Hall/CRC Res. Notes Math.},
   publisher={Boca Raton, FL: Chapman \& Hall/CRC},
        date={2002},
      volume={431},
        ISBN={1-58488-220-4},
}

\bib{LemmRubilianiZhang23}{article}{
      author={Lemm, M.},
      author={Rubiliani, C.},
      author={Zhang, J.},
       title={On the microscopic propagation speed of long-range quantum
  many-body systems},
        date={2023},
     journal={arXiv:2310.14896},
         url={https://arxiv.org/abs/2310.14896},
}

\bib{LemmRubilianiZhang25}{article}{
      author={Lemm, M.},
      author={Rubiliani, C.},
      author={Zhang, J.},
       title={On the quantum dynamics of long-ranged {B}ose-{H}ubbard
  hamiltonians},
        date={2025},
     journal={arXiv:2505.01786},
         url={https://arxiv.org/abs/2505.01786},
}

\bib{lenzmann_semirel}{article}{
      author={Lenzmann, E.},
       title={Well-posedness for semi-relativistic {H}artree equations of
  critical type},
        date={2007},
        ISSN={1385-0172,1572-9656},
     journal={Math. Phys. Anal. Geom.},
      volume={10},
      number={1},
       pages={43\ndash 64},
         url={https://doi.org/10.1007/s11040-007-9020-9},
}

\bib{LiebLoss}{book}{
      author={Lieb, E.~H.},
      author={Loss, M.},
       title={Analysis.},
     edition={2nd ed.},
      series={Grad. Stud. Math.},
   publisher={Providence, RI: American Mathematical Society (AMS)},
        date={2001},
      volume={14},
        ISBN={0-8218-2783-9},
}

\bib{LiebYau1}{article}{
      author={Lieb, E.~H.},
      author={Yau, H.-T.},
       title={The {C}handrasekhar theory of stellar collapse as the limit of
  quantum mechanics},
        date={1987},
        ISSN={0010-3616,1432-0916},
     journal={Comm. Math. Phys.},
      volume={112},
      number={1},
       pages={147\ndash 174},
         url={http://projecteuclid.org/euclid.cmp/1104159813},
}

\bib{Loefstroem}{article}{
      author={Loefstroem, J.},
       title={Interpolation of weighted spaces of differentiable functions on
  {{\({R}^ n\)}}.},
        date={1982},
        ISSN={0373-3114},
     journal={Ann. Mat. Pura Appl. (4)},
      volume={132},
       pages={189\ndash 214},
}

\bib{Machihara_et_al}{article}{
      author={Machihara, S.},
      author={Nakanishi, K.},
      author={Ozawa, T.},
       title={Small global solutions and the nonrelativistic limit for the
  nonlinear {Dirac} equation.},
        date={2003},
        ISSN={0213-2230},
     journal={Rev. Mat. Iberoam.},
      volume={19},
      number={1},
       pages={179\ndash 194},
         url={https://eudml.org/doc/39690},
}

\bib{Maligranda}{article}{
      author={Maligranda, L.},
       title={On commutativity of interpolation with intersection},
        date={1985},
        ISSN={1592-9531},
     journal={Suppl. Rend. Circ. Mat. Palermo (2)},
      volume={10},
       pages={113\ndash 118},
         url={https://eudml.org/doc/220493},
}

\bib{MichelangeliSchlein}{article}{
      author={Michelangeli, A.},
      author={Schlein, B.},
       title={Dynamical collapse of boson stars},
        date={2012},
        ISSN={0010-3616,1432-0916},
     journal={Comm. Math. Phys.},
      volume={311},
      number={3},
       pages={645\ndash 687},
         url={https://doi.org/10.1007/s00220-011-1341-7},
}

\bib{NakamuraTsutaya}{article}{
      author={Nakamura, M.},
      author={Tsutaya, K.},
       title={Scattering theory for the {Dirac} equation of {Hartree} type and
  the semirelativistic {Hartree} equation},
        date={2012},
        ISSN={0362-546X},
     journal={Nonlinear Anal., Theory Methods Appl., Ser. A, Theory Methods},
      volume={75},
      number={8},
       pages={3531\ndash 3542},
}

\bib{Pusateri}{article}{
      author={Pusateri, F.},
       title={Modified scattering for the boson star equation},
        date={2014},
        ISSN={0010-3616,1432-0916},
     journal={Comm. Math. Phys.},
      volume={332},
      number={3},
       pages={1203\ndash 1234},
         url={https://doi.org/10.1007/s00220-014-2094-x},
}

\bib{SeegerTrebels}{article}{
      author={Seeger, A.},
      author={Trebels, W.},
       title={Embeddings for spaces of {Lorentz}-{Sobolev} type},
        date={2019},
        ISSN={0025-5831},
     journal={Math. Ann.},
      volume={373},
      number={3-4},
       pages={1017\ndash 1056},
}

\bib{SigalSoffer88}{article}{
      author={Sigal, I.~M.},
      author={Soffer, A.},
       title={Local decay and propagation estimates for time- dependent and
  time-independent hamiltonians},
        date={1988},
     journal={Prepint, Princeton Univ.},
         url={http://www.math.toronto.edu/sigal/publications/SigSofVelBnd.pdf},
}

\bib{SigalSoffer90}{article}{
      author={Sigal, I.~M.},
      author={Soffer, A.},
       title={Long-range many-body scattering. {A}symptotic clustering for
  {C}oulomb-type potentials},
        date={1990},
        ISSN={0020-9910,1432-1297},
     journal={Invent. Math.},
      volume={99},
      number={1},
       pages={115\ndash 143},
         url={https://doi.org/10.1007/BF01234414},
}

\bib{SigalWu25}{article}{
      author={Sigal, I.~M.},
      author={Wu, X.},
       title={On light cone bounds in quantum open systems},
        date={2025},
     journal={arXiv:2503.20635},
         url={https://arxiv.org/abs/2503.20635},
}

\bib{SigalWu}{article}{
      author={Sigal, I.~M.},
      author={Wu, X.},
       title={On propagation of information in quantum mechanics and maximal
  velocity bounds},
        date={2025},
        ISSN={0377-9017,1573-0530},
     journal={Lett. Math. Phys.},
      volume={115},
      number={1},
       pages={Paper No. 17, 25},
         url={https://doi.org/10.1007/s11005-025-01899-z},
}

\bib{Skibsted91}{article}{
      author={Skibsted, E.},
       title={Propagation estimates for {$N$}-body {S}chroedinger operators},
        date={1991},
        ISSN={0010-3616,1432-0916},
     journal={Comm. Math. Phys.},
      volume={142},
      number={1},
       pages={67\ndash 98},
         url={http://projecteuclid.org/euclid.cmp/1104248490},
}

\bib{stein_weiss}{article}{
      author={Stein, E.~M.},
      author={Weiss, G.},
       title={Interpolation of operators with change of measures},
        date={1958},
        ISSN={0002-9947},
     journal={Trans. Am. Math. Soc.},
      volume={87},
       pages={159\ndash 172},
}

\bib{Tao}{book}{
      author={Tao, T.},
       title={Nonlinear dispersive equations. {Local} and global analysis},
      series={CBMS Reg. Conf. Ser. Math.},
   publisher={Providence, RI: American Mathematical Society (AMS)},
        date={2006},
      volume={106},
        ISBN={0-8218-4143-2},
}

\bib{Yang}{article}{
      author={Yang, C.},
       title={Small data scattering of semirelativistic {Hartree} equation},
        date={2019},
        ISSN={0362-546X},
     journal={Nonlinear Anal., Theory Methods Appl., Ser. A, Theory Methods},
      volume={178},
       pages={41\ndash 55},
}

\end{biblist}
\end{bibdiv}
	\doclicenseThis
\end{document}